\documentclass[aps,prx,reprint,superscriptaddress,longbibliography,floatfix]{revtex4-2}
\usepackage[utf8]{inputenc}

\makeatletter
\newcommand{\AddInputPath}[1]{%
    \providecommand*{\input@path}{}
    \g@addto@macro{\input@path}{#1}
}
\AddInputPath{{templates/}}
\makeatother
\usepackage{amsmath,amsthm,amssymb,amsfonts}
\usepackage{mathrsfs}
\usepackage{mathdots}
\usepackage{mathtools, physics, bbm}
\usepackage{booktabs}
\usepackage{listings}

\usepackage[normalem]{ulem}

\usepackage{xcolor}
\definecolor{flaginput}{RGB}{213,230,249}
\definecolor{flaginternal}{RGB}{239,239,239}
\definecolor{flagoutput}{RGB}{215,238,220}

\DeclareMathOperator{\U}{U}

\usepackage{tikz}
\usetikzlibrary{quantikz2}
\usetikzlibrary{arrows.meta,calc,decorations.pathreplacing,positioning,backgrounds}

\definecolor{HCfacezero}{HTML}{DCEAF3}
\definecolor{HCfaceone}{HTML}{F5E5CC}
\definecolor{HCfacetwo}{HTML}{DDEBDD}
\definecolor{HCrhozero}{HTML}{356A9A}
\definecolor{HCrhoone}{HTML}{A45729}
\definecolor{HCrhotwo}{HTML}{6A5687}
\definecolor{HCplaquette}{HTML}{C7922F}
\definecolor{HCink}{HTML}{24272B}
\definecolor{HCmuted}{HTML}{70757A}
\definecolor{HCfaint}{HTML}{C9CDD1}

\tikzset{
  HC/A vertex/.style={circle,draw=HCink,fill=HCrhozero!18,
    line width=.45pt,minimum size=3.4pt,inner sep=0pt},
  HC/B vertex/.style={rectangle,rounded corners=.35pt,draw=HCink,
    fill=HCrhoone!18,line width=.45pt,minimum size=3.5pt,inner sep=0pt},
  HC/A vertex large/.style={HC/A vertex,line width=.55pt,minimum size=5.2pt},
  HC/B vertex large/.style={HC/B vertex,line width=.55pt,minimum size=5.3pt},
  HC/x edge/.style={draw=HCink,line width=.80pt,line cap=round},
  HC/y edge/.style={draw=HCink,line width=.72pt,dash pattern=on 3.0pt off 1.7pt,
    line cap=round},
  HC/z edge/.style={draw=HCink,line width=.78pt,dash pattern=on .45pt off 1.20pt,
    line cap=round},
  HC/x inactive/.style={draw=HCfaint,line width=.38pt,line cap=round},
  HC/y inactive/.style={draw=HCfaint,line width=.38pt,
    dash pattern=on 2.2pt off 1.25pt,line cap=round},
  HC/z inactive/.style={draw=HCfaint,line width=.42pt,
    dash pattern=on .35pt off .95pt,line cap=round},
  HC/active 0/.style={draw=HCrhozero,line width=1.05pt},
  HC/active 1/.style={draw=HCrhoone,line width=1.05pt},
  HC/active 2/.style={draw=HCrhotwo,line width=1.05pt},
  HC/domain/.style={draw=HCmuted,line width=.55pt},
  HC/plaquette halo/.style={draw=HCplaquette,line width=2.55pt,
    line join=round,line cap=round},
  HC/rho badge/.style={circle,draw=HCink,fill=white,line width=.45pt,
    minimum size=11pt,inner sep=0pt,font=\scriptsize\bfseries},
  HC/rho 0/.style={HC/rho badge,draw=HCrhozero,text=HCrhozero},
  HC/rho 1/.style={HC/rho badge,draw=HCrhoone,text=HCrhoone},
  HC/rho 2/.style={HC/rho badge,draw=HCrhotwo,text=HCrhotwo},
  HC/callout/.style={rounded corners=2pt,draw=HCmuted!55,fill=white,
    inner sep=4pt,align=left,font=\scriptsize},
  HC/pipeline/.style={rounded corners=1.5pt,draw=HCmuted!70,fill=white,
    minimum height=8.5mm,inner xsep=3pt,align=center,font=\scriptsize},
  HC/flow/.style={-{Stealth[length=4.2pt,width=4.0pt]},draw=HCmuted,
    line width=.55pt},
  HC/bit/.style={rectangle,draw=HCmuted!70,fill=white,line width=.35pt,
    minimum width=2.7mm,minimum height=4.0mm,inner sep=0pt,font=\tiny},
  HC/bit one/.style={HC/bit,draw=HCrhozero,fill=HCrhozero!20,
    text=HCrhozero,font=\scriptsize\bfseries},
  HC/block 0/.style={rectangle,draw=HCrhozero!75,fill=HCrhozero!11,
    line width=.45pt,minimum width=34mm,minimum height=4.3mm,inner sep=1pt},
  HC/block 1/.style={rectangle,draw=HCrhoone!75,fill=HCrhoone!11,
    line width=.45pt,minimum width=34mm,minimum height=4.3mm,inner sep=1pt},
  HC/block 2/.style={rectangle,draw=HCrhotwo!75,fill=HCrhotwo!11,
    line width=.45pt,minimum width=34mm,minimum height=4.3mm,inner sep=1pt}
}

\newcommand{\HCA}[2]{({1.5*((#1)+(#2))},{0.8660254*(-(#1)+(#2))})}
\newcommand{\HCB}[2]{({1.5*((#1)+(#2))+1},{0.8660254*(-(#1)+(#2))})}

\newcommand{\HCDomainPath}{(-.6,0) -- (3.9,-2.598076) --
  (8.4,0) -- (3.9,2.598076) -- cycle}

\newcommand{\HCDrawMatching}[3]{%
  \begin{scope}[shift={#1}]
    \begin{scope}[x=.36cm,y=.36cm]
      \begin{scope}
        \clip \HCDomainPath;
        \foreach \i in {-1,...,3}{
          \foreach \j in {-1,...,3}{
            \draw[HC/x inactive] \HCA{\i}{\j} -- \HCB{\i-1}{\j};
            \draw[HC/y inactive] \HCA{\i}{\j} -- \HCB{\i}{\j-1};
            \draw[HC/z inactive] \HCA{\i}{\j} -- \HCB{\i}{\j};
          }
        }
        \foreach \i in {-1,...,3}{
          \foreach \j in {-1,...,3}{
            \pgfmathtruncatemacro{\HCc}{mod(mod(\i+2*\j,3)+3,3)}
            \pgfmathtruncatemacro{\HCrhoy}{mod(\HCc+1,3)}
            \pgfmathtruncatemacro{\HCrhoz}{mod(\HCc+2,3)}
            \ifnum\HCc=#2\relax
              \draw[HC/x edge,#3] \HCA{\i}{\j} -- \HCB{\i-1}{\j};
            \fi
            \ifnum\HCrhoy=#2\relax
              \draw[HC/y edge,#3] \HCA{\i}{\j} -- \HCB{\i}{\j-1};
            \fi
            \ifnum\HCrhoz=#2\relax
              \draw[HC/z edge,#3] \HCA{\i}{\j} -- \HCB{\i}{\j};
            \fi
          }
        }
      \end{scope}
      \draw[HC/domain] \HCDomainPath;
      \foreach \i in {0,1,2}{
        \foreach \j in {0,1,2}{
          \node[HC/A vertex] at \HCA{\i}{\j} {};
          \node[HC/B vertex] at \HCB{\i}{\j} {};
        }
      }
      \node[font=\tiny,anchor=north east] at \HCA{0}{0} {$q_1$};
      \node[font=\tiny,anchor=north west] at \HCB{0}{0} {$q_2$};
      \node[HC/rho #2] at (3.9,3.25) {#2};
      \node[font=\scriptsize\bfseries,anchor=west] at (4.45,3.25)
        {$\rho(e)=#2$};
    \end{scope}
  \end{scope}%
}

\tikzset{
  tensor/.style={
    draw,
    minimum width=0.6cm,
    minimum height=0.6cm,
    inner sep=2pt
  }
}

\usetikzlibrary{calc}
\tikzset{
  gatebox/.style={draw, fill=white, minimum size=5.5mm, inner sep=1pt},
  wire/.style={line width=0.5pt, rounded corners=2.2mm},
}

\newcommand{\vket}[1]{\lvert #1 \rangle\!\rangle}   
\newcommand{\vbra}[1]{\langle\!\langle #1 \rvert}   

\graphicspath{{./}{figures/}}

\newcommand{\stk}[2]{\genfrac{}{}{0pt}{0}{#1}{#2}}   

\usepackage{stmaryrd,MnSymbol}
\usepackage[colorlinks = true, citecolor=blue, linkcolor=blue, urlcolor=blue, linktoc=page]{hyperref}
\renewcommand{\>}{\rangle}

\newcommand{\C}{{\mathbb{C}}} 

\newcommand{\cnot}{{\mathtt{CNOT}}}

\newcommand{\cz}{{\mathtt{CZ}}}
\newcommand{\I}{{\mathtt{I}}}

\renewcommand{\U}{{\mathsf{U}}}

\renewcommand{\L}{{\mathsf{L\,}}} 

\newcommand{\rr}{{\rrangle}}

\newcommand{\ins}{{\mathrm{in}}}
\newcommand{\outs}{{\mathrm{out}}}

\newcommand{\id}{{\mathbb{I}}}

\newcommand{\qnoise}[1]{\gate[1,style={dashed,rounded corners,fill=red!10}]{#1}}

\definecolor{darkred}{RGB}{120,0,0}
\definecolor{darkpurple}{RGB}{120,0,120}

\usepackage{cleveref}

\theoremstyle{plain}
\newtheorem{theorem}{Theorem}[section]
\newtheorem{lemma}[theorem]{Lemma}
\newtheorem{proposition}[theorem]{Proposition}
\newtheorem{corollary}[theorem]{Corollary}

\newtheorem{definition}[theorem]{Definition}
\newtheorem{example}[theorem]{Example}
\newtheorem{remark}[theorem]{Remark}

\usepackage{comment}

\begin{document}

\title{Fault tolerance of quantum circuits with tensor networks and symplectic geometry}

\author{Soham Ghosh}
\email{soham.ghosh@tum.de}\email{mails.soham@gmail.com}
\affiliation{Chair of Theoretical Information Technology, Technical University of Munich}

\author{Holger Boche}
\email{boche@tum.de}
\affiliation{Chair of Theoretical Information Technology, Technical University of Munich}
\affiliation{BMFTR Research Hub 6G-Life, Germany}
\affiliation{Munich Centre for Quantum Science and Technology (MCQST)}
\affiliation{Munich Quantum Valley (MQV)}
\affiliation{Cluster of Excellence CeTI, Technical University of Dresden}

\author{Andrew Tanggara}
\email{andrew.tanggara@gmail.com}
\affiliation{Centre for Quantum Technologies, National University of Singapore}
\affiliation{Nanyang Quantum Hub, Nanyang Technological University}

\begin{abstract}
We develop an operator-algebraic framework for analyzing fault tolerance in quantum circuits and establish necessary and sufficient conditions for exact recovery of the ideal computation under a prescribed noise model.
For prescribed circuit families and noise models, we derive a semidefinite program (SDP) based test whose infeasibility certifies that there does not exist any recovery map restoring its intended operation. For non-adaptive stabilizer circuits, analyzing the stabilizer symmetries of its components with symplectic geometry, we obtain closed-form detection and logical-effect matrices
that characterize circuit distance algebraically. Applied to a finite
Hastings--Haah honeycomb Floquet-code circuit, these matrices certify
the dynamical encoding of two logical qubits and distance four under Pauli data errors including measurement errors. Using these matrices, we then derive MacWilliams
identities for circuit weight enumerators and obtain linear-programming (LP) upper
bounds on circuit distance. Output-code constraints give additional
Singleton-like bounds. 
Under independent
and identically distributed (i.i.d.) depolarizing noise, the weight enumerators
determine decoder failure probabilities and yield upper bounds on
circuit pseudothresholds. 
Finally, we derive LP distance bounds for families of flag
syndrome-extraction circuits specified by their CNOT orderings,by
incorporating the $t$-flag criterion as additional linear constraints to the distance-bounding LP.
An analytical CSS Hamming-code benchmark yields a tight distance-three
bound attained by the Chao--Reichardt one-flag construction. 
\end{abstract}

\maketitle

\tableofcontents

\section{Motivation and overview}
\label{sec:introduction}

The presence of noise is the main obstacle to realizing scalable quantum computers
and the speed-ups they promise over classical computers.
Uncontrolled interactions with the environment can destroy quantum
coherence, while imperfect gates and measurement devices introduce errors during the
computation itself \cite{Shor1996FT,gottesman2022opportunities}.
Protecting quantum information requires preserving amplitudes,
relative phases, and entanglement correlations. At the logical level, errors in classical computers are only bit-flips from $0$ to $1$ (and vice versa), while quantum noise can also
change a state continuously, through a small rotation or a gradual
loss of coherence \cite{gottesman1997stabilizer}.

Measurements introduce a further difficulty. Classical error
correction can compare redundant bit values, which can ideally be
read without changing them. Directly measuring an unknown quantum
state, however, can disturb the very information we wish to protect.
How can we then obtain enough information to correct an error
without learning about the protected quantum state?

Quantum error correction resolves part of these difficulties by encoding
logical information using redundant physical qubits
\cite{Shor1995,Steane1996,CalderbankShor1996}. Entangling operations transfer
\emph{error-syndrome} \cite{gottesman1997stabilizer,PrabhuReichardt2023FewerQubits,Reichardt2021SteaneFlags,LiouLai2023SharedFlags} information to additional ancillary qubits, which
are then measured. For a correctable set of errors, these
measurements reveal information about the error without
revealing or disturbing the protected logical information,
thereby enabling recovery \cite{knill1997theory}.

The remaining difficulty is that the devices performing these
operations are themselves imperfect. How do we encode information
reliably when the act of encoding itself can introduce errors? How do we perform logical computation on
encoded states, extract error syndromes, and read out the answer
when the operations used to do them can fail? Faults can accumulate over time, and even a single fault can
propagate through subsequent gates into correlated errors on several qubits, while
an incorrect syndrome measurement can lead to an incorrect recovery. As a result, protecting an encoded
state requires more than choosing a quantum error correcting code. Since near-term physical devices are susceptible to noise, the goal is to design fault-tolerant quantum circuits
that suppress errors while keeping the resource overhead
low, even when the size of the computation increases
\cite{Shor1996FT,AharonovBenOr2008,AliferisGottesmanPreskill2006,KnillLaflammeZurek1998,NguyenPattison2025Overheads}. However, this is a daunting task that requires careful choices about
where and when to introduce redundancy across circuit components and in what sequence we should perform circuit operations to achieve error suppression
within a reasonable resource budget in additional qubits, operations, and execution time.

The central theoretical task is to quantify how fault-tolerant a quantum circuit
is under a specified noise model and to use that information
to guide its design. This raises several related but distinct
questions. Given some noise model, does there exist a recovery operation that restores the
circuit's intended operation, including its classical measurement record?
What is the smallest number of spacetime faults that are undetected by any decoder while changing the logical action of the circuit? For a specified noise
model and decoder, what is the probability that recovery fails?
And, for a finite circuit, how large can the physical error rate be
before the encoded circuit implementation ceases to outperform an
unencoded reference? These questions are captured by different
metrics. \emph{Circuit distance} is a worst-case measure which quantifies the smallest
harmful fault that is not detectable by any decoder, whereas \emph{logical failure probability} depends on the
noise, the decoder, and the number of fault patterns that cause
failure. The finite-circuit \emph{pseudo-threshold}, on the other hand, compares logical failure
probability with a specified unencoded benchmark
\cite{knill1997theory,BenyOreshkov2011Simulation,
beverland2024faulttolerancestabilizerchannels,
ChamberlandBeverland2018Flags,CrossDiVincenzoTerhal2009Pseudothreshold}. 
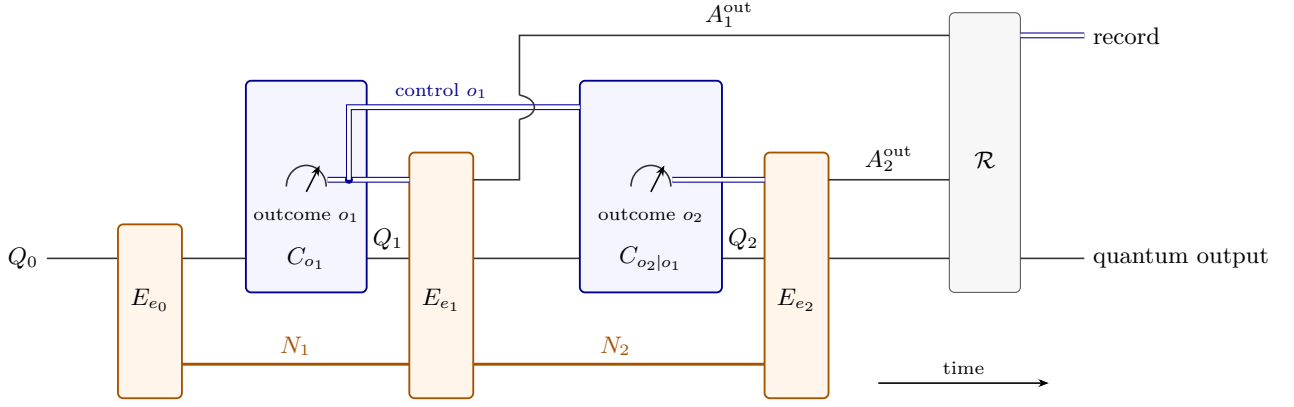
\begin{figure*}[t]
\centering
\begin{tikzpicture}[
  x=.94cm,y=1cm,font=\small,
  qwire/.style={draw=black!80,line width=.6pt},
  cwire/.style={draw=blue!55!black,double,double distance=1.4pt,
    line width=.4pt},
  memory/.style={draw=orange!65!black,line width=1.1pt},
  circuit/.style={draw=blue!55!black,fill=blue!4,
    rounded corners=2pt,line width=.7pt},
  noise/.style={draw=orange!65!black,fill=orange!8,
    rounded corners=2pt,line width=.7pt},
  flow/.style={-{Stealth[length=4pt]},line width=.6pt}
]
\draw[qwire] (0,0) -- (12.7,0);
\draw[noise] (1,-1.85) rectangle (1.9,.45);
\draw[circuit] (2.8,-.45) rectangle (4.5,2.35);
\draw[noise] (5.1,-1.85) rectangle (6,1.4);
\draw[circuit] (7.5,-.45) rectangle (9.5,2.35);
\draw[noise] (10.1,-1.85) rectangle (11,1.4);
\draw[draw=black!65,fill=black!3,rounded corners=2pt]
  (12.7,-.45) rectangle (13.7,3.25);
\node[anchor=east] at (0,0) {$Q_0$};
\node at (1.45,-.55) {$E_{e_0}$};
\node at (3.65,0) {$C_{o_1}$};
\node at (5.55,-.55) {$E_{e_1}$};
\node at (8.5,0) {$C_{o_2\mid o_1}$};
\node at (10.55,-.55) {$E_{e_2}$};
\node at (13.2,1.3) {$\mathcal R$};
\node[above] at (4.8,0) {$Q_1$};
\node[above] at (9.8,0) {$Q_2$};
\foreach \x in {3.65,8.5}{
  \draw[qwire] (\x-.28,.95) arc[start angle=180,end angle=0,radius=.28];
  \draw[flow] (\x,.88) -- (\x+.19,1.23);
}
\node[font=\scriptsize] at (3.65,.57) {outcome $o_1$};
\node[font=\scriptsize] at (8.5,.57) {outcome $o_2$};
\draw[cwire] (3.95,1.04) -- (5.1,1.04);
\fill[blue!55!black] (4.25,1.04) circle (1.5pt);
\draw[cwire] (4.25,1.04) -- (4.25,2) -- (7.5,2);
\node[above,font=\scriptsize,text=blue!55!black]
  at (5.55,2) {control $o_1$};
\draw[cwire] (8.8,1.04) -- (10.1,1.04);
\draw[qwire] (6,1.04) -- (6.65,1.04) -- (6.65,1.84);
\draw[qwire] (6.65,2.16) -- (6.65,2.95) -- (12.7,2.95);
\draw[qwire] (6.65,1.84)
  .. controls (6.93,1.89) and (6.93,2.11) .. (6.65,2.16);
\node[above] at (9.6,2.95) {$A_1^{\mathrm{out}}$};
\draw[qwire] (11,1.04) -- (12.7,1.04);
\node[above] at (11.85,1.04) {$A_2^{\mathrm{out}}$};
\draw[memory] (1.9,-1.4) -- (5.1,-1.4);
\draw[memory] (6,-1.4) -- (10.1,-1.4);
\node[above,text=orange!65!black] at (3.5,-1.4) {$N_1$};
\node[above,text=orange!65!black] at (8,-1.4) {$N_2$};
\draw[cwire] (13.7,2.95) -- (14.6,2.95);
\node[anchor=west] at (14.6,2.95) {record};
\draw[qwire] (13.7,0) -- (14.6,0);
\node[anchor=west] at (14.6,0) {quantum output};
\draw[flow] (11.7,-1.65) -- (14.1,-1.65);
\node[above,font=\scriptsize] at (12.9,-1.65) {time};
\end{tikzpicture}
\caption{An adaptive circuit interacting with temporally correlated noise, shown for
two rounds. The circuit implements the quantum instruments $\{C_{o_1}\}_{o_1}$ and $\{C_{o_2\mid o_1}\}_{o_2}$ where the instrument at round $2$ depends on the observed outcome $o_1$. Double blue lines
carry classical information. A copy of $o_1$ selects the second
instrument before noise acts on its retained record. The orange maps
$E_{e_r}$ act on the displayed data, record, and memory ports;
$N_1,N_2$ carry information between noise interactions. Summing over the circuit outcomes and noise Kraus operator indices gives the
physical channel. This representation is a generalization of the strategic code scenario
\cite{tanggara2024strategiccodeunifiedspatiotemporal}, where the goal is only to protect an encoded initial state. The recovery is assumed to be noiseless. }
\label{fig:adaptive-circuit-noise}
\end{figure*}

This opens up the need for one unifying framework within which all these metrics can be studied. To achieve this goal, we take a direct approach by analyzing the linear operators that define a quantum circuit and its noise model. As illustrated in
\Cref{fig:adaptive-circuit-noise}, the ideal circuit consists of
quantum instruments \cite{chiribella2009theoretical} whose branch operators may depend on earlier
outcomes, such as $C_{o_1}$ and $C_{o_2\mid o_1}$ (which depends on observed outcome $o_1$). Similarly, the noise model is a
causal process with Kraus branches $E_{e_r}$ and memory systems
$N_r$ that carry correlations across time. The final recovery $\mathcal R$
acts on the noisy measurement record and quantum output. We use the \emph{quantum-comb} formalism to represent circuit components
and noise processes as causally ordered modules that can be composed while
retaining their spatio-temporal structure
\cite{chiribella2009theoretical}. Our framework supports analysis at different
levels of abstraction, namely, \Cref{fig:adaptive-circuit-noise} depicts a
coarse-grained circuit-noise model, with noise acting on the exposed ports of complete
circuit stages, whereas \Cref{fig: component-example-circuit-noise}
depicts a fine-grained model, with noise assigned to individual
state preparations, idle qubits, gates, and measurements.

For a fixed circuit and noise model, we prove necessary and
sufficient conditions (see \Cref{Thm: Iff Fault tolerance}) for the existence of a recovery that restores the intended
quantum operation together with its measurement record. The conditions are expressed directly in terms of the circuit and noise Kraus
operators, while the proof is constructive and yields an explicit construction of a recovery linear map. An intended quantum operation is a linear map which in general does not have a unique circuit representation. Circuits that realize the same
linear map can have different fault-tolerance properties \cite{RodatzEtAl2025FaultToleranceByConstruction}. For a family of circuits implementing the same linear map, 
the circuit and its corresponding recovery operation vary together. We formulate an SDP relaxation
of this joint circuit-recovery search problem (see \Cref{prop:joint-ft-sdp}). Every circuit--recovery pair gives a feasible point, so infeasibility rules out all
target-implementing circuits in the prescribed family. Feasibility alone does not establish
that such a pair exists. The test is most useful for large or
continuously parameterized families with informative shared constraints
and manageable SDP size, while it need not be more efficient than testing each circuit separately for a small finite family.

We next specialize to non-adaptive stabilizer circuits. Let the index $\zeta\in\{0,1\}$ denote the level of abstraction in the circuit-noise model, namely $\zeta=0$ (resp. $\zeta=1$) denotes fine-grained (resp. coarse-grained) circuit-noise model. We use symplectic geometry to analyze the stabilizer symmetries of the
circuit components and derive closed-form expressions for the
\emph{detection matrix} $H_\zeta$ and the \emph{logical-effect matrix}
$L_\zeta$ \cite{beverland2024faulttolerancestabilizerchannels}. Any spacetime Pauli fault can be represented as binary symplectic vectors \cite{beverland2024faulttolerancestabilizerchannels,Derks2025designingfault,rengaswamy2018synthesis}. For a given spacetime fault label $e\in\mathcal E^{(\zeta)}$, $H_\zeta  e^{\mathsf T}$ is its detection
syndrome, and $L_\zeta  e^{\mathsf T}$ characterizes its effect on the logical information processed by the circuit. The circuit distance $d_{\mathrm{circ},\zeta}$ is therefore the
minimum weight (see \Cref{def:spacetime Pauli weight}) of a fault vector $ e\in\mathcal E^{(\zeta)}$ having \emph{trivial syndrome} ($H_\zeta  e^{\mathsf T} = 0$) and \emph{non-trivial logical effect} ($L_\zeta  e^{\mathsf T} \neq 0$). Furthermore, we show (see \Cref{thm:ft-HL-criterion}) that under stochastic Pauli noise, exact recovery is possible if and only if
faults of positive probability with the same detection syndrome have the same logical effect. The matrices thus characterize circuit distance and correctability algebraically,
although computing distance can still remain difficult
\cite{ChamberlandBeverland2018Flags,
beverland2024faulttolerancestabilizerchannels,
WebsterJacobHiggott2026Distance,KapshikarKundu2023DistanceHardness}.  As an illustration (see \Cref{app:honeycomb}), we analyze a seven-round, 18-qubit instance of the
Hastings--Haah honeycomb Floquet code
\cite{hastings2021dynamically}. Our analysis certifies that the schedule
dynamically encodes two logical qubits and has circuit distance four
under the specified internal data- and record-Pauli fault model, with
no noise on the data input-output boundaries.

We then shift our attention from the detection and logical-effect signatures of an individual fault pattern to how fault signatures arise based on the distribution of faults over circuit spacetime locations. Fix an abstraction
level $\zeta$, and choose a partition $\Pi=\{I_b\}_{b=1}^B$ of the permitted spacetime circuit fault locations $\Lambda^{(\zeta)} = \bigsqcup_{b=1}^{B}I_{b}$ into blocks $I_b$, where each block contains $N_b$ circuit locations (see \Cref{def:binary-fault-space} and \Cref{def:we-circuit-family}). 
For example, the blocks could group locations by circuit round (\emph{time}) or by
circuit region (\emph{space}). For a given fault vector $e \in \mathcal{E}^{\zeta}$, its \emph{weight profile} (see \Cref{def:weight-profile}) $\mathbf w =(w_1,\ldots,w_B)$ is a tuple where $w_b$ counts the total number of non-identity Pauli faults in block $I_b$, where the total weight of the fault is recovered by $|\mathbf w| = \sum_b w_b$. For each profile and signature, define
\[
\resizebox{0.98\columnwidth}{!}{$\displaystyle
A_{\mathbf w}^{(\zeta)}(\sigma,\lambda)
\triangleq
\bigl|\{\mathbf e\in\mathcal E^{(\zeta)}:
\operatorname{wt}_{\Pi}(\mathbf e)=\mathbf w,\;
H_\zeta\mathbf e^{\mathsf T}=\sigma,\;
L_\zeta\mathbf e^{\mathsf T}=\lambda\}\bigr|
$}.
\]
This is the number of fault patterns with syndrome $\sigma$, logical effect $\lambda$ and weight profile $\mathbf w$. All these counts can be packaged into one generating function
\begin{equation}
\label{eq: Gen function}
W_{\sigma,\lambda}^{(\zeta),\Pi}(\mathbf u,\mathbf v)
\triangleq
\sum_{\mathbf w} A_{\mathbf w}^{(\zeta)}(\sigma,\lambda)
\prod_{b=1}^{B}u_b^{N_b-w_b}v_b^{w_b},
\end{equation}
where $\mathbf u \triangleq (u_b)_{b \in \mathbb Z^+_B}$ and $\mathbf v \triangleq (v_b)_{b \in \mathbb Z^+_B}$ are formal variables. From this function, we construct non-negative integer-valued variables $U,T,\hat U, \hat T$ called \emph{weight enumerators} (see \Cref{def:Weight-enumerators-circuit}) and derive MacWilliams identities and additional constraints that every circuit in a prescribed family must satisfy. Relaxing these variables to real numbers, we obtain linear-programming (LP) distance bounds for families of non-adaptive stabilizer circuits. Similar to our SDP, infeasibility rules out every circuit in a prescribed family from achieving target distance $D$, while feasibility doesn't necessarily give a construction. Explicit knowledge about the circuit family can lead to construction of additional non-negative integer valued variables satisfying some constraints which can give stronger bounds. To illustrate this, we consider families of $t$-flag \cite{ChamberlandBeverland2018Flags,PatoEtAl2024Flag,TansuwannontEtAl2020CyclicFlags} CSS code syndrome extraction circuits prescribed by their CNOT orderings (see \Cref{subsec:flag-syndrome-extraction}). We then incorporate the $t$-flag criteria as constraints on an additional variable $C$ (see \Cref{eq:flp-resolved-counts}) and obtain LP distance bounds . We benchmark the resulting bound analytically on CSS Hamming codes and show infeasibility at distance $D\geq 4$. Furthermore, we verify that the $1$-flag construction of Chao and Reichardt \cite{ChaoReichardt2018TwoQubits} provides a feasible point at distance $D=3$.

The weight enumerators retain information about fault-tolerance that circuit distance alone cannot capture.
To see this, consider independent and identically distributed single-qubit depolarizing noise. Let $p$ be the
probability of a non-identity Pauli fault at each circuit location, with $X$,
$Y$, and $Z$ each occurring with probability $p/3$. If $d_{\mathrm{circ},\zeta}$ is the
circuit distance, the total number of weight-$d_{\mathrm{circ},\zeta}$ faults with zero syndrome
and nonzero logical effect is
\[
a_{d_{\mathrm{circ},\zeta}} \triangleq \sum_{|\mathbf w|=d_{\mathrm{circ},\zeta}}\ \sum_{\lambda\ne 0}
A_{\mathbf w}^{(\zeta)}(0,\lambda).
\]
For small $p$, we then get the probability of undetectable logical failure up to leading order in $p$ as
\[
P_{\mathrm{ud}}(p)
\triangleq
a_{d_{\mathrm{circ},\zeta}}\left(\frac{p}{3}\right)^{d_{\mathrm{circ},\zeta}}+O(p^{d_{\mathrm{circ},\zeta}+1}).
\]
Hence, two circuits $C,C'$ with the same distance can have different
low-noise undetected-failure probabilities if their values of $a_{d_{\mathrm{circ},\zeta}},a'_{d_{\mathrm{circ},\zeta}}$
differ. Given a deterministic decoder $\delta$ (see \Cref{def:iid-decoder}), the generating function in \Cref{eq: Gen function} also enable us to calculate the exact failure probability $P_{\mathrm{fail},\delta}(p)$ of it failing to correct an error. To assess practically motivated error suppression for quantum circuits, we use the single-qubit physical error rate $p$ as a reference benchmark and require,
for a prescribed integer $j\geq0$, that
\begin{equation}
    P_{\mathrm{fail},\delta}(p) \leq 10^{-j}p. \label{eq:Intro-pseudo}
\end{equation}
For a given circuit and decoder $\delta$, and fixed $j\geq0$
and $0<p_{\max}\leq1$, the pseudo-threshold
\cite{ChamberlandBeverland2018Flags} relative to the target
$10^{-j}p$ is the supremum of all $x\in[0,p_{\max}]$ such that
\Cref{eq:Intro-pseudo} holds for every $0<p<x$
(see \Cref{def:iid-pseudothreshold}). Finally, incorporating an additional constraint to the the distance linear program we obtain bounds on pseudo-threshold for any given family of non-adaptive stabilizer circuits (see \Cref{prop:iid-performance-lp}).

\subsection{Paper organization and technical overview}

After introducing notation and preliminaries in
\Cref{sec:preliminaries}, we define the operator-algebraic circuit and noise
representations and develop the composition of their stabilizer
relations in \Cref{sec:general-quantum-circuits}.
The main results and applications are organized as follows:

\begin{itemize}
    \item \textbf{Fault tolerance of general quantum circuits.}
    In \Cref{sec: Fault tolerance}, we establish conditions
    for recovering the complete ideal circuit operation,
    including its quantum output and measurement record.
    Let $\mathcal N_0$ denote the channel implemented by
    the noiseless circuit and $\mathcal N$ its noisy version.
    \Cref{Thm: Iff Fault tolerance} gives necessary and
    sufficient Kraus-operator conditions for the existence
    of a completely positive, trace-preserving recovery
    $\mathcal R$ satisfying
    \begin{equation}
        \mathcal R\circ\mathcal N=\mathcal N_0.
        \label{eq:overview-recovery}
    \end{equation}
    Our result extends the exact recoverability condition
    for circuits implementing quantum memories
    \cite{tanggara2024strategiccodeunifiedspatiotemporal}
    to arbitrary quantum circuits.

    \textbf{Application.}
    Different circuits implementing the same quantum channel
    may have different fault-tolerance properties under a
    given noise model
    \cite{RodatzEtAl2025FaultToleranceByConstruction}.
    For a prescribed family implementing the same target
    channel, we therefore ask whether some circuit admits
    a recovery that restores its ideal operation under the
    specified noise model.
    We formulate an SDP relaxation of this joint
    circuit--recovery search problem
    (\Cref{prop:joint-ft-sdp}).
    Every valid circuit--recovery pair yields a feasible
    point, so infeasibility rules out every such pair
    within the family. Feasibility alone does not establish
    that such a pair exists.

    \item \textbf{Fault tolerance of non-adaptive stabilizer circuits.}
    In \Cref{sec: Fault tolerance of stabilizer circuits},
    we specialize to non-adaptive stabilizer circuits.
    Each nonzero circiut Kraus operator maps an input stabilizer
    code space to an output stabilizer code space; while the
    circuit may also measure logical Pauli operators
    of the input code
    (\Cref{Cor: Logical dim Pres,prop:common-relation-extraction}).
    We compose the Pauli symmetries of individual circuit
    components to obtain symmetries of the entire circuit
    valid for every possible measurement-outcome sequence
    (\Cref{prop: stabilizers of composition,prop:outcome-independent-stabilizer-relations}).
    These symmetries are described by a binary linear
    space, which we call the \emph{common relation space}
    $\mathcal W_{\mathbf C}$
    (\Cref{def:common-stabilizer-relations}).

    These relations serve four operational roles.
    Two provide ways to detect faults.
    First, certain parities of the circuit's measurement
    outcomes have fixed values in the noiseless setting.
    These are called \emph{detectors}
    \cite{Derks2025designingfault}.
    Second, noiseless measurements of output-code
    stabilizers, combined with the circuit's outcomes,
    give further parities whose values are fixed in the
    noiseless setting.
    A violation of either kind of parity detects a fault
    (\Cref{prop:common-relation-extraction,tab:ft-operational-signatures}).
    These violations are captured by the
    \emph{detection matrix} $H_\zeta$
    \cite{kliuchnikov2023stabilizercircuitverification,beverland2024faulttolerancestabilizerchannels},
    defined in \Cref{def:ft-signature-matrices}.
    Here $\zeta=0$ and $\zeta=1$ specify the fine-grained
    and coarse-grained circuit noise models, respectively.
    For a spacetime Pauli fault vector $e$
    (\Cref{def:binary-fault-space}), the vector
    $\sigma\triangleq H_\zeta e^{\mathsf T}$ is its
    \emph{detection syndrome}.
    A zero entry means that the corresponding detection
    relation is preserved; a nonzero entry records its
    violation.
    Thus $H_\zeta e^{\mathsf T}=0$ precisely when the
    fault is undetected by these checks.
    We derive closed-form expressions for these signatures
    at both levels in
    \Cref{prop:ft-level0-signature,prop:ft-level1-signature}.

    The other two roles describe the circuit's intended
    logical operation.
    Some relations identify outcome parities that report
    the outcomes of logical Pauli measurements on the input
    code. Others connect input and output logical Pauli
    operators, specifying how the remaining encoded quantum
    information is carried through the circuit
    (\Cref{prop:common-relation-extraction,prop:outcome-independent-stabilizer-relations}).
    Violations of the chosen logical relations are captured
    by the \emph{logical-effect matrix} $L_\zeta$
    \cite{kliuchnikov2023stabilizercircuitverification,beverland2024faulttolerancestabilizerchannels}.
    The vector $\lambda\triangleq L_\zeta e^{\mathsf T}$ is the
    fault's \emph{logical-effect signature}
    (\Cref{def:ft-signature-matrices}).
    Its entries indicate which of these relations are
    violated, with closed-form expressions obtained from
    \Cref{prop:ft-level0-signature,prop:ft-level1-signature}. For a zero-syndrome fault, $\lambda\ne0$ identifies a harmful undetected logical error.
    For a detected fault, $\lambda$ records violations relative to the chosen logical representatives; its value must be interpreted together with the syndrome and a correction convention.
    If both signatures vanish, the ideal circuit channel
    is unchanged
    (\Cref{rem:ft-operational-signatures,prop:ft-matrix-indicators,tab:ft-matrix-meaning}).
    More generally, among faults with the same detection
    syndrome, the logical-effect signature distinguishes
    their effects on the circuit channel
    (\Cref{lem:ft-signature-completeness}).

    Under stochastic Pauli noise
    (\Cref{def:stochastic-pauli-noise}), we show that exact
    recovery is possible precisely when positive-probability
    faults with the same detection syndrome have the same
    logical-effect signature
    (\Cref{thm:ft-HL-criterion}).
    Circuit distance is the smallest number of faulty
    spacetime locations that can produce an undetected
    nontrivial logical effect
    (\Cref{def:ft-circuit-distance}):
    \begin{equation}
        \begin{aligned}
        d_{\mathrm{circ},\zeta}
        =\min\bigl\{\operatorname{wt}_{\mathrm{ST}}(e):{}
        &e\in\mathcal V^{(\zeta)},\\
        &H_\zeta e^{\mathsf T}=0,\;
        L_\zeta e^{\mathsf T}\ne0\bigr\}.
        \end{aligned}
        \label{eq:overview-distance}
    \end{equation}
    Here $\mathcal V^{(\zeta)}\subseteq\mathcal E^{(\zeta)}$
    is the allowed fault subspace, and
    $\operatorname{wt}_{\mathrm{ST}}(e)$ counts the faulty
    spacetime locations in $e$
    (\Cref{def:binary-fault-space,def:spacetime-pauli-weight}).

    \textbf{Application.}
    In \Cref{app:honeycomb}, we analyze a seven-round,
    18-qubit instance of the Hastings--Haah honeycomb
    Floquet code \cite{hastings2021dynamically}.
    Our analysis certifies that the measurement schedule
    dynamically encodes two logical qubits and has circuit
    distance four under the specified internal data- and
    record-Pauli fault model, with ideal data boundaries
    (\Cref{prop:honeycomb-relations,app:honeycomb-circuit-distance}).

    \item \textbf{Weight enumerators and linear programming bounds
    for non-adaptive stabilizer circuits.} In \Cref{sec:weight-enumerators}, we use signature-resolved circuit weight enumerators to formulate bounds across prescribed circuit families.
    For a partition $\Pi$ of circuit spacetime locations,
    we construct a multivariate polynomial
    $W_{\sigma,\lambda}^{(\zeta),\Pi}(\mathbf u,\mathbf v)$
    that counts spacetime faults by weight profile,
    detection syndrome $\sigma$, and logical effect
    $\lambda$ (\Cref{def:Weight-enumerators-circuit}).
    We count zero-syndrome faults and harmless
    zero-syndrome faults, and constrain these counts
    using the MacWilliams identities we derive
    (\Cref{prop:we-signature-macwilliams}) and
    boundary-code constraints
    (\Cref{prop:lp-input-boundary}).
    In \Cref{sec:linear-programming-bounds}, these
    constraints yield an LP whose infeasibility rules
    out circuit distance at least $D$ for an entire
    prescribed family of non-adaptive stabilizer circuits
    (\Cref{thm:lp-obstruction}).
    Feasibility does not necessarily give a construction.
    Under the stated boundary-fault assumptions, the
    boundary constraints also yield Singleton-like bounds
    (\Cref{cor:boundary-singleton}).

For a decoder that applies no correction at zero detection syndrome (\Cref{def:iid-decoder}), we consider i.i.d.\ single-qubit depolarizing noise with noiseless output checks and recovery.
    We derive an additional constraint, which we append
    to the distance-bounding LP to obtain upper bounds
    on circuit \emph{pseudo-thresholds}
    (\Cref{def:iid-pseudothreshold,prop:iid-performance-lp}).

    \textbf{Application.}
    In \Cref{subsec:flag-syndrome-extraction}, we derive
    LP distance bounds for flag-based syndrome-extraction
    circuits prescribed by their CNOT orderings
    \cite{ChaoReichardt2018FewQubits,PatoEtAl2024Flag,ChamberlandBeverland2018Flags}.
    We express the $t$-flag criterion as linear
    constraints on operation-fault counts and incorporate
    them into the LP
    (\Cref{prop:flag-forbidden-sectors,thm:flp-bound}).
    We benchmark the resulting bound analytically on
    CSS Hamming codes, including data boundary faults,
    and show infeasibility for $D\geq4$
    (\Cref{prop:ham-boundary-bound}).
    Furthermore, we verify that the $1$-flag construction
    of Chao and Reichardt
    \cite{ChaoReichardt2018TwoQubits}
    attains distance $D=3$ and supplies a feasible point,
    establishing tightness
    (\Cref{thm:ham-tight-benchmark}).
\end{itemize}

We discuss possibilities of further developments and applications in
\Cref{sec:discussion-outlook}.
The causal normalization constraints for the circuit
and noise representations are derived in
\Cref{app:causal-normalisation}.
\subsection{Related prior work}

Our circuit and noise representations build on the
quantum-comb formalism for causal quantum networks
\cite{chiribella2009theoretical,
GutoskiWatrous2007Strategies}.
Strategic codes apply this formalism to adaptive quantum
error correction under spatially and temporally correlated
noise
\cite{tanggara2024strategiccodeunifiedspatiotemporal}.
In this work, the recovery goal is the complete ideal circuit
operation, including its quantum output and measurement
record. Our exact-recovery criterion is a constructive
Kraus-operator formulation of the channel-recovery duality
of B\'eny and Oreshkov
\cite{BenyOreshkov2010,BenyOreshkov2011Simulation},
with the Knill--Laflamme condition as the encoded-storage
case \cite{knill1997theory}.
SDP optimization of recovery is also established
\cite{Fletcher2007,
Kosut2008} and joint encoder--decoder optimization, product lifting, and partial-transpose constraints are also treated by constrained bilinear SDP hierarchies \cite{BertaEtAl2022ConstrainedBilinear}.
Our synthesis relaxation applies these methods to circuit-comb variables, prescribed circuit families, and exact recovery of a full quantum--classical target channel.

For stabilizer circuits, outcome relations, logical action,
and fault distance have been studied through spacetime codes,
circuit verification, stabilizer channels, and detector error
models
\cite{DelfossePaetznick2023,
kliuchnikov2023stabilizercircuitverification,
beverland2024faulttolerancestabilizerchannels,
Derks2025designingfault,
Bacon2017,
Gidney2021Stim,
BlackwellHaah2025,
WebsterJacobHiggott2026Distance}.
In particular, the characterization of distance through
check and logical-effect matrices is already present in
\cite{beverland2024faulttolerancestabilizerchannels}.
We derive explicit expressions for these matrices from the
Pauli symmetries of circuit components, using a common
relation space that also describes detectors, measured
observables, output checks, and preserved logical operators.
Related approaches to composition and circuit design include
operator matching in tensor-network codes
\cite{cao2022quantum,
shen2025quantum}, circuit transformations that preserve
fault equivalence
\cite{RodatzEtAl2025FaultToleranceByConstruction}, and
composable detector error models
\cite{kliuchnikov2026composing}.

Our enumerator bounds build on quantum weight enumerators,
MacWilliams identities, and their use in linear programming
\cite{ShorLaflamme1997MacWilliams,Rains1998QuantumWeightEnumerators,
MacWilliams1963Weights,
Delsarte1973AssociationSchemes,
GluesingLuerssen2015Fourier,
LaiAshikhmin2018Split,
AshikhminLitsyn1999UpperBounds}.
Tensor and circuit enumerators extend these methods to
composite codes and noisy circuits
\cite{CaoLackey2024TensorEnumerators,
KuklianskyLackey2024CircuitEnumerators}, while
channel-dependent enumerators give bounds for codes under
specified noise channels \cite{OuyangLai2022ChannelLP}.
Here, we resolve fault counts by spacetime weight profile,
detection syndrome, and logical effect, and combine their MacWilliams identities with circuit and boundary-code
constraints. This yields LP bounds across prescribed circuit
families, including performance bounds under i.i.d.
depolarizing noise and the stated decoder assumptions.

Finally, our flag-circuit application builds on established
flag constructions and optimization methods
\cite{ChaoReichardt2018TwoQubits,
ChamberlandBeverland2018Flags,PatoEtAl2024Flag,
chao2020flag,
PrabhuReichardt2023FewerQubits,
AnkerMarvian2024ClassicalFlags,
CrigerEtAl2026ILP}.
We express the $t$-flag condition as additional linear
constraints on fault counts for families of CNOT orderings,
and use the Chao--Reichardt one-flag construction as an
analytical CSS Hamming-code benchmark.


\section{Notation and preliminaries}
\label{sec:preliminaries}

All Hilbert spaces are finite dimensional and carry fixed
orthonormal bases, taken to be computational bases for qubits.
We write $\L(\mathcal H_0,\mathcal H_1)$ for the space of
linear operators from $\mathcal H_0$ to $\mathcal H_1$, and
$\L(\mathcal H)\triangleq\L(\mathcal H,\mathcal H)$.
The symbols ${}^{\mathsf T}$, ${}^{\star}$ and ${}^{\dagger}$
denote transpose, entry-wise complex conjugation and adjoint,
respectively, in the fixed bases. For a nonnegative integer
$\ell$, let $\mathbb Z^+_{\ell}\triangleq\{1,\ldots,\ell\}$ and $[\ell] \triangleq \{0,\cdots,\ell\}$. 

Binary vectors are row vectors unless explicitly displayed
as columns. Binary label arithmetic, ranks and kernels are
over $\mathbb F_2$. For a binary matrix
$B\in\mathbb F_2^{m\times n}$, we use the left-kernel convention
\[
    \ker B\triangleq\{u\in\mathbb F_2^m:uB=0\}.
\]
Consequently, for a row vector $v\in\mathbb F_2^n$,
\[
    Bv^{\mathsf T}=0
    \quad\Longleftrightarrow\quad
    vB^{\mathsf T}=0
    \quad\Longleftrightarrow\quad
    v\in\ker B^{\mathsf T}.
\]
This convention applies to binary matrices. For a linear
operator $C$ between Hilbert spaces, $\ker C$ denotes the
subspace of kets $|\psi\rangle$ satisfying $C|\psi\rangle=0$.
For $V\subseteq\mathbb F_2^n$, its ordinary binary dual is
\[
    V^\perp\triangleq
    \{u\in\mathbb F_2^n:uv^{\mathsf T}=0
      \text{ for every }v\in V\}.
\]
We distinguish this dual from the symplectic dual defined below.

\subsection{Pauli and Clifford operators}

\begin{definition}[Pauli operators]
\label{def: symplectic Pauli definition}
Let $X\ket{j}=\ket{j+1}$ and $Z\ket{j}=(-1)^j\ket{j}$
for $j\in\mathbb F_2$, with addition modulo two, and let
$Y=\iota XZ$, where $\iota$ is the imaginary unit.
For $u=[x,z]\in\mathbb F_2^{2n}$, define
\begin{equation}
    \begin{aligned}
        D(u)&\triangleq\bigotimes_{j=1}^n X^{x_j}Z^{z_j},\\
        P(u)&\triangleq\iota^{\tau(u)}D(u),
        \qquad \tau(u)\triangleq\sum_{j=1}^n x_jz_j.
    \end{aligned}
\end{equation}
The sum defining $\tau$ is evaluated over the integers.
Thus the pairs $(0,0),(1,0),(0,1),(1,1)$ represent
$I,X,Z,Y$, respectively, and $P(u)^\dagger=P(u)$ and
$P(u)^2=I$.
The Heisenberg--Weyl group and its Hermitian subset are
\[
    \begin{aligned}
        \mathrm{HW}(\mathbb C^{2^n})
        &\triangleq\{\iota^aP(u):a\in\mathbb Z_4,
                      \ u\in\mathbb F_2^{2n}\},\\
        \mathcal P(\mathbb C^{2^n})
        &\triangleq\{(-1)^aP(u):a\in\mathbb F_2,
                      \ u\in\mathbb F_2^{2n}\}.
    \end{aligned}
\]
We write $\overline{\mathcal P}(\mathbb C^{2^n})$ for the
set of Hermitian Paulis modulo sign, represented by the
canonical operators $P(u)$. The Hermitian subset
$\mathcal P(\mathbb C^{2^n})$ is not a group, since the product
of two anticommuting Hermitian Paulis is not Hermitian.
\end{definition}

For $u=[x,z],v=[x',z']\in\mathbb F_2^{2n}$, define
the symplectic inner product by
\begin{equation}
    \begin{aligned}
        \langle u,v\rangle_s
        &\triangleq u\Omega_n v^{\mathsf T}
         =x(z')^{\mathsf T}+z(x')^{\mathsf T},\\
        \Omega_n&\triangleq
            \begin{bmatrix}0&I_n\\I_n&0\end{bmatrix},
    \end{aligned}
\end{equation}
where $\Omega_n$ denotes the \emph{symplectic form} \cite{rengaswamy2018synthesis}. The symplectic form is nondegenerate because $\Omega_n$ is invertible.
For a subspace $V\subseteq\mathbb F_2^{2n}$, its
symplectic dual is
\[
    V^{\perp_s}\triangleq
    \{u\in\mathbb F_2^{2n}:\langle u,v\rangle_s=0
      \text{ for every }v\in V\}.
\]
Since multiplication by $\Omega_n$ is invertible,
$\dim V^{\perp_s}=2n-\dim V$.
The subspace $V$ is called \emph{isotropic} if its vectors are pairwise
symplectically orthogonal, equivalently
$V\subseteq V^{\perp_s}$.

\begin{lemma}[Symplectic Pauli algebra]
\label{prop: Symplectic Pauli Algebra}
For $u=[x,z],v=[x',z']\in\mathbb F_2^{2n}$,
\begin{equation}
    \begin{aligned}
        P(u)P(v)&=(-1)^{\langle u,v\rangle_s}P(v)P(u),\\
        P(u)P(v)&=\iota^{\kappa(u,v)}P(u+v),
    \end{aligned}
\end{equation}
where
\begin{equation}
    \begin{aligned}
        \kappa(u,v)\triangleq{}&
        \tau(u)+\tau(v)+2\sum_jz_jx'_j -\tau(u+v)\pmod4.
    \end{aligned}
    \label{eq:kappa}
\end{equation}
Here $u+v$ is evaluated modulo two before applying $\tau$.
Moreover, $\kappa(u,v)\bmod2=\langle u,v\rangle_s$.
Thus $P(u)$ and $P(v)$ commute precisely when
$\langle u,v\rangle_s=0$. In that case their product
sign is $(-1)^{\kappa(u,v)/2}$, where the half of an
even representative is taken modulo two.
\end{lemma}

\begin{proof}
The relation $ZX=-XZ$ gives
\[
    D(u)D(v)=(-1)^{\sum_jz_jx'_j}D(u+v).
\]
Exchanging $u$ and $v$ yields the commutation relation.
Inserting the phases in the definition of $P$ gives
\eqref{eq:kappa}. Finally,
\[
    \tau(u+v)
    =\tau(u)+\tau(v)+\sum_j(x_jz'_j+x'_jz_j) \pmod 2,
\]
so $\kappa(u,v)\bmod2=\langle u,v\rangle_s$.
\end{proof}

A Clifford unitary $U$ normalizes the Heisenberg--Weyl
group, meaning that
$U\mathrm{HW}(Q)U^\dagger=\mathrm{HW}(Q)$.
The Clifford unitaries on a qubit register $Q$ form the
group $\operatorname{Cliff}(Q)$.
We use the following binary representation
\cite{rengaswamy2018synthesis}.

\begin{proposition}[Symplectic representation of Clifford unitaries]
\label{Prop: Symplectic Clifford}
For every $U\in\operatorname{Cliff}(\mathbb C^{2^n})$,
there are a binary matrix
$F_U\in\mathbb F_2^{2n\times2n}$ and a sign map
$s_U:\mathbb F_2^{2n}\to\mathbb F_2$ such that
\begin{equation}
    \begin{aligned}
        UP(u)U^\dagger&=(-1)^{s_U(u)}P(uF_U),\\
        F_U\Omega_nF_U^{\mathsf T}&=\Omega_n.
    \end{aligned}
\end{equation}
The first $n$ rows of $F_U$ are the binary labels of the
images of $X_1,\ldots,X_n$, and the last $n$ rows are
those of $Z_1,\ldots,Z_n$. The sign map need not be linear.
\end{proposition}

\begin{proof}
Conjugation by $U$ sends a Hermitian Pauli to a Hermitian
Pauli. Its label and sign are therefore uniquely determined.
Since multiplication of Pauli labels is binary addition,
conjugation induces an additive map $u\mapsto uF_U$;
its values on the standard basis give the rows of $F_U$.
Conjugation also preserves commutation, so
\[
    uF_U\Omega_nF_U^{\mathsf T}v^{\mathsf T}
    =u\Omega_n v^{\mathsf T}
    \qquad\text{for all }u,v.
\]
Taking standard basis vectors proves the matrix identity.
\end{proof}

Matrices satisfying $F\Omega_nF^{\mathsf T}=\Omega_n$
are called \emph{symplectic matrices} and they form the group
$\mathrm{Sp}(2n,\mathbb F_2)$.

\subsection{Vectorization}

Following \cite{chiribella2009theoretical,GutoskiWatrous2007Strategies},
we recall the vectorization identities and Choi representation
used below to describe quantum operations and their composition.
All vectorizations and transposes are taken with respect to
fixed orthonormal bases.

\begin{definition}[Vectorization]
\label{def: Vectorized identity}
For $C:\mathcal H_0\to\mathcal H_1$, define
\begin{equation}
    \begin{aligned}
        \vket{C}&\triangleq\sum_{i,j}C_{ij}
            \ket{i}_{\mathcal H_1}\otimes\ket{j}_{\mathcal H_0},\\
        \vket{I_{\mathcal H}}&=\sum_j\ket{j}\otimes\ket{j}.
    \end{aligned}
\end{equation}
Vectorization is a linear bijection from
$\L(\mathcal H_0,\mathcal H_1)$ to
$\mathcal H_1\otimes\mathcal H_0$, with the output factor
first. Its adjoint is denoted by $\vbra{C}$.
These vectors are unnormalized. In particular,
\[
    \vbra{I_{\mathcal H}}\vket{I_{\mathcal H}}
    =\dim\mathcal H.
\]
\end{definition}

\begin{remark}
\label{remark: choi inverse}
Vectorization also gives the Choi representation of a
completely positive map. If
$\Phi(\rho)=\sum_a C_a\rho C_a^\dagger$, where
$C_a:\mathcal H_0\to\mathcal H_1$, its unnormalized
Choi operator is
\[
    \begin{aligned}
        J_\Phi
        &\triangleq(\Phi\otimes\mathrm{id}_{\mathcal H_0})
             \bigl(\vket{I_{\mathcal H_0}}
                   \vbra{I_{\mathcal H_0}}\bigr)\\
        &=\sum_a\vket{C_a}\vbra{C_a}.
    \end{aligned}
\]
Here $\mathrm{id}_{\mathcal H_0}$ is the identity map
on $\L(\mathcal H_0)$. Conversely, the map is recovered by
\[
    \Phi(\rho)
    =
    \operatorname{Tr}_{\mathcal H_0}
    \!\left[
        J_\Phi
        (I_{\mathcal H_1}\otimes\rho^{\mathsf T})
    \right],
    \qquad \rho\in\L(\mathcal H_0),
\]
where the transpose is taken in the basis used for vectorization.
The correspondence $\Phi\mapsto J_\Phi$ is the Choi
isomorphism for linear maps \cite{Choi1975}.
\end{remark}

\begin{lemma}[Vectorization identities]
\label{lemma: Vectorization identities}
For compatible operators,
\begin{equation}
    \begin{aligned}
        \vket{C}
        &=(C\otimes I_{\mathcal H_0})\vket{I_{\mathcal H_0}}\\
        &=(I_{\mathcal H_1}\otimes C^{\mathsf T})
            \vket{I_{\mathcal H_1}},\\
        (A\otimes B^{\mathsf T})\vket{C}&=\vket{ACB},\\
        \langle\langle A\vket{B}&=\operatorname{Tr}(A^\dagger B).
    \end{aligned}
    \label{eq:prelim-vectorisation-identities}
\end{equation}
In the last identity, $A$ and $B$ have the same domain
and codomain. In particular, for $A,B\in\L(\mathcal H)$,
\[
    \vbra{I_{\mathcal H}}(A\otimes B^{\mathsf T})
       \vket{I_{\mathcal H}}=\operatorname{Tr}(AB).
\]
Contracting an identity ket with an identity bra gives
\[
    \begin{aligned}
        (I_{\mathcal H}\otimes\vbra{I_{\mathcal H}})
        (\vket{I_{\mathcal H}}\otimes I_{\mathcal H})
            &=I_{\mathcal H},\\
        (\vbra{I_{\mathcal H}}\otimes I_{\mathcal H})
        (I_{\mathcal H}\otimes\vket{I_{\mathcal H}})
            &=I_{\mathcal H}.
    \end{aligned}
\]
\end{lemma}

\begin{proof}
Expanding in the fixed bases gives
\[
    (A\otimes B^{\mathsf T})\vket{C}
    =\sum_{p,q}\left(\sum_{i,j}A_{pi}C_{ij}B_{jq}\right)
           \ket{p}\otimes\ket{q}
    =\vket{ACB}.
\]
The first two identities follow by taking $C$ in this
formula to be an identity operator.
Likewise,
$\vbra{A}\vket{B}=\sum_{i,j}A_{ij}^{\star}B_{ij}
=\operatorname{Tr}(A^\dagger B)$, which also gives
the stated trace formula. Finally, each identity
contraction sends a basis ket $\ket{j}$ to $\ket{j}$,
since $\langle i|j\rangle=\delta_{ij}$.
\end{proof}

\section{Quantum circuits}
\label{sec:general-quantum-circuits}

We define quantum circuits through their component operations
and describe their composition by contracting the corresponding
open wires. This separates the circuit from its noise model
and gives a common description of ideal and noisy realizations.
We then specialise this construction to stabilizer circuits.

\begin{definition}[Quantum circuit]
\label{def:quantum_circuit}
A \textit{quantum circuit} $\mathbf C$ is defined by a set
of measurement outcome strings
$o\triangleq(o_1,\ldots,o_\ell)$ of length $\ell$, where
each $o_r\triangleq(o_{r,1},\ldots,o_{r,m'_r})
\in\{0,1\}^{m'_r}$, and a sequence of $\ell$ linear
operators indexed by the measurement outcomes:
\begin{equation}
\resizebox{0.98\columnwidth}{!}{$
    \begin{gathered}
        \mathbf C\triangleq
        \{C_{o_{1:1}},\ldots,C_{o_\ell|o_{1:\ell-1}}\}_o
        \quad\text{for}\quad
        C_{o_r|o_{1:r-1}}:
        Q_{r-1}\longrightarrow A'_r\otimes Q_r.
    \end{gathered}
$}
\end{equation}
Here $Q_r\cong\C^{2^{n_r}}$ and
$o_{1:r-1}\triangleq(o_1,\ldots,o_{r-1})$ is the
measurement outcome history from round $1$ to $r-1$. Each operator $C_{o_r|o_{1:r-1}}$ is composed of a
$\ket{0}^{\otimes m_r}_{A_r}$ preparation on the
$m_r$-qubit ancilla register $A_r\cong\C^{2^{m_r}}$,
a history-dependent unitary
$U_{o_{1:r-1}}\in\U(A_r\otimes Q_{r-1})$,
an $m'_r$-qubit computational-basis projector
$\ketbra{o_r}{o_r}$ on $A'_r\cong\C^{2^{m'_r}}$,
and a conditional unitary $U_{o_r}\in\U(Q_r)$
depending on the outcome $o_r$.
For the action of $U_{o_{1:r-1}}$, we impose
$A_r\otimes Q_{r-1} \cong A'_r\otimes Q_r$, which gives $m_r+n_{r-1}=m'_r+n_r$. Explicitly,
\begin{equation}
\resizebox{0.98\columnwidth}{!}{$
    \begin{aligned}
        C_{o_r|o_{1:r-1}}
        &=(\ketbra{o_r}{o_r}_{A'_r}\otimes U_{o_r})
            U_{o_{1:r-1}}
            (\ket{0}^{\otimes m_r}_{A_r}\otimes\id_{Q_{r-1}})\\
        &=
        \begin{quantikz}
            \setwiretype{n}
              & \lstick[1]{$|0^{m_r}\>_{A_r}$}
              & \gate[2]{U_{o_{1:r-1}}}\setwiretype{q}
              & \meter{}
              & \ctrl{0}\wire[d][1][]{c}
                \wire[l][1]["{o_r\in\{0,1\}^{m'_r}}"{above,pos=0}]{c}
                \setwiretype{c}
              & \\
            \setwiretype{n}
              & Q_{r-1}\;
              & \setwiretype{b}
              & \setwiretype{b}
              & \gate[1]{U_{o_r}}
              & \;Q_r
        \end{quantikz}\;.
    \end{aligned}
$}
\label{eq: circuit def}
\end{equation}
\end{definition}

\begin{definition}[Circuit components]
\label{def: circuit component}
A \textit{component} of $\mathbf C$ is the collection
of linear operators $\{C_{o_r|o_{1:r-1}}\}_{o_r}$ for
a fixed history $o_{1:r-1}$ and round $r\in\mathbb Z^+_{\ell}$.
For each history, these operators satisfy
$\sum_{o_r}C_{o_r|o_{1:r-1}}^\dagger
C_{o_r|o_{1:r-1}}=I_{Q_{r-1}}$,
by unitarity and completeness of the measurement projectors.
\end{definition}

\begin{remark}
Here $m_r=0$ (respectively, $m'_r=0$) means that
$C_{o_r\mid o_{1:r-1}}$ prepares no ancilla (respectively, performs no
measurement). For $r<\ell$, $U_{o_r}$ can be absorbed into the next
round's history-dependent unitary, while the final conditional unitary must
remain unless another component follows. We keep these factors
separate to distinguish dependence on $o_r$ at round $r$ from dependence on its
preceding history. Any quantum instrument on qubit registers admits the decomposition in
\Cref{eq: circuit def} after adding sufficient ancillas and, if needed,
refining outcomes to label individual Kraus operators. Discarding or
coarse-graining the refined outcomes recovers the original instrument
\cite{chiribella2009theoretical,pollock2018non,kliuchnikov2023stabilizercircuitverification}.
\end{remark}

We have defined a quantum circuit as a sequence of linear
operators indexed by measurement outcomes. To describe its
action under a noise model, or in the ideal case, we now
associate a linear map with this circuit description.

\begin{definition}[Circuit map and noise model]
\label{def: Circuit map and noise model}
For a quantum circuit $\mathbf C$, write
$A_r^{\mathrm{out}}\equiv A'_r$ for the retained outcome
register and $A_{1:\ell}^{\mathrm{out}}
\triangleq\bigotimes_{r=1}^\ell A_r^{\mathrm{out}}$.
Let
$\vket{C_{o_{1:r}}}\in
A_r^{\mathrm{out}}\otimes Q_r\otimes Q_{r-1}$
be the vectorization of $C_{o_{1:r}}$.
We define its \textit{circuit map}
\begin{equation}
    \begin{aligned}
        \mathbf C_{\mathrm{cm}}:\L(Q_0)
        \longrightarrow
        \L\bigl(A_{1:\ell}^{\mathrm{out}}\otimes Q_\ell
            \otimes Q_{\ell-1}^{\otimes2}\otimes\cdots\otimes Q_0^{\otimes2}\bigr)
    \end{aligned}
\end{equation}
as the CP map
\begin{equation}
    \begin{gathered}
        \mathbf C_{\mathrm{cm}}(\cdot)
        =\sum_o\mathbf C_o(\cdot)\mathbf C_o^\dagger
        \quad\text{for}\\
        \mathbf C_o
        =\vket{C_{o_{1:\ell}}}\otimes\cdots
            \otimes\vket{C_{o_{1:1}}}\otimes I_{Q_0}.
    \end{gathered}
\end{equation}
A \textit{noise model} over $\mathbf C$ is a CP map
\begin{equation}
    \begin{aligned}
        \mathbf E:
        \L\bigl(A_{1:\ell}^{\mathrm{out}}\otimes Q_\ell
            \otimes Q_{\ell-1}^{\otimes2}\otimes\cdots
            \otimes Q_0^{\otimes2}\bigr) \longrightarrow
        \L(A_{1:\ell}^{\mathrm{out}}\otimes Q_\ell)
    \end{aligned}
\end{equation}
with Kraus operators $\{\mathbf E_e\}_e$ satisfying 
\begin{equation}
    \begin{gathered}
        \sum_e\mathbf E_e^\dagger\mathbf E_e
          =I_{A_\ell^{\mathrm{out}}Q_\ell}
               \otimes\mathbf E^{(\ell-1)},\\
        \operatorname{Tr}_{Q_r}\mathbf E^{(r)}
          =I_{A_r^{\mathrm{out}}Q_r}\otimes\mathbf E^{(r-1)},
        \qquad 1\leq r<\ell,\\
        \operatorname{Tr}_{Q_0}\mathbf E^{(0)}=I_{Q_0},
        \qquad \mathbf E^{(r)}\succeq0
               \quad(0\leq r<\ell).
    \end{gathered}
    \label{eq: Noise comb constraints}
\end{equation}
A \textit{noisy realization} of $\mathbf C$ under
$\mathbf E$ is the channel
\[
    \mathbf C_{\mathrm E}(\cdot)
       \triangleq(\mathbf E\circ\mathbf C_{\mathrm{cm}})(\cdot)
\]
from $\L(Q_0)$ to
$\L(A_{1:\ell}^{\mathrm{out}}\otimes Q_\ell)$,
with Kraus operators $\{\mathbf E_e\mathbf C_o\}_{e,o}$.
An \textit{ideal realization} is given by
\begin{equation}
    \begin{gathered}
        \mathbf C_{\mathrm{id}}(\cdot)
           \triangleq
           \mathbf I\mathbf C_{\mathrm{cm}}(\cdot)\mathbf I^\dagger
           \quad\text{for}\\
        \mathbf I
           =I_{A_{1:\ell}^{\mathrm{out}}Q_\ell}
              \otimes\vbra{I_{Q_{\ell-1}}}
              \otimes\cdots\otimes\vbra{I_{Q_0}}.
    \end{gathered}
\end{equation}
\end{definition}

\begin{remark}
The circuit map $\mathbf C_{\mathrm{cm}}$ is generally
not trace-preserving. Its unnormalized vectorizations give
\[
    \operatorname{Tr}\mathbf C_{\mathrm{cm}}(\rho)
      =\left(\prod_{r=0}^{\ell-1}\dim Q_r\right)
           \operatorname{Tr}\rho.
\]
However, completeness relations of the circuit components and
\Cref{eq: Noise comb constraints} imply
\[
    \sum_{e,o}
       (\mathbf E_e\mathbf C_o)^\dagger
        (\mathbf E_e\mathbf C_o)=I_{Q_0}.
\]
Thus, both the noisy and ideal realizations of the circuit are trace-preserving.
Both identities are proved in \Cref{app:causal-normalisation}. Moreover, the circuit map satisfies the quantum-comb causal normalization
constraints \cite{chiribella2009theoretical,pollock2018non}
\begin{equation}
    \begin{gathered}
        \sum_o\mathbf C_o\mathbf C_o^\dagger
           =\mathbf C^{(\ell)},
        \qquad \mathbf C^{(\ell)}\succeq0,\\
        \operatorname{Tr}_{A_r^{\mathrm{out}}Q_r}
           \mathbf C^{(r)}
           =I_{Q_{r-1}}\otimes\mathbf C^{(r-1)},
        \qquad r\in\mathbb Z^+_{\ell},\\
        \mathbf C^{(r-1)}\succeq0,
        \qquad \mathbf C^{(0)}=I_{Q_0}.
    \end{gathered}
\end{equation}
See
\Cref{app:causal-normalisation} for more details.
\end{remark}

\begin{remark}[Noise model explicit form]
\label{rem:noise_model_explicit}
Consider a circuit $\mathbf C$ with circuit map
$\mathbf C_{\mathrm{cm}}(\cdot)
=\sum_o\mathbf C_o(\cdot)\mathbf C_o^\dagger$.
A noise map acts on it as
\[
    (\mathbf E\circ\mathbf C_{\mathrm{cm}})(\cdot)
       =\sum_{e,o}\mathbf E_e\mathbf C_o(\cdot)
                        \mathbf C_o^\dagger\mathbf E_e^\dagger.
\]
The noise operators at time-step $r \in \{0,\cdots,\ell\}$ are given by
\begin{equation}
    \begin{aligned}
        E_{e_r}:A_r^{\mathrm{out}}\otimes Q_r\otimes N_r
        &\longrightarrow
        A_r^{\mathrm{out}}\otimes Q_r\otimes N_{r+1},\\
        \sum_{e_r}E_{e_r}^\dagger E_{e_r}
        &=I_{A_r^{\mathrm{out}}Q_rN_r},
        \quad r\in \{0,\cdots,\ell\}.
    \end{aligned}
\end{equation}
Here $A_0^{\mathrm{out}}=\mathbb C$ and $N_r$ is the
noise memory, with $N_0=N_{\ell+1}=\mathbb C$. 
Put $e=(e_0,\ldots,e_\ell)$ and
$I_N=I_{N_\ell}\otimes\cdots\otimes I_{N_1}$.
The spacetime noise operator takes the form
\begin{equation}
    \begin{aligned}
        \mathbf E_e
        &=\bigl(I_{A_{1:\ell-1}^{\mathrm{out}}}
             \otimes E_{e_\ell}
             \otimes\vbra{E_{e_{\ell-1}}^{\star}}
             \otimes\cdots\otimes\vbra{E_{e_0}^{\star}}\bigr)\\
        &\qquad{}\times
          \bigl(\vket{I_{A_{1:\ell-1}^{\mathrm{out}}}}
             \otimes\vket{I_N}\otimes I_{\mathrm{in}}\bigr).
    \end{aligned}
\end{equation}
Here $I_{\mathrm{in}}$ is the identity on the domain
of $\mathbf E_e$. By \Cref{lemma: Vectorization identities}, it can be shown that its explicit action on a circuit
branch is
\begin{equation}
    \mathbf E_e\mathbf C_o
       =E_{e_\ell}(C_{o_{1:\ell}}\otimes I_{N_\ell})
          \cdots E_{e_1}(C_{o_1}\otimes I_{N_1})E_{e_0}.
\end{equation}
For a memoryless noise model, setting $N_r=\mathbb C$
for all $r$, we have
\begin{equation}
    \mathbf E_e\mathbf C_o
       =E_{e_\ell}C_{o_{1:\ell}}E_{e_{\ell-1}}
          \cdots E_{e_1}C_{o_1}E_{e_0}.
\end{equation}
\end{remark}
\begin{figure*}[t]
    \centering
    \begin{quantikz}
        \lstick{$|0\rangle_A$}
            & \qnoise{P(e_{0,0})} & \targ{}
            & \qnoise{P(e_{0,1})} & \targ{}
            & \qnoise{P(e_{0,2})} & \meter{}
            & \qnoise{P(e_{0,3})}\setwiretype{c} & \setwiretype{c}\\
        \lstick{$Q_1$}
            & \qnoise{P(e_{1,0})} & \ctrl{-1}
            & \qnoise{P(e_{1,1})} &
            & \qnoise{P(e_{1,2})} &
            & \qnoise{P(e_{1,3})} & \\
        \lstick{$Q_2$}
            & \qnoise{P(e_{2,0})} &
            & \qnoise{P(e_{2,1})} & \ctrl{-2}
            & \qnoise{P(e_{2,2})} &
            & \qnoise{P(e_{2,3})} & \\
        \lstick{$Q_3$}
            & \qnoise{P(e_{3,0})} &
            & \qnoise{P(e_{3,1})} &
            & \qnoise{P(e_{3,2})} &
            & \qnoise{P(e_{3,3})} &
    \end{quantikz}
    \caption{A fixed Pauli fault pattern for the $Z_1Z_2$
    measurement circuit indicating the Kraus operator in \Cref{eq: Krus operator for fig}. Each $\{P(e_{s,t})\}_{s,t=0}^3$ is a single qubit Pauli fault.}
    \label{fig: component-example-circuit-noise}
\end{figure*}
\begin{example}[Preparation, gate, and measurement components]
\label{example: Circuits and noise}
A circuit may be resolved into elementary state preparations,
unitary gates, and measurement projectors. Let
$Q=Q_1\otimes Q_2\otimes Q_3$ and
$I_Q=I_1\otimes I_2\otimes I_3$.
Here the subscripts label physical data qubits.
Consider the four circuit components
\begin{equation}
    \begin{aligned}
        C^{(1)}&=|0\rangle_A\otimes I_Q,\quad
        C^{(2)}=\cnot_{1\rightarrow A},\\
        C^{(3)}&=\cnot_{2\rightarrow A},\quad
        C_o^{(4)}=|o\rangle\langle o|_A\otimes I_Q,
        \qquad o\in\{0,1\}.
    \end{aligned}
\end{equation}
In circuit notation, the components are illustrated by
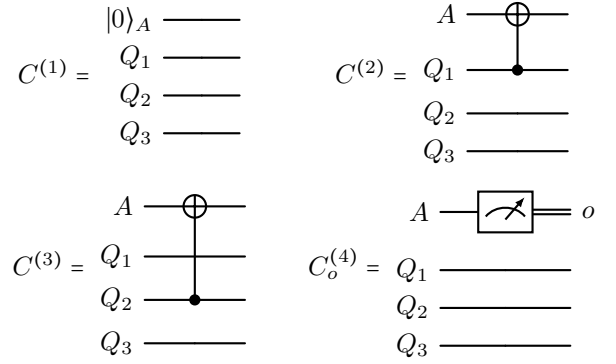
\[
\begin{array}{@{}c@{\qquad}c@{}}
C^{(1)}=
\begin{quantikz}
    \lstick{$|0\rangle_A$} & & \\
    \lstick{$Q_1$} & & \\
    \lstick{$Q_2$} & & \\
    \lstick{$Q_3$} & &
\end{quantikz}
&
C^{(2)}=
\begin{quantikz}
    \lstick{$A$} & \targ{} & \\
    \lstick{$Q_1$} & \ctrl{-1} & \\
    \lstick{$Q_2$} & & \\
    \lstick{$Q_3$} & &
\end{quantikz}
\\[1.2em]
C^{(3)}=
\begin{quantikz}
    \lstick{$A$} & \targ{} & \\
    \lstick{$Q_1$} & & \\
    \lstick{$Q_2$} & \ctrl{-2} & \\
    \lstick{$Q_3$} & &
\end{quantikz}
&
C_o^{(4)}=
\begin{quantikz}
    \lstick{$A$} & \meter{} & \rstick{$o$}\setwiretype{c}\\
    \lstick{$Q_1$} & & \\
    \lstick{$Q_2$} & & \\
    \lstick{$Q_3$} & &
\end{quantikz}
\end{array}
\]
Vectorizing the components gives the circuit map
\begin{equation}
    \begin{aligned}
        \mathbf C_{\mathrm{cm}}(\rho)
        &=\sum_{o=0}^1\mathbf C_o\rho\mathbf C_o^\dagger,\\
        \mathbf C_o
        &=\vket{C_o^{(4)}}\otimes\vket{C^{(3)}}\\
        &\quad{}\otimes\vket{C^{(2)}}\otimes\vket{C^{(1)}}
            \otimes I_Q.
    \end{aligned}
\end{equation}
The ideal realization of the circiuit is given by
\begin{equation}
    \begin{aligned}
        K_o&=\mathbf I\mathbf C_o
            =C_o^{(4)}C^{(3)}C^{(2)}C^{(1)}\\
           &=|o\rangle_A\otimes\Pi_o^{Z_1Z_2},\\
        \Pi_o^{Z_1Z_2}
           &\triangleq\tfrac12\bigl(I_Q+(-1)^oZ_1Z_2\bigr).
    \end{aligned}
\end{equation}
Indeed, the two CNOTs write the computational-basis parity
of $Q_1,Q_2$ into $A$, so outcome $o$ projects the data
onto the $(-1)^o$ eigenspace of $Z_1Z_2$. Thus
\begin{equation}
    \mathbf C_{\mathrm{id}}(\rho)
    =\sum_{o=0}^1K_o\rho K_o^\dagger
\end{equation}
is the nondestructive measurement of $Z_1Z_2$ on systems $Q_1,Q_2$, with outcome $o$
stored in $A$ and $Q_3$ left unchanged
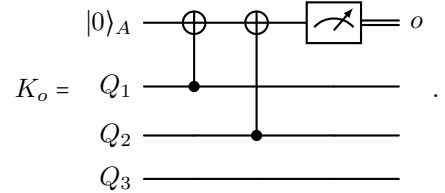
\[
K_o=
\begin{quantikz}
    \lstick{$|0\rangle_A$}
        & \targ{} & \targ{} & \meter{}
        & \rstick{$o$}\setwiretype{c}\\
    \lstick{$Q_1$} & \ctrl{-1} & & & \\
    \lstick{$Q_2$} & & \ctrl{-2} & & \\
    \lstick{$Q_3$} & & & &
\end{quantikz}.
\]
Consider a Pauli fault model, 
where noise operators $\{\mathbf E_e\}_e$, are given by
\begin{equation}
    \mathbf E_e \triangleq P(e_3)\otimes\vbra{P(e_2)^\star} \otimes \vbra{P(e_1)^\star} \otimes \vbra{P(e_0)^\star}\otimes\vbra{I_Q}.
\end{equation}
For every $t\in\{0,1,2,3\}$, write $P(e_t)=P(e_{0,t})\otimes P(e_{1,t})\otimes P(e_{2,t})\otimes P(e_{3,t})$, where $P(e_{0,t})\in\mathcal P(A)$ and $P(e_{j,t})\in\mathcal P(Q_j)$ for $j\in\{1,2,3\}$.
The kraus operators for the noisy realization of the circuit is therefore,
\begin{equation}
    \mathbf E_e \mathbf C_o = P(e_3)C^{(4)}_oP(e_2)C^{(3)}P(e_1)C^{(2)}P(e_0)C^{(1)}. \label{eq: Krus operator for fig}
\end{equation}
This operator is illustrated in
\Cref{fig: component-example-circuit-noise}.
Given spacetime fault probabilities $p_{e}$ for fault $\mathbf E_e$, the noisy realization $\mathbf C_E(\cdot)$ is a CPTP map
\begin{equation}
    \begin{aligned}
        \mathbf C_{\mathrm E}(\rho)
        &=(\mathbf E\circ\mathbf C_{\mathrm{cm}})(\rho)\\
        &=\sum_{e,o}p_{e}
          \mathbf E_e\mathbf C_o\rho \mathbf C_o^\dagger\mathbf E^\dagger_e,
    \end{aligned}
\end{equation}
where $p_{e}\geq0$ and $\sum_{ e}p_{e}=1$.
\end{example}


\subsection{Stabilizer circuits}

Syndrome extraction and many logical operations on stabilizer
codes use Clifford gates, Pauli eigen-state preparations and Pauli
measurements. We define stabilizer circuits in the component
form of \Cref{def:quantum_circuit}, allowing different levels
of abstraction and the noise models introduced above. Representing circuits as states enables us to develop a stabilizer formalism for quantum circuits analogous to that for quantum states. In \Cref{sec:circuit_stabilizers}, we show how stabilizers of a circuit's components compose and how they determine
detectors, measured observables and preserved logical information.
In \Cref{sec: Fault tolerance of stabilizer circuits}, we
specialize to non-adaptive circuits under Pauli noise and
use these symmetries to obtain explicit fault signatures
and a minimum-weight characterization of circuit distance.

We first recall stabilizer states and their operator
representation \cite{gottesman1997stabilizer,
kliuchnikov2023stabilizercircuitverification}.
Related circuit and fault-tolerance descriptions are developed
in \cite{Bombin_2023,beverland2024faulttolerancestabilizerchannels}.

\begin{definition}[Stabilizer states and operators]
\label{def: Stab states op}
An $n$-qubit unit vector $|\psi\rangle\in\mathbb C^{2^n}$
is a \textit{pure stabilizer state} if there is a Pauli
subgroup $\mathcal S_\psi\subseteq\mathcal P(\mathbb C^{2^n})$
with $n$ independent generators and
$-I\notin\mathcal S_\psi$ such that
$s|\psi\rangle=|\psi\rangle$ for every
$s\in\mathcal S_\psi$. Such a subgroup is abelian \footnote{If
$P,Q\in\mathcal S_\psi$ anticommute, their group commutator
$PQP^{-1}Q^{-1}=-I$ belongs to $\mathcal S_\psi$,
contrary to $-I\notin\mathcal S_\psi$.}. An $n$-qubit density operator $\rho$ is a
\textit{stabilizer state} if it is a convex mixture of
pure stabilizer states
\[
    \rho=\sum_j p_j|\psi_j\rangle\langle\psi_j|,
    \qquad p_j\geq0,\qquad \sum_jp_j=1.
\]

A linear operator
$C:\mathbb C^{2^{n_\ins}}\to\mathbb C^{2^{n_\outs}}$
is a \textit{stabilizer operator} if there is a Pauli
subgroup
$\mathcal S_{\vket C}\subseteq
\mathcal P(\mathbb C^{2^{n_\outs+n_\ins}})$
with $n_\outs+n_\ins$ independent generators and
$-I\notin\mathcal S_{\vket C}$ such that
$s\vket C=\vket C$ for every $s\in\mathcal S_{\vket C}$.
\end{definition}

\begin{lemma}
\label{Lemma: Fundamental stabilizer lemma}
For a linear operator
$A:\mathbb C^{2^{n_\ins}}\to\mathbb C^{2^{n_\outs}}$,
put $n=n_\ins+n_\outs$. The following are equivalent:
\begin{enumerate}
    \item $A=0$.
    \item There is a Hermitian Pauli
    $P\in\mathcal P(\mathbb C^{2^n})$ such that
    $P\vket A=\vket A$ and $-P\vket A=\vket A$.
    \item $P\vket A=\vket A$ for every
    $P\in\mathcal P(\mathbb C^{2^n})$.
\end{enumerate}
\end{lemma}

\begin{proof}
If $A=0$, every Pauli fixes $\vket A$, proving
$(1)\Rightarrow(3)$. Taking $P=I$ and $-P=-I$ gives
$(3)\Rightarrow(2)$. Conversely, adding the two equalities
in (2) gives $2\vket A=0$. Vectorization is injective,
so $A=0$, proving $(2)\Rightarrow(1)$.
\end{proof}

\begin{definition}[Stabilizer circuit]
\label{def: stab circuit}
A \textit{stabilizer circuit} $\mathbf C$ is specified by
measurement outcome strings
$o=(o_1,\ldots,o_\ell)$, with
$o_r=(o_{r,1},\ldots,o_{r,m'_r})\in\{0,1\}^{m'_r}$,
and a sequence of $\ell$ operators indexed by these outcomes:
\begin{equation}
    \begin{gathered}
        \mathbf C=\{C_{o_{1:1}},\ldots,C_{o_{1:\ell}}\}_o,
        \qquad
        C_{o_{1:r}}:Q_{r-1}\longrightarrow A'_r\otimes Q_r.
    \end{gathered}
\end{equation}
Here $Q_r\cong\mathbb C^{2^{n_r}}$,
$o_{1:r}=(o_1,\ldots,o_r)$, and
$\{C_{o_{1:r}} \triangleq C_{o_r|o_{1:r-1}}\}_{o_r}$ denotes its $r^{th}$ component. Each component operator is composed of a
$|0^{m_r}\rangle_{A_r}$ preparation, a history-dependent
Clifford unitary $U_{o_{1:r-1}}$ on $A_r\otimes Q_{r-1}$,
a computational-basis projector
$|o_r\rangle\langle o_r|$ on $A'_r$, and a conditional
Pauli correction $P_{o_r} \in \mathcal P(Q_r)$. The register
identification $A_r\otimes Q_{r-1}\cong A'_r\otimes Q_r$
requires
\[
    m_r+n_{r-1}=m'_r+n_r,
\]
and these dimensions may vary with $r$. Explicitly,
\begin{equation}
\resizebox{0.98\columnwidth}{!}{$
    \begin{aligned}
        C_{o_{1:r}}
        &=(|o_r\rangle\langle o_r|_{A'_r}\otimes P_{o_r})
            U_{o_{1:r-1}}
            (|0^{m_r}\rangle_{A_r}\otimes I_{Q_{r-1}})\\
        &=
        \begin{quantikz}
            \setwiretype{n}
              & \lstick[1]{$|0^{m_r}\rangle_{A_r}$}
              & \gate[2]{U_{o_{1:r-1}}}\setwiretype{q}
              & \meter{}
              & \ctrl{0}\wire[d][1][]{c}
                \wire[l][1]["{o_r\in\{0,1\}^{m'_r}}"{above,pos=0}]{c}
                \setwiretype{c}
              & \\
            \setwiretype{n}
              & Q_{r-1}\;
              & \setwiretype{b}
              & \setwiretype{b}
              & \gate[1]{P_{o_r}}
              & \;Q_r
        \end{quantikz}\;.
    \end{aligned}
$}
\end{equation}
The measured register $A'_r$ is denoted as
$A_r^{\mathrm{out}}$ in the circuit map.
\end{definition}

\begin{remark}
    The definition allows adaptive choices. The results below
for non-adaptive circuits require the gates and measured
Paulis, including any physically applied corrections,
to be fixed independently of the outcomes. We leave the analysis for adaptive circuits for future work.
\end{remark}

\begin{definition}[Realizable outcomes, detectors and observables]
\label{def: Det obs gen def}
Let $A'_r$ be the measured ancillary register of component
$r$, with $\log_2\dim A'_r=m'_r$, and put
$\ell_M=\sum_{r\in\mathbb Z^+_{\ell}}m'_r$.
We identify the complete outcome string with the binary row vector
\[
    o=[o_{1,1},\ldots,o_{1,m'_1},\ldots,
         o_{\ell,1},\ldots,o_{\ell,m'_\ell}]
      \in\mathbb F_2^{\ell_M}.
\]
All dot products below are ordinary binary dot products,
$o\cdot d=od^{\mathsf T}$.
An outcome $o$ is \emph{realizable} if and only if
$\mathbf I\mathbf C_o\neq0$; equivalently, it occurs
with nonzero probability for some input state. The
\emph{realizable outcome set} is
\begin{equation}
    \Omega_C\triangleq
    \{o\in\mathbb F_2^{\ell_M}:\mathbf I\mathbf C_o\neq0\}.
\end{equation}
It is nonempty because the ideal circuit is trace preserving.

A vector $d\in\mathbb F_2^{\ell_M}$ is a \emph{detector}
with parity $a_d\in\mathbb F_2$ if, for every outcome $o$,
\begin{equation}
    \mathbf I\mathbf C_o\neq0
    \quad\Longrightarrow\quad o\cdot d=a_d.
\end{equation}
Let $P\in\mathcal P(Q_0)$ be a Hermitian input Pauli,
including its sign, and let
\[
    \Pi_P^b\triangleq\frac{I_{Q_0}+(-1)^bP}{2},
    \qquad b\in\{0,1\},
\]
be the projector onto its $(-1)^b$ eigenspace.
A vector $v_P\in\mathbb F_2^{\ell_M}$ is an
\emph{observable representative} of $P$ with parity
$a_{v_P}\in\mathbb F_2$ if, for every $o$ and $b$,
\begin{equation}
    \mathbf I\mathbf C_o\Pi_P^b\neq0
    \quad\Longrightarrow\quad
    o\cdot v_P=a_{v_P}\oplus b.
\end{equation}
We write
\begin{equation}
    \begin{aligned}
        \mathcal D_C
        &\triangleq\{d:d\text{ is a detector of }\mathbf C\},\\
        \mathcal O_C(P)
        &\triangleq\{v_P:v_P\text{ represents }P\text{ in }\mathbf C\}
    \end{aligned}
\end{equation}
for the detecting set and the observable set associated
with $P$. We call $P$ \emph{measurable} by $\mathbf C$
when $\mathcal O_C(P)\neq\varnothing$, and define the
overall observable set by
\[
    \mathcal O_C\triangleq
    \bigcup_{P\in\mathcal P(Q_0)}\mathcal O_C(P).
\]
Since $\Omega_C\neq\varnothing$, each detector parity is
unique. For fixed $P$, each observable representative
also has a unique parity. 
\end{definition}

\begin{remark}
\label{remark: measureable Identity is detector}
The detector set is precisely the observable set associated
with the identity:
\[
    \mathcal D_C=\mathcal O_C(I).
\]
Indeed, $\Pi_I^0=I$ and $\Pi_I^1=0$, so the observable
condition reduces to
\[
    \mathbf I\mathbf C_o\neq0
    \quad\Longrightarrow\quad o\cdot v_I=a_{v_I},
\]
which is the detector condition. In particular, $0$ is always a detector with parity zero, and the nontrivial detectors
are therefore the elements of $\mathcal D_C\setminus\{0\}$.

Moreover, $0\in\mathcal O_C(P)$ if and only if $P=\pm I$.
To see the forward implication, let $a$ be the parity
of the zero representative. Its defining condition forces
$\mathbf I\mathbf C_o\Pi_P^{a\oplus1}=0$ for every $o$.
Trace preservation of the ideal circuit then gives
\[
    \begin{aligned}
        \Pi_P^{a\oplus1}
        &=\sum_o\Pi_P^{a\oplus1}
          (\mathbf I\mathbf C_o)^\dagger
          (\mathbf I\mathbf C_o)\Pi_P^{a\oplus1}\\
        &=0.
    \end{aligned}
\]
Thus $P=(-1)^aI$. Conversely, $P=(-1)^aI$ has the
zero representative with parity $a$.
For explicit examples of detectors, observables and
preserved logical qubits at levels $0$ and $1$, see
\Cref{example: Z measurement L0,ex:honeycomb-level-1}.
\end{remark}

\begin{proposition}[Observable redundancy]
\label{prop:obs_redundancy}
Let $P$ be a measurable Pauli of a stabilizer circuit
$\mathbf C$. For any two representatives
$v_P,v'_P\in\mathcal O_C(P)$ with parities
$a_{v_P},a_{v'_P}$, there is a detector
$d\in\mathcal D_C$ with parity $a_d$ such that
\[
    v'_P=v_P\oplus d,
    \qquad a_{v'_P}=a_{v_P}\oplus a_d.
\]
\end{proposition}

\begin{proof}
Set $d=v_P\oplus v'_P$ and
$a_d=a_{v_P}\oplus a_{v'_P}$.
For any $o\in\Omega_C$, the identity
$\Pi_P^0+\Pi_P^1=I$ implies that
$\mathbf I\mathbf C_o\Pi_P^b\neq0$ for at least one
$b\in\{0,1\}$. The two representative conditions give
\[
    o\cdot v_P=a_{v_P}\oplus b,
    \qquad o\cdot v'_P=a_{v'_P}\oplus b.
\]
Adding these equalities modulo two yields $o\cdot d=a_d$.
This holds for every realisable outcome, so $d$ is a
detector with the stated parity.
\end{proof}

\begin{remark}
\label{remark: Adaptive case}
Operationally, a detector $d\in\mathcal D_C$ selects the
outcome indices
$\mathcal I_d=\{i\in\mathbb Z^+_{\ell_M}:d_i=1\}$.
In the absence of noise, every realizable outcome satisfies
\[
    \sum_{i\in\mathcal I_d}o_i=a_d\pmod 2.
\]
A violation of this parity therefore detects an error.
Satisfying all detector parities, however, does not generally
characterize outcome realizability for an adaptive circuit.

For example, consider an adaptive sequence of projective measurements.
First measure $Z$. If $o_1=0$, measure $Z$ again;
if $o_1=1$, measure $X$. Finally, repeat the observable
chosen for the second measurement.
When $o_1=0$, the qubit is in $\lvert0\rangle$, so both
subsequent outcomes are zero. When $o_1=1$, either $X$
outcome is possible, and the third measurement repeats it.
Hence the realizable outcome set is
\[
    \Omega_C=\{000,100,111\}.
\]
Any detector parity
$d_1o_1+d_2o_2+d_3o_3=a_d$ must hold on these three
strings, which forces $a_d=0$, $d_1=0$, and $d_2=d_3$.
Thus the only nontrivial detector parity is
\[
    o_2+o_3=0\pmod 2.
\]
The string $011$ satisfies this parity but is not realizable, namely,
$o_1=0$ requires the second measurement to be $Z$ on
$\lvert0\rangle$, so $o_2$ must be zero.
The adaptive circuit therefore imposes the additional
conditional restriction $o_1=0\Rightarrow o_2=0$,
which the detector parities alone do not capture. We leave characterization of detectors and observables for adaptive circiuts for future work.
\end{remark}

\begin{remark}
Similarly, an observable representative $v_P$ selects
$\mathcal I_{v_P}=\{i\in\mathbb Z^+_{\ell_M}:(v_P)_i=1\}$.
If the input lies in the $(-1)^b$ eigenspace of $P$, then
\[
    \sum_{i\in\mathcal I_{v_P}}o_i=a_{v_P}\oplus b\pmod2.
\]
Thus $o\cdot v_P\oplus a_{v_P}$ gives the measured
eigenvalue bit. By \Cref{prop:obs_redundancy}, different
representatives of the same observable give the same bit
on every realisable outcome.
For an arbitrary input state $\rho$, branches assigned
the bit $b$ annihilate the opposite eigenspace. Together
with trace preservation, this implies that their total
probability is $\operatorname{Tr}(\Pi_P^b\rho)$. A violation in the measured observable parity therefore is a logical measurement error.
\end{remark}

\subsection{Stabilizers of a stabilizer circuit}\label{sec:circuit_stabilizers}
The stabilizers of a vectorized circuit describe which input
Paulis are measured, which detectors are enforced, and
which logical Pauli operators are preserved. We first make
these interpretations precise and show how to obtain the
stabilizers of a composed circuit by matching operators on
its internal wires
(see \Cref{prop: stabilizers of composition}). For non-adaptive
circuits, a common set of relations then gives the detectors,
observable representatives and preserved logical pairs
(\Cref{prop:outcome-independent-stabilizer-relations}).
Finally, \Cref{prop: Code conditioned stabilizers L0,prop: code-conditioned stabilizers level 1}
give explicit binary formulas at levels $0$ and $1$, including
conditioning on an input code. The honeycomb example (see \Cref{ex:honeycomb-level-1})
illustrates how these formulas extract all three types of
information from a measurement schedule.

\begin{definition}[{Input and output stabilizers and preserved operators of a stabilizer circuit}]
    For a nonzero stabilizer operator $C:Q\rightarrow Q'$ with $Q\cong\C^{2^n}$ and $Q'\cong\C^{2^{n'}}$, a Hermitian Pauli operator $s\in\mathcal{P}(Q)$ such that $Cs = C$ is an \textit{input stabilizer} of $C$, illustrated as
    \begin{equation}
        \begin{quantikz}
            \lstick[1]{$s$} \setwiretype{b} & \gate[1]{C} \wire[l][1]["{Q}"{above,pos=0.5}]{q} \wire[r][1]["{Q'}"{above,pos=0.5}]{q} & 
        \end{quantikz}
            =
            \begin{quantikz}
                \setwiretype{b} & \gate[1]{C} \wire[l][1]["{Q}"{above,pos=0.5}]{q} \wire[r][1]["{Q'}"{above,pos=0.5}]{q} & 
            \end{quantikz} \;.
    \end{equation}
    A Pauli operator $s'\in\mathcal{P}(Q')$ such that $s'C=C$ is an \textit{output stabilizer} of $C$, illustrated as
    \begin{equation}
        \begin{quantikz}
            \setwiretype{b} & \gate[1]{C} \wire[l][1]["{Q}"{above,pos=0.5}]{q} \wire[r][1]["{Q'}"{above,pos=0.5}]{q} & \rstick[1]{$s'$}
        \end{quantikz}
            =
            \begin{quantikz}
                \setwiretype{b} & \gate[1]{C} \wire[l][1]["{Q}"{above,pos=0.5}]{q} \wire[r][1]["{Q'}"{above,pos=0.5}]{q} & 
            \end{quantikz} \;.
    \end{equation}
    In the vectorized representation, $s\in\mathcal{P}(Q)$ is an input stabilizer of $C$ if $s_Q^T|C\rr = |C\rr$ and $s'\in\mathcal{P}(Q')$ is an output stabilizer of $C$ if $s_{Q'}'|C\rr = |C\rr$.
    A pair of Hermitian Paulis $L\in\mathcal P(Q)$ and
    $L'\in\mathcal P(Q')$, including their signs, is a
    \textit{preserved operator pair} if
    $L'CL=C$, or equivalently $L'C=CL$.
    In the vectorized representation, this is
    $(L'\otimes L^{\mathsf T})\vket C=\vket C$,
    illustrated as
    \begin{equation}
        \begin{quantikz}
            \lstick[1]{$L$} \setwiretype{b} & \gate[1]{C} \wire[l][1]["{Q}"{above,pos=0.5}]{q} \wire[r][1]["{Q'}"{above,pos=0.5}]{q} & \rstick[1]{$L'$}
        \end{quantikz}
            =
            \begin{quantikz}
                \setwiretype{b} & \gate[1]{C} \wire[l][1]["{Q}"{above,pos=0.5}]{q} \wire[r][1]["{Q'}"{above,pos=0.5}]{q} & 
            \end{quantikz} \;.
    \end{equation}
    Input and output stabilizers are the special cases
    $(L,I)$ and $(I,L')$, respectively. We call a preserved
    pair \emph{nontrivial} when $CL\notin\{C,-C\}$;
    equivalently, $L'C\notin\{C,-C\}$. Such a pair relates
    a nontrivial input logical Pauli to its output counterpart.
    This differs from a \emph{measurable input Pauli}, whose
    value is inferred from the classical record as in
    \Cref{def: Det obs gen def}.
\end{definition}

\begin{remark}
A preserved pair transfers Pauli eigenvalues through the
operator: if $L|\varphi^\pm\rangle=\pm|\varphi^\pm\rangle$,
then
\[
    L'C|\varphi^\pm\rangle
    =CL|\varphi^\pm\rangle
    =\pm C|\varphi^\pm\rangle.
\]
Thus every nonzero image $C|\varphi^\pm\rangle$ retains
the corresponding eigenvalue of $L'$. The signs of $L,L'$
matter; for example, the gate $X$ preserves the pair
$(Z,-Z)$.
\end{remark}

For $C\neq0$, the vector $\vket C$ is proportional to a
pure stabilizer state on $n+n'$ qubits. Its stabilizer group
therefore has $n+n'$ independent generators. We use the
\emph{stacked} notation
\[
    \stk{P'}{P}\triangleq P'\otimes P^{\mathsf T},
\]
with the output Pauli above the input Pauli. By
vectorization,
\[
    \stk{P'}{P}\vket C=\vket C
    \quad\Longleftrightarrow\quad P'CP=C.
\]
The transpose is part of this notation, including its sign
for $Y$ factors.

\begin{proposition}
\label{prop: Abelian}
For a nonzero stabilizer operator $C:Q\to Q'$, its input
and output stabilizers form abelian Pauli subgroups
$\mathcal S_{\mathrm{in}}$ and $\mathcal S_{\mathrm{out}}$,
neither containing $-I$. Every preserved pair $(P,P')$
satisfies
\[
    P\in\mathcal N(\mathcal S_{\mathrm{in}}),
    \qquad P'\in\mathcal N(\mathcal S_{\mathrm{out}}),
\]
where the normalizers are taken in the corresponding
Heisenberg--Weyl groups. Moreover, $(sP,s'P')$ is a
preserved pair for all $s\in\mathcal S_{\mathrm{in}}$
and $s'\in\mathcal S_{\mathrm{out}}$.
\end{proposition}
\begin{proof}
If $Cs=C$ and $Ct=C$, then $Cst=C$. If $s,t$
anticommuted, we would also have $Cst=-Cts=-C$,
contradicting $C\neq0$. Thus the input stabilizers commute,
are closed under multiplication, and exclude $-I$.
The output argument is identical.

For a preserved pair, $P'C=CP$. If an input stabilizer
$s$ anticommuted with $P$, then
\[
    CP=P'Cs=CPs=-CsP=-CP,
\]
which is impossible because $P$ is invertible and $C\neq0$.
Hence $P$ commutes with every input stabilizer; similarly,
$P'$ commutes with every output stabilizer. Within the
Pauli group, these centralizers equal the normalizers:
conjugation can only change a Pauli's sign, and a stabilizer
group cannot contain both signs. Finally,
\[
    (s'P')C=s'CP=CP=C(sP),
\]
so multiplying either endpoint by a stabilizer preserves
the relation.
\end{proof}

\begin{corollary}[Logical Pauli preservation]
\label{Cor: Logical dim Pres}
Let $C:Q\to Q'$ be a nonzero stabilizer operator on $n$
input and $n'$ output qubits. Order its phase-free endpoint
labels as input followed by output, and define
\[
    \begin{gathered}
    V_C\triangleq\bigl\{[p,p']:
       P(p')CP(p)=(-1)^\eta C
       \text{ for some }\eta\in\mathbb F_2\bigr\},\\
    V_{\mathrm{in}}\triangleq\{p:[p,0]\in V_C\},
    \qquad
    V_{\mathrm{out}}\triangleq\{p':[0,p']\in V_C\}.
    \end{gathered}
\]
Here $p\in\mathbb F_2^{2n}$ and
$p'\in\mathbb F_2^{2n'}$; the transpose in the stacked
notation changes signs but not these labels. Write
$r=\dim V_{\mathrm{in}}$ and $r'=\dim V_{\mathrm{out}}$.
The logical Pauli spaces are
\[
    \mathsf L_{\mathrm{in}}
       \triangleq V_{\mathrm{in}}^{\perp_s}/V_{\mathrm{in}},
    \qquad
    \mathsf L_{\mathrm{out}}
       \triangleq V_{\mathrm{out}}^{\perp_s}/V_{\mathrm{out}}.
\]
Their classes are $[p]_{\mathrm{in}}=p+V_{\mathrm{in}}$
and $[p']_{\mathrm{out}}=p'+V_{\mathrm{out}}$, with the
symplectic product inherited from representatives.
Then $n-r=n'-r'\triangleq k$, and $C$ induces a
well-defined linear bijection
\[
    \Phi_C:\mathsf L_{\mathrm{in}}\to\mathsf L_{\mathrm{out}},
    \qquad [p]_{\mathrm{in}}\mapsto[p']_{\mathrm{out}}
    \quad\text{when }[p,p']\in V_C.
\]
Both spaces have dimension $2k$, and the bijection preserves
the symplectic product:
\begin{equation}
    \bigl\langle\Phi_C([p]_{\mathrm{in}}),
                 \Phi_C([q]_{\mathrm{in}})\bigr\rangle_s
    =\langle[p]_{\mathrm{in}},[q]_{\mathrm{in}}\rangle_s.
    \label{eq: Bijection logical symplectic preservation}
\end{equation}
\end{corollary}
\begin{proof}
The $n+n'$ independent stabilizers of $\vket C$ give
$\dim V_C=n+n'$. Their commutativity implies
\begin{equation}
    \langle p,q\rangle_s+\langle p',q'\rangle_s=0
    \qquad([p,p'],[q,q']\in V_C).
    \label{eq:logical-pair-isotropy}
\end{equation}
In particular, $V_{\mathrm{in}}$ and
$V_{\mathrm{out}}$ are isotropic. Let
$\pi_{\mathrm{in}}([p,p'])=p$ and
$\pi_{\mathrm{out}}([p,p'])=p'$.
Their kernels have dimensions $r'$ and $r$, respectively.
Every projected input label commutes with
$V_{\mathrm{in}}$, and every projected output label
commutes with $V_{\mathrm{out}}$. Rank--nullity and
nondegeneracy of the ambient symplectic forms therefore give
\begin{align}
    n+n'-r'&\leq 2n-r,
        \label{eq: in ineq}\\
    n+n'-r&\leq 2n'-r'.
        \label{eq: out ineq}
\end{align}
Together these imply $n-r=n'-r'=k$, so both inequalities
are equalities. Hence
\[
    \operatorname{im}\pi_{\mathrm{in}}
       =V_{\mathrm{in}}^{\perp_s},
    \qquad
    \operatorname{im}\pi_{\mathrm{out}}
       =V_{\mathrm{out}}^{\perp_s},
\]
and both logical quotients have dimension $2k$.

For any input logical representative $p$, the first equality
supplies an output partner $p'$. If $[p,p'],[q,q']\in V_C$
and $p+q\in V_{\mathrm{in}}$, then
\[
    [0,p'+q']=[p,p']+[q,q']+[p+q,0]\in V_C.
\]
Thus $p'+q'\in V_{\mathrm{out}}$, proving that
$\Phi_C$ is independent of all choices of representatives.
Linearity follows by adding pairs in $V_C$; the second
projection equality gives surjectivity, hence bijectivity.

The quotient symplectic forms are well-defined because
adding a stabilizer label to a representative does not
change its pairing with any commuting label. They are
nondegenerate since
$(V_{\mathrm{in}}^{\perp_s})^{\perp_s}=V_{\mathrm{in}}$,
and similarly at the output. Finally,
\eqref{eq:logical-pair-isotropy} gives
$\langle p,q\rangle_s=\langle p',q'\rangle_s$, proving
\eqref{eq: Bijection logical symplectic preservation}.
\end{proof}

\begin{remark}[Logical action of a stabilizer circuit]
The bijection $\Phi_C$ is intrinsic to $C$, whereas its
matrix depends on the input and output logical Pauli bases.
For ordered symplectic bases
$\mathcal B_{\mathrm{in}}=([g_i^{\mathrm{in}}]_{\mathrm{in}})_{i=1}^{2k}$
and
$\mathcal B_{\mathrm{out}}=([g_i^{\mathrm{out}}]_{\mathrm{out}})_{i=1}^{2k}$,
our row-vector convention defines $F_{\mathrm{log}}$ by
\[
    \Phi_C([g_i^{\mathrm{in}}]_{\mathrm{in}})
    =\sum_{j=1}^{2k}(F_{\mathrm{log}})_{ij}
         [g_j^{\mathrm{out}}]_{\mathrm{out}}.
\]
This symplectic matrix specifies the \emph{phase-free}
logical action. The signs of the preserved pairs additionally
determine the signed logical Pauli action.

For example, $C=H$ exchanges the labels $x,z$ of $X,Z$.
Using $(x,z)$ at both boundaries gives
\[
    F_{\mathrm{log}}=\begin{pmatrix}0&1\\1&0\end{pmatrix}.
\]
Keeping the input basis but choosing $(z,x)$ at the output
gives $F_{\mathrm{log}}=I_2$. This is identity relative to
the new bases; the physical gate and $\Phi_H$ are unchanged.

Compatible Clifford encoders $U_{\mathrm{in}}$ and
$U_{\mathrm{out}}$ can fix these frames by mapping
$|0^r\rangle\otimes Q_{\mathrm L}$ and
$|0^{r'}\rangle\otimes Q_{\mathrm L}$ onto the input and
output stabilizer codes. The logical operator is then
\[
    C_{\mathrm{log}}
    =(\langle0^{r'}|\otimes I)
       U_{\mathrm{out}}^\dagger C U_{\mathrm{in}}
       (|0^r\rangle\otimes I).
\]
It is proportional to a Clifford unitary on $k$ qubits,
as follows from the bipartite stabilizer normal form
\cite[Sec.~2.4]{kliuchnikov2023stabilizercircuitverification}.
The boundary Cliffords may serve only to define the
logical frames; they need not be applied physically
\cite{beverland2024faulttolerancestabilizerchannels}.
\end{remark}

A circuit component may itself be a sequence of simpler
stabilizer operations. We now show how its stabilizers
follow from those of its components, so the same
construction can be applied at any level of abstraction.

\begin{definition}[Parallel and sequential composition]
For stabilizer circuit operators
$C_{o_{t,j}}:Q_{t,j}\to A_{t,j}\otimes Q'_{t,j}$,
$j\in\mathbb Z^+_{q}$, parallel composition is their tensor product:
\[
    \vket{C_{o_t}}
    =\vket{C_{o_{t,q}}}\otimes\cdots\otimes\vket{C_{o_{t,1}}}
    =\begin{quantikz}
       \setwiretype{b}&\gate{C_{o_{t,1}}}&\\
       \setwiretype{b}&\gate{C_{o_{t,2}}}&\\
       \setwiretype{n}&\push{\vdots}&\\
       \setwiretype{b}&\gate{C_{o_{t,q}}}&
      \end{quantikz}.
\]
For $C_{o_r}:Q_{r-1}\to A_r\otimes Q_r$,
$r\in\mathbb Z^+_{\ell}$, sequential composition connects each
output data register to the next input. Identities on
retained records are understood in the operator product
$C_o=C_{o_\ell}\cdots C_{o_1}$. Equivalently,
\[
    \begin{aligned}
    \vket{C_o}
      &=\mathbf I\bigl(
          \vket{C_{o_\ell}}\otimes\cdots\otimes\vket{C_{o_1}}
        \bigr)\\
      &=\begin{quantikz}
        \setwiretype{b}&\gate{C_{o_1}}&\gate{C_{o_2}}
        &\push{\cdots}&\gate{C_{o_\ell}}&
        \end{quantikz},
    \end{aligned}
\]
where
\[
    \mathbf I
    =I_{A_{1:\ell}}\otimes I_{Q_\ell}
       \otimes\Bigl(\bigotimes_{j=\ell-1}^{1}
                      \vbra{I_{Q_j}}\Bigr)\otimes I_{Q_0}.
\]
\end{definition}

\begin{theorem}[Stabilizer composition rule]
\label{prop: stabilizers of composition}
Let $C_{o_r}:Q_{r-1}\to A_r\otimes Q_r$ be nonzero
stabilizer operators, with $n_{r-1},n_r,n_{A_r}$ qubits
in the respective registers. Define
\[
    \vket{C_o^{\mathrm{par}}}
      =\vket{C_{o_\ell}}\otimes\cdots\otimes\vket{C_{o_1}},
    \qquad
    \vket{C_o^{\mathrm{seq}}}
      =\mathbf I\vket{C_o^{\mathrm{par}}}.
\]
The parallel stabilizer group is
\[
    \mathcal S_{\vket{C_o^{\mathrm{par}}}}
       =\mathcal S_{\vket{C_{o_\ell}}}\otimes\cdots
          \otimes\mathcal S_{\vket{C_{o_1}}}.
\]
If $C_o^{\mathrm{seq}}\neq0$, its full stabilizer group
consists of the operators
\[
    (-1)^{\sum_{j=1}^{\ell-1}\lambda_j}
    \stk{(R_\ell\otimes\cdots\otimes R_1)
                  \otimes P_\ell^{\mathrm{out}}}
        {P_1^{\mathrm{in}}},
\]
obtained from local stabilizers
\[
    \stk{R_r\otimes P_r^{\mathrm{out}}}
        {P_r^{\mathrm{in}}}
       \in\mathcal S_{\vket{C_{o_r}}},
       \qquad r\in\mathbb Z^+_{\ell},
\]
whose internal factors match up to sign:
\[
    P_j^{\mathrm{out}}P_{j+1}^{\mathrm{in}}
       =(-1)^{\lambda_j}I_{Q_j},
       \qquad \lambda_j\in\mathbb F_2,
       \quad j\in\mathbb Z^+_{\ell-1}.
\]
Here $R_r,P_r^{\mathrm{in}},P_r^{\mathrm{out}}$ are
Hermitian Paulis on $A_r,Q_{r-1},Q_r$, respectively.
\end{theorem}

\begin{remark}
    This is the operator-matching construction used in
stabilizer tensor networks
\cite{cao2022quantum,shen2025quantum}.
We give the proof with the transposes and signs required
by our circuit convention.
\end{remark}

\begin{proof}
Set $\ket\Psi=\vket{C_o^{\mathrm{par}}}$,
$N_r=n_{A_r}+n_{r-1}+n_r$, and $N=\sum_rN_r$.
The stabilizer generators of each $\mathcal{S}_{\vket{C_{o_r}}}$ act on disjoint tensor factors,
so together they give $N$ independent stabilizers. 
Thus
$\mathcal S\triangleq\mathcal S_{\ket\Psi}
=\bigotimes_{r=\ell}^{1}\mathcal S_{\vket{C_{o_r}}}$,
proving the parallel rule, and
\begin{equation}
    \ket\Psi\bra\Psi
       =\frac{\nu}{2^N}\sum_{s\in\mathcal S}s,
       \qquad \nu=\langle\Psi|\Psi\rangle.
    \label{eq:composition-parallel-expansion}
\end{equation}

Now let $\ket\psi=\mathbf I\ket\Psi\neq0$.
Write $N'=\sum_{j=1}^{\ell-1}n_j$ and $N''=N-2N'$.
For $s=s_\ell\otimes\cdots\otimes s_1\in\mathcal S$,
let
\[
    B(s)=\stk{(R_\ell\otimes\cdots\otimes R_1)
                    \otimes P_\ell^{\mathrm{out}}}
                  {P_1^{\mathrm{in}}}.
\]
Using \Cref{lemma: Vectorization identities} we get,
\begin{equation}
    \mathbf I s\mathbf I^\dagger
       =B(s)\prod_{j=1}^{\ell-1}
          \operatorname{Tr}
           (P_j^{\mathrm{out}}P_{j+1}^{\mathrm{in}}).
    \label{eq:composition-pauli-contraction}
\end{equation}
Pauli trace orthogonality makes this nonzero exactly when
\begin{equation}
    P_j^{\mathrm{out}}P_{j+1}^{\mathrm{in}}
       =(-1)^{\lambda_j}I_{Q_j}
       \qquad(j\in\mathbb Z^+_{\ell-1}).
    \label{eq:composition-matching}
\end{equation}
Let $\mathcal T$ be the set of matching elements and put
\[
    \theta(t)=(-1)^{\sum_j\lambda_j}B(t)
             =2^{-N'}\mathbf I t\mathbf I^\dagger.
\]
This definition is independent of how signs are distributed
among the local factors. The matching condition also gives
\[
    \vbra{I_{Q_j}}
      (P_j^{\mathrm{out}}\otimes
        (P_{j+1}^{\mathrm{in}})^{\mathsf T})
      =(-1)^{\lambda_j}\vbra{I_{Q_j}}.
\]
Applying it on every internal wire yields
\begin{equation}
    \mathbf I t=\theta(t)\mathbf I,
    \qquad \theta(t)\ket\psi=\ket\psi.
    \label{eq:composition-matching-intertwining}
\end{equation}
Thus every operator in
$\mathcal G=\{\theta(t):t\in\mathcal T\}$ is a
stabilizer of the contracted vector.

To prove completeness, contract
\eqref{eq:composition-parallel-expansion} as $\mathbf I \ketbra{\Psi}{\Psi}\mathbf I^\dagger$ and normalize to obtain
\begin{equation}
    \rho\triangleq\frac{\ket\psi\bra\psi}{\langle\psi|\psi\rangle}
      =\frac{\nu}{2^{N-N'}\langle\psi|\psi\rangle}
          \sum_{t\in\mathcal T}\theta(t).
    \label{eq:composition-contracted-expansion}
\end{equation}
After collecting repeated operators, write
$\rho=\sum_{g\in\mathcal G}c_g g$.
If a Hermitian Pauli $h$ stabilizes $\ket\psi$, then
$\operatorname{Tr}(h\rho)=1$. If $h\notin\mathcal G$,
neither does $-h$, because all elements of $\mathcal G$
fix the same nonzero vector. Pauli trace orthogonality
would then give $\operatorname{Tr}(h\rho)=0$.
Hence $\mathcal G$ contains every Pauli stabilizer.
Such stabilizers commute and form a group: anticommutation
would force $\ket\psi=-\ket\psi$, and products also
fix $\ket\psi$.

For each $g\in\mathcal G$, orthogonality now gives
$1=\operatorname{Tr}(g\rho)=2^{N''}c_g$.
Therefore
\[
    \rho=2^{-N''}\sum_{g\in\mathcal G}g,
    \qquad
    1=\operatorname{Tr}(\rho^2)=2^{-N''}|\mathcal G|.
\]
Thus $|\mathcal G|=2^{N''}$, so the contracted vector is
proportional to a pure stabilizer state with precisely
the claimed group.
\end{proof}

The composition rule determines the stabilizers of each
branch. For a non-adaptive circuit, we can collect the
branch-dependent signs into the retained measurement
record and work with one common relation space. Its
input, output and record support then distinguishes
observables, output checks and detectors. The following proposition establishes this formally.

\begin{definition}[Common stabilizer relations]
\label{def:common-stabilizer-relations}
Let $\mathbf C$ be a non-adaptive stabilizer with branch operators 
\[
    C_o=|o\rangle_A\otimes K_o:Q\longrightarrow A\otimes Q',
    \qquad
    \Omega_{\mathbf C}=\{o:K_o\neq0\},
\]
where $Q\cong\mathbb C^{2^n}$,
$Q'\cong\mathbb C^{2^{n'}}$, and
$A\cong(\mathbb C^2)^{\otimes\ell_M}$.
Define the common stabilizer group
\begin{equation}
    G_{\mathbf C}\triangleq
    \bigcap_{o\in\Omega_{\mathbf C}}
        \mathcal S_{\vket{C_o}}.
    \label{eq:outcome-independent-stabilizer-group}
\end{equation}
For $b\in\mathbb F_2^{\ell_M}$,
$p\in\mathbb F_2^{2n}$, and
$p'\in\mathbb F_2^{2n'}$, put $Z_A^b=P(0,b)_A$.
The \emph{common relation space} $\mathcal W_{\mathbf C}$
is defined to contain triples $\omega=[b,p,p']$ admitting a unique \footnote{Since $\Omega_{\mathbf C}\neq\varnothing$,
$G_{\mathbf C}$ is abelian and excludes $-I$.
Its record factors are $Z$-type, since they must preserve
each nonzero branch's record $|o\rangle$.
Thus $\mathcal W_{\mathbf C}$ is a binary linear space,
and each $c_\omega$ is unique.}
$c_\omega\in\mathbb F_2$ such that
\begin{equation}
    g(\omega)\triangleq
    (-1)^{c_\omega}
    \stk{Z_A^b\otimes P(p')}{P(p)}
    \in G_{\mathbf C}.
    \label{eq:common-stabilizer-relation}
\end{equation}
Define
\begin{equation}
    \begin{aligned}
        \mathcal W_{\mathrm{det}}
        &\triangleq\{[b,0,0]\in\mathcal W_{\mathbf C}\},\\
        \mathcal W_{\mathrm{obs}}
        &\triangleq\{[b,p,0]\in\mathcal W_{\mathbf C}\},\\
        \mathcal W_{\mathrm{out}}
        &\triangleq\{[b,0,p']\in\mathcal W_{\mathbf C}\}.
    \end{aligned}
    \label{eq:common-relation-subspaces}
\end{equation}
These subspaces satisfy
$\mathcal W_{\mathrm{obs}}\cap\mathcal W_{\mathrm{out}}
=\mathcal W_{\mathrm{det}}$.
We define the \emph{preserved-pair quotient} space as
\begin{equation}
    \mathcal L_{\mathrm{POP}}\triangleq
    \mathcal W_{\mathbf C}/
    (\mathcal W_{\mathrm{obs}}+\mathcal W_{\mathrm{out}}).
    \label{eq:common-preserved-pair-quotient}
\end{equation}
\end{definition}

\begin{proposition}[Detectors, observables, and output checks]
\label{prop:common-relation-extraction}
For the circuit of \Cref{def:common-stabilizer-relations},
every $\omega=[b,p,p']\in\mathcal W_{\mathbf C}$ satisfies
\begin{equation}
    (Z_A^b\otimes P(p'))C_o
    =(-1)^{c_\omega}C_oP(p)
    \qquad(\forall o).
    \label{eq:common-operational-relation}
\end{equation}
The detectors and observable representatives are
\begin{equation}
    \begin{aligned}
        \mathcal D_{\mathbf C}
        &=\{b:[b,0,0]\in\mathcal W_{\mathbf C}\},\\
        \mathcal O_{\mathbf C}(P(p))
        &=\{b:[b,p,0]\in\mathcal W_{\mathbf C}\},
    \end{aligned}
    \label{eq:detector-observable-extraction}
\end{equation}
with parities $c_{[b,0,0]}$ and $c_{[b,p,0]}$, respectively. Moreover, the realizable outcome set is given by
\begin{equation}
    \Omega_{\mathbf C}
    =
    \left\{
        o\in\mathbb F_2^{\ell_M}:
        o\cdot d=c_{[d,0,0]}
        \ \text{for every }d\in\mathcal D_{\mathbf C}
    \right\}.
    \label{eq:common-realisable-outcomes}
\end{equation}
\end{proposition}

\begin{proof}
\Cref{eq:common-operational-relation} follows from \Cref{lemma: Vectorization identities}.
For $p=p'=0$, this is exactly the detector condition.
For $p'=0$, it is equivalent to
\[
    C_o\Pi_{P(p)}^{\,1+c_\omega+o\cdot b}=0
    \qquad(\forall o),
\]
which is the observable condition. 

Non-adaptivity by \Cref{def:common-stabilizer-relations} gives a Clifford isometry $V$, formed
from the fixed gates and preparations, such that
$C_o=(\ketbra{o}{o}_A\otimes I)V$.
Since $\sum_o\ketbra{o}{o}_A=I_A$, we have
$B=\sum_o C_o=V$.
Thus $|\psi\rangle=2^{-n/2}\vket{B}$ is a normalized
stabilizer state.

By orthogonality of the records,
$(-1)^aZ_A^d\otimes I$ stabilizes $|\psi\rangle$
exactly when $d$ is a detector of parity $a$.
For a non-detector, $Z_A^d\otimes I$ anticommutes
with some stabilizer of $|\psi\rangle$ and therefore
has zero expectation.
Expanding the record projector gives
\[
\begin{aligned}
    2^{-n}\operatorname{Tr}(C_o^\dagger C_o)
    &=\langle\psi|
      (\ketbra{o}{o}_A\otimes I)|\psi\rangle\\
    &=2^{-\ell_M}
      \sum_{d\in\mathbb F_2^{\ell_M}}
      (-1)^{o\cdot d}
      \langle\psi|(Z_A^d\otimes I)|\psi\rangle\\
    &=2^{-\ell_M}
      \sum_{d\in\mathcal D_{\mathbf C}}
      (-1)^{c_{[d,0,0]}+o\cdot d}.
\end{aligned}
\]
Linearity of detector parities makes the last sum
equal to $|\mathcal D_{\mathbf C}|$ when all parities
hold and zero otherwise by pairwise cancellation.
The left-hand side is positive exactly when $C_o\neq0$,
proving \eqref{eq:common-realisable-outcomes}.
\end{proof}

\begin{proposition}[Outcome independence and preserved logical information]
\label{prop:outcome-independent-stabilizer-relations}
For the circuit of \Cref{def:common-stabilizer-relations},
the endpoint spaces $V_{K_o}$ of
\Cref{Cor: Logical dim Pres} coincide for all
$o\in\Omega_{\mathbf C}$.
Denote their common value by $V$, and put
\[
    V_{\mathrm{in}}\triangleq\{p:[p,0]\in V\},
    \qquad
    V_{\mathrm{out}}\triangleq\{p':[0,p']\in V\}.
\]
It then holds
\begin{equation}
    \begin{aligned}
        \mathcal W_{\mathbf C}/\mathcal W_{\mathrm{det}}
        &\cong V,\\
        \mathcal W_{\mathrm{obs}}/\mathcal W_{\mathrm{det}}
        &\cong V_{\mathrm{in}},\\
        \mathcal W_{\mathrm{out}}/\mathcal W_{\mathrm{det}}
        &\cong V_{\mathrm{out}}.
    \end{aligned}
    \label{eq:common-endpoint-identifications}
\end{equation}
The phase-free logical bijections $\Phi_{K_o}$ in \Cref{Cor: Logical dim Pres} also
coincide for all $o \in \Omega_{\mathbf C}$. Each nonzero branch preserves $k$ logical
qubits, where
\begin{equation}
    \begin{gathered}
        k=n-\dim V_{\mathrm{in}}
          =n'-\dim V_{\mathrm{out}},\\
        \dim\mathcal L_{\mathrm{POP}}=2k.
    \end{gathered}
    \label{eq:common-preserved-pair-dimension}
\end{equation}
\end{proposition}

\begin{proof}
We begin by noting down the stabilizer generators of the possible circuit components. For a fixed Clifford gate $U$, the stabilizer generators
of $\vket U$ are $\{\stk{UPU^\dagger}{P}\}_{P \in \mathcal{P}(\mathrm{dom} U)}$; identity wires correspond to $U=I$.
An $m$-qubit preparation $|0\rangle^{\otimes m}$ has
generators $Z_j$, $j\in\mathbb Z^+_{m}$.
For a $Z$ measurement, $\vket{\Pi_t^Z}$ has generators
\[
    (-1)^tZ\otimes I,
    \qquad
    (-1)^tI\otimes Z.
\]
Thus all local binary label spaces are fixed; only
the measurement signs depend on the outcomes.
By \Cref{prop: stabilizers of composition}, matching
equal binary Pauli labels on the shared registers
gives all branch endpoint relations.
The available labels and matching conditions are
outcome-independent, so all $V_{K_o}$ coincide.

For any $[p,p']\in V$, fix matching local stabilizers.
Their measurement signs multiply to $(-1)^{o\cdot b}$,
while the remaining phases give a fixed sign $(-1)^c$.
Hence
\[
    P(p')K_oP(p)=(-1)^{c+o\cdot b}K_o
    \qquad(o\in\Omega_{\mathbf C}),
\]
for some $c\in\mathbb F_2$ and
$b\in\mathbb F_2^{\ell_M}$.
Since $Z_A^b|o\rangle=(-1)^{o\cdot b}|o\rangle$,
this gives $[b,p,p']\in\mathcal W_{\mathbf C}$.

For $w,w' \in \mathcal W_{\mathbf{C}}$ with same $p,p'$, we have $w- w' = [b+b',0,0]\in\mathcal W_{\mathrm{det}}$.
Conversely, by definition for arbitray $d \in \mathcal{W}_{\operatorname{det}}$, $w + d , w'+d$ have the same endpoint labels $p,p'$ by definition. Thus, the cosets of $\mathcal W_{\mathrm{det}}$
correspond bijectively and linearly to the pairs in $V$.
Restricting to $p'=0$ or $p=0$ gives the other two
identifications in \eqref{eq:common-endpoint-identifications}.

The logical bijection in \Cref{Cor: Logical dim Pres}
is determined by $V$, so all $\Phi_{K_o}$ coincide.
From  \Cref{Cor: Logical dim Pres} and its proof it then follows that
$k=n-\dim V_{\mathrm{in}}=n'-\dim V_{\mathrm{out}}$
and $\dim V=n+n'$.
Finally, using \eqref{eq:common-endpoint-identifications}
and
$\mathcal W_{\mathrm{obs}}\cap\mathcal W_{\mathrm{out}}
=\mathcal W_{\mathrm{det}}$,
\[
    \begin{aligned}
        \dim\mathcal L_{\mathrm{POP}}
        &=\dim\mathcal W_{\mathbf C}
          -\dim(\mathcal W_{\mathrm{obs}}
                +\mathcal W_{\mathrm{out}})\\
        &=\dim V-\dim V_{\mathrm{in}}
                 -\dim V_{\mathrm{out}}\\
        &=2k.
    \end{aligned}
\]
\end{proof}

\begin{definition}[Encoding circuit for {$[[n,k,d]]$} stabilizer code]
Let $\mathcal{S}\subseteq\mathcal{P}(Q)$ be the stabilizer group of an $[[n,k,d]]$ stabilizer code with codespace $\mathsf{T}(\mathcal{S})\subseteq Q$. Let $Q_{\L}\cong\mathbb{C}^{2^k}$ be a $k$-qubit logical input register and let $A=A_1\otimes\cdots\otimes A_{n-k}\cong\mathbb{C}^{2^{n-k}}$ be an $(n-k)$-qubit ancilla register, where $A_i\cong\mathbb{C}^2$, such that $Q\cong A\otimes Q_{\L}\cong\mathbb{C}^{2^n}$. An \emph{encoding circuit} $\mathbf{C}_{\mathrm{enc}}$ for $\mathcal{S}$ is a stabilizer circuit defined by a sequence of two linear operators
\begin{equation}
\resizebox{0.98\columnwidth}{!}{$
\mathbf{C}_{\mathrm{enc}}
\triangleq
\left\{C^{(1)},C^{(2)}\right\},
\qquad
C^{(1)}:Q_{\L}\longrightarrow Q,
\qquad
C^{(2)}:Q\longrightarrow Q,
$}
\end{equation}
where the preparation component and Clifford component are respectively given by $C^{(1)}
=
\ket{\mathbf s}_{A}\otimes I_{Q_{\mathrm{L}}}$ and $C^{(2)}
=
U_{\mathrm{enc}}$, for a Clifford unitary $U_{\mathrm{enc}}\in\mathrm{Cliff}(Q)$
and a fixed string $\mathbf s=[s_1,\ldots,s_{n-k}]\in\mathbb F_2^{n-k}$. 
{We require $U_{\mathrm{enc}}$ and $\mathbf s$ to satisfy $\mathcal S=\left\langle(-1)^{s_i}U_{\mathrm{enc}}Z_{A_i}U_{\mathrm{enc}}^\dagger:i\in\mathbb Z^+_{n-k}\right\rangle$.}
The circuit operator $C_{\mathrm{enc}}: Q_\L \to Q$ can thus be decomposed as 
\begin{equation}
\resizebox{0.98\columnwidth}{!}{$
\begin{aligned}
C_{\mathrm{enc}}
=
C^{(2)}C^{(1)}
=
U_{\mathrm{enc}}
\left(
    \ket{\mathbf{s}}_{A}\otimes I_{Q_{\L}}
\right) = \begin{quantikz}[
    row sep=0.9cm,
    column sep=0.5cm
]
\setwiretype{n}
    & \phase[
          label style={
              label position=above,
              anchor=south,
              yshift=-1pt
          }
      ]{\text{\scriptsize fixed bits}}
    & \rstick{$\mathbf{s}$}\setwiretype{c}
    \\
\setwiretype{n}
    & \lstick{$\ket{\mathbf{s}}_{A}$\;}
      \gategroup[
          2,
          steps=2,
          style={
              dashed,
              rounded corners,
              fill=blue!05,
              inner xsep=1pt
          },
          background,
          label style={
              label position=above,
              anchor=north,
              yshift=0.4cm
          }
      ]{$C_{\mathrm{enc}}$}
    & \gate[2]{U_{\mathrm{enc}}}\setwiretype{q}
    & \rstick[2]{$Q$}
    \\
\setwiretype{n}
    & \lstick{$Q_{\L}$\;}
    & \setwiretype{b}
    &
\end{quantikz}.
\end{aligned}
$}
\end{equation}
\end{definition}

\begin{remark}
For fixed $U_{\mathrm{enc}}$, changing $\mathbf s$ changes
the signs of the encoded stabilizers and hence selects a
different syndrome sector, rather than the same signed code.
A known choice of $\mathbf s$ can be absorbed into the
encoder by replacing $U_{\mathrm{enc}}$ with
$U_{\mathrm{enc}}(X_A^{\mathbf s}\otimes I)$.
Thus we take $\mathbf s=0$ from now on, without restricting
the encoded stabilizer code or its logical frame.
\end{remark}

\begin{proposition}
    The encoding circuit satisfies $\operatorname{im}(C_{\mathrm{enc}})=\mathsf{T}(\mathcal{S})$, or equivalently, $\mathcal{S}= \left\langle(-1)^{s_i}U_{\mathrm{enc}}Z_{A_i}U_{\mathrm{enc}}^{\dagger}:i\in\mathbb Z^+_{n-k}\right\rangle$. Decompose $Q_\L = Q_{\L,1} \otimes \cdots \otimes Q_{\L,k}$ where for $i \in \mathbb Z^+_{k}$ each $Q_{\L,i} \cong \C^2$. 
We then have the $2k$ logical operator representatives as $\{\bar{X}_i \triangleq U_{\mathrm{enc}}X_{Q_{\L,i}}U^\dagger_{\mathrm{enc}}\}_{i \in \mathbb Z^+_{k}}$ and $ \{\bar{Z}_j \triangleq U_{\mathrm{enc}}Z_{Q_{\L,j}}U^\dagger_{\mathrm{enc}}\}_{j \in \mathbb Z^+_{k}}$.
\end{proposition}

\begin{proof}
The proof is given in
\cite[Secs.~4.1--4.2 and~5.3]{gottesman1997stabilizer}.
\end{proof}

\begin{remark}[Code conditioning]
\label{rem:code-conditioning-common-relations}
For the code-conditioned circuit Kraus operators $\{\widetilde C_o=C_oC_{\mathrm{enc}}\}_o$,
$C_o=0$ implies $\widetilde C_o=0$, so conditioning can only
remove realizable outcomes. Any detector equation $o\cdot d=a_d$
valid on $\Omega_{\mathbf C}$ remains valid on the restricted
outcome set. Hence
\[
    \varnothing\neq
    \Omega_{\widetilde{\mathbf C}}\subseteq\Omega_{\mathbf C},
    \qquad
    \mathcal D_{\mathbf C}\subseteq
    \mathcal D_{\widetilde{\mathbf C}},
\]
with unchanged detector parities. Nonemptiness follows from
trace preservation and the isometric encoder. In particular, let $S$ be the input code
stabilizer group. If $v\in\mathcal O_{\mathbf C}(s)$ has parity
$a_v$ for some $s\in S$, then $sC_{\mathrm{enc}}=C_{\mathrm{enc}}$
gives
\[
    (Z_A^v\otimes I_{Q'})\widetilde C_o
    =(-1)^{a_v}\widetilde C_o.
\]
Thus an observable representative of an input code stabilizer
becomes a detector with the same parity. If $\mathbf C$ is non-adaptive, so is the conditioned
circuit, and the extraction rules of
\Cref{prop:outcome-independent-stabilizer-relations} apply
directly to $\widetilde{\mathbf C}$.
\end{remark}
At level $0$, the Pauli relations are obtained by propagation
through the encoder and Clifford gates. The following
proposition expresses this as a binary linear system;
restricting its endpoint labels then gives detectors,
observable representatives and output checks.

\begin{proposition}[Code-conditioned common relations space at level $0$]
\label{prop: Code conditioned stabilizers L0}
Let $U=U_N\cdots U_1:A\otimes Q\to A'\otimes Q'$
be a fixed Clifford circuit on $m+n=m'+n'$ qubits,
where $A,Q,A',Q'$ contain $m,n,m',n'$ qubits, respectively.
The gates $U_i$ act on this common total register; its
input and output partitions may differ. For the previously defined encoder
$C_{\mathrm{enc}}:Q_\L\to Q$, with $\dim Q_\L=2^k$, write
\[
    \begin{aligned}
        C_o&=(|o\rangle\langle o|_{A'}\otimes I_{Q'})
             U(|0^m\rangle_A\otimes I_Q),\\
        \widetilde C_o&=C_oC_{\mathrm{enc}},
        \qquad o\in\mathbb F_2^{m'}.
    \end{aligned}
\]
Let $F_U,F_{U_{\mathrm{enc}}}$ and $s_U,s_{U_{\mathrm{enc}}}$
be the symplectic matrices and sign maps of $U,U_{\mathrm{enc}}$,
respectively (see \Cref{Prop: Symplectic Clifford}). The common relations space $\mathcal W_{\widetilde{\mathbf C}}$
consists precisely of the labels $\omega=[b,p,p']$, with
$b\in\mathbb F_2^{m'}$, $p=[x_\L,z_\L]\in\mathbb F_2^{2k}$
and $p'=[x_{Q'},z_{Q'}]\in\mathbb F_2^{2n'}$, for which there exist
$z_A\in\mathbb F_2^m$, $z_{A_{\mathrm{enc}}}\in\mathbb F_2^{n-k}$
and $x_Q,z_Q\in\mathbb F_2^n$ satisfying
\begin{equation}
    \begin{aligned}
        [0,x_\L,z_{A_{\mathrm{enc}}},z_\L]F_{U_{\mathrm{enc}}}
        &=[x_Q,z_Q],\\
        [0,x_Q,z_A,z_Q]F_U
        &=[0,x_{Q'},b,z_{Q'}].
    \end{aligned}
    \label{eq: L0 code conditioned relation space}
\end{equation}
These witnesses $z_A,z_{A_{\mathrm{enc}}},x_Q,z_Q$ are unique for each $\omega$, and its common
stabilizer sign is
\begin{equation}
    \begin{aligned}
        c_\omega
        =s_{U_{\mathrm{enc}}}
          ([0,x_\L,z_{A_{\mathrm{enc}}},z_\L])+s_U([0,x_Q,z_A,z_Q]).
    \end{aligned}
    \label{eq: L0 code conditioned common sign}
\end{equation}
\end{proposition}

\begin{proof}
The relation space and sign formula follow directly
from \Cref{prop: stabilizers of composition}, using
the Clifford conjugation rule of
\Cref{Prop: Symplectic Clifford} and the $Z$ stabilizers
of the prepared ancillas.
The witnesses are unique because symplectic matrices are invertible and $F_U^{-1}$ determines
$x_Q,z_Q,z_A$, and $F_{U_{\mathrm{enc}}}^{-1}$
then determines $z_{A_{\mathrm{enc}}}$.
\end{proof}

\begin{example}[Single-qubit $Z$ measurement with two ancillas]
\label{example: Z measurement L0}
Take $C_{\mathrm{enc}}=I_Q$, so $n=k=n'=1$ and $m=m'=2$.
Order the output registers as $A_2,A_1,Q'$, and define
\[
    \begin{aligned}
        U&=\cnot_{A_1\to A_2}\cnot_{Q\to A_1},\\
        B&=U(|00\rangle_{A_2A_1}\otimes I_Q),\\
        C_o^Z&=(|o\rangle\langle o|_{A_2A_1}\otimes I_{Q'})B,
        \qquad o=[o_2,o_1].
    \end{aligned}
\]
Here $Q'$ is the data wire after the gates:
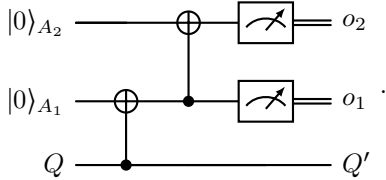
\[
\begin{quantikz}
\lstick{$|0\rangle_{A_2}$}
    & \qw & \targ{} & \meter{} & \rstick{$o_2$}\setwiretype{c}\\
\lstick{$|0\rangle_{A_1}$}
    & \targ{} & \ctrl{-1} & \meter{} & \rstick{$o_1$}\setwiretype{c}\\
\lstick{$Q$}
    & \ctrl{-1} & \qw & \qw & \rstick{$Q'$}
\end{quantikz}.
\]

The relevant symplectic action is
\begin{equation}
    \begin{aligned}
        &[0,0,x,u_2,u_1,z]F_U\\
        &\qquad=[x,x,x,u_2,u_2+u_1,u_1+z].
    \end{aligned}
    \label{eq: F action example Z 2Anc}
\end{equation}
Matching with $[0,0,x',b_2,b_1,z']$ as in
\Cref{prop: Code conditioned stabilizers L0} gives
\begin{equation}
    \begin{gathered}
        \mathcal W_{\mathbf C}
        =\bigl\{[b,[0,z],[0,z']]:
        b=[b_2,b_1]\in\mathbb F_2^2,\;
        z,z'\in\mathbb F_2,\\
        b_2+b_1+z+z'=0\bigr\}.
    \end{gathered}
    \label{eq: Z two ancillas common relations}
\end{equation}
The matching equations force $x=x'=0$, so the
relations involve only $I$ and $Z$ factors. The CNOT gates introduce no signs on these operators;
hence $c_\omega=0$ for every
$\omega\in\mathcal W_{\mathbf C}$.
Equivalently,
\[
    G_{\mathbf C}
    =
    \left\langle
        \stk{ZZI}{I},\;
        \stk{IZI}{Z},\;
        \stk{IZZ}{I}
    \right\rangle,
\]
with upper factors ordered as $A_2,A_1,Q'$ and lower
factor the input $Q$. By \Cref{prop:common-relation-extraction},
\[
    \begin{aligned}
        \mathcal D_{\mathbf C}
        &=\mathcal O_{\mathbf C}(I)
          =\operatorname{span}\{[1,1]\},\\
        \mathcal O_{\mathbf C}(Z)
        &=[0,1]+\mathcal D_{\mathbf C}
          =\{[0,1],[1,0]\},\\
        \mathcal O_{\mathbf C}(X)
        &=\mathcal O_{\mathbf C}(Y)=\varnothing,\\
        \Omega_{\mathbf C}
        &=\{[0,0],[1,1]\}.
    \end{aligned}
\]
All displayed parities are zero.
Thus $Z_Q$ is the sole nontrivial measured input Pauli,
up to sign. Either record bit reports its eigenvalue;
the two representatives differ by the detector $[1,1]$,
which checks that the bits agree.

Setting $p=0$ in \Cref{eq:common-operational-relation}
gives the common output checks generated by
$Z_{A_2}Z_{A_1}$ and $Z_{A_1}Z_{Q'}$.
The second is the single independent check beyond the
detector, namely a noiseless measurement of $Z_{Q'}$ by a decoder tests the outcome bit matches with $o_1$.
Here
\[
    \begin{aligned}
        \dim\mathcal W_{\mathbf C}&=3,
        &\dim\mathcal W_{\mathrm{det}}&=1,\\
        \dim\mathcal W_{\mathrm{obs}}&=2,
        &\dim\mathcal W_{\mathrm{out}}&=2.
    \end{aligned}
\]
Hence $\mathcal W_{\mathbf C}
=\mathcal W_{\mathrm{obs}}+\mathcal W_{\mathrm{out}}$.
By \Cref{prop:outcome-independent-stabilizer-relations},
$\mathcal L_{\mathrm{POP}}=\{0\}$ and each nonzero
branch preserves $k_{\mathrm{pres}}=1-(2-1)=0$
logical qubits.

Indeed, $B|t\rangle=|t,t,t\rangle$, so the nonzero
branches are
\[
    C_{[t,t]}^Z
    =|t,t\rangle_{A_2A_1}\otimes
      |t\rangle_{Q'}\langle t|_Q,
    \qquad t\in\mathbb F_2.
\]
The measured eigenvalue is $(-1)^t$, and its eigenstate
remains on $Q'$.
Moreover, $\vket{C_{[t,t]}^Z}=|t\rangle^{\otimes4}$,
whose full stabilizer group is generated by
$(-1)^tZ$ on each factor.

Finally, condition on the input code
$\operatorname{span}\{|0\rangle_Q\}$ using
$C_{\mathrm{enc}}:\mathbb C\to Q$,
$\alpha\mapsto\alpha|0\rangle_Q$.
Since $Z_QC_{\mathrm{enc}}=C_{\mathrm{enc}}$,
\Cref{rem:code-conditioning-common-relations}
turns both $Z_Q$ observable representatives into
detectors with parity zero. Thus
\[
    \mathcal D_{\widetilde{\mathbf C}}=\mathbb F_2^2,
    \qquad
    \Omega_{\widetilde{\mathbf C}}=\{[0,0]\}.
\]
\end{example}

At level $1$, we consider each component is a Pauli measurement. The following proposition gives explicit expressions for computing the common
stabilizer relations of non-adaptive level-$1$ circuits conditioned
on an input code. These expressions are derived for fixed qubit registers throughout time and determine the detectors, measured
observables and preserved logical operators. We illustrate their
application to a finite Hastings--Haah honeycomb Floquet-code
circuit in \Cref{app:honeycomb}.

\begin{definition}[Measurement schedule and spacetime commutation matrices]
\label{def:L1-schedule-matrices}
Fix an outcome-independent sequence of nonidentity
Paulis $M_r=P(\mathbf m_r)$ on $n$ data qubits with
$r\in\mathbb Z^+_{\ell}$.
Define the \emph{measurement schedule matrix}
$\mathbf M\in\mathbb F_2^{\ell\times2n}$,
its \emph{symplectic dual}
$\mathbf Q\in\mathbb F_2^{\ell\times2n}$,
and the \emph{spacetime commutation matrix}
$\mathbf T\in\mathbb F_2^{\ell\times\ell}$ by
\begin{equation}
    \begin{gathered}
        \mathbf M\triangleq
        \begin{bmatrix}
            \mathbf m_1\\
            \vdots\\
            \mathbf m_\ell
        \end{bmatrix},
        \qquad
        \mathbf Q\triangleq\mathbf M\Omega_n,\\
        \mathbf T_{rs}\triangleq
        \begin{cases}
            \langle\mathbf m_s,\mathbf m_r\rangle_s,&s<r,\\
            0,&s\ge r,
        \end{cases}
        \qquad r,s\in\mathbb Z^+_{\ell}.
    \end{gathered}
    \label{eq:L1-matrices}
\end{equation}
\end{definition}

\begin{proposition}[Code-conditioned common relation space at level $1$]
\label{prop: code-conditioned stabilizers level 1}
\label{prop: Measurement gadget 1}
For the schedule of \Cref{def:L1-schedule-matrices},
let $C_{\mathrm{enc}}:Q_\L\to Q_0$ encode $k$ qubits
into $n$ qubits, with ancillas prepared in
$|0^{n-k}\rangle$.
Write $F_{\mathrm{enc}},s_{\mathrm{enc}}$ for the
symplectic matrix and sign map of its encoding Clifford.
For $o\in\mathbb F_2^\ell$, set
\[
    \begin{aligned}
        \Pi_t^{M_r}&=\frac{I+(-1)^tM_r}{2},\\
        C_o&=|o\rangle_A\otimes
              \Pi_{o_\ell}^{M_\ell}\cdots\Pi_{o_1}^{M_1},\\
        \widetilde C_o&=C_oC_{\mathrm{enc}}.
    \end{aligned}
\]
A label $\omega=[b,p,p']$, with
$b\in\mathbb F_2^\ell$,
$p=[x_\L,z_\L]\in\mathbb F_2^{2k}$ and
$p'\in\mathbb F_2^{2n}$, belongs to
$\mathcal W_{\widetilde{\mathbf C}}$ exactly when
some $z_{A_{\mathrm{enc}}}\in\mathbb F_2^{n-k}$ satisfies
\begin{equation}
    \begin{aligned}
        \nu&=[0,x_\L,z_{A_{\mathrm{enc}}},z_\L]
              F_{\mathrm{enc}},\\
        \nu\mathbf Q^{\mathsf T}
            +b\mathbf T^{\mathsf T}&=0,\\
        p'&=\nu+b\mathbf M.
    \end{aligned}
    \label{eq: code-conditioned level-1 relation space}
\end{equation}
The witnesses $\nu,z_{A_{\mathrm{enc}}}$ are unique. Put
\begin{equation}
    \nu_r=\nu+\sum_{s<r}b_s\mathbf m_s,
    \qquad r\in\mathbb Z^+_{\ell+1}.
    \label{eq:L1-prefix-and-w}
\end{equation}
The middle equation in \Cref{eq: code-conditioned level-1 relation space} gives
$\langle\nu_r,\mathbf m_r\rangle_s=0$ for $r\in\mathbb Z^+_{\ell}$.
For $[P(u),P(v)] = 0$, define
$\gamma(u,v) \triangleq \kappa(u,v)/2\pmod2$, as in
\Cref{prop: Symplectic Pauli Algebra}.
Then the common stabilizer sign is
\begin{equation}
    \begin{gathered}
        c_\omega
        = s_{\mathrm{enc}}
          ([0,x_\L,z_{A_{\mathrm{enc}}},z_\L])
          +\Theta(\nu,b), \quad \mathrm{where} \\
          \Theta(\nu,b)
        \triangleq\sum_{r=1}^{\ell}
            b_r\gamma(\mathbf m_r,\nu_r).
    \end{gathered}
    \label{eq:L1-matrix-phase}
\end{equation}
\end{proposition}

\begin{proof}
The stabilizers of $\vket{\Pi_t^{M_r}}$ are generated
by $(-1)^t\stk{M_r}{I}$ and
$\stk{P(\mu)}{P(\mu)}$ for
$\mu\in\mathbb F_2^{2n}$ satisfying
$\langle\mu,\mathbf m_r\rangle_s=0$.
The retained record $|t\rangle_{A_r}$ contributes
the generator $(-1)^tZ_{A_r}$. By \Cref{prop: stabilizers of composition}, matching
these local Pauli labels gives
\[
    \nu_{r+1}=\nu_r+b_r\mathbf m_r,
    \qquad
    \langle\nu_r,\mathbf m_r\rangle_s=0.
\]
Together with the encoder relation, these are exactly
\eqref{eq: code-conditioned level-1 relation space}.
Adding the local sign exponents, including the encoder
sign, gives \eqref{eq:L1-matrix-phase} by
\Cref{prop: Symplectic Pauli Algebra}.
Finally, $\nu=p'+b\mathbf M$ and invertibility of
$F_{\mathrm{enc}}$ give uniqueness of the witnesses.
\end{proof}

\begin{corollary}[Detectors, observables, and logical qubits]
\label{cor:L1-extraction}
For $C_{\mathrm{enc}}=I$, the detectors and observable
representatives are
\begin{equation}
    \begin{aligned}
        \mathcal D_{\mathbf C}
        &=\ker[\,\mathbf T\ \mathbf M\,]
          =\ker[\,\mathbf T^{\mathsf T}\ \mathbf M\,],\\
        \mathcal O_{\mathbf C}(P(p))
        &=\{v:v\mathbf T=0,\ v\mathbf M=p\},
    \end{aligned}
    \label{eq:L1-unconditioned-kernels}
\end{equation}
with parities $\Theta(0,d)$ and $\Theta(p,v)$,
respectively. Each nonzero branch preserves
\begin{equation}
    k_{\mathrm{pres}}
    =n-\operatorname{rank}[\,\mathbf T\ \mathbf M\,]
       +\operatorname{rank}\mathbf T
    \label{eq:L1-unconditioned-ranks}
\end{equation}
logical qubits.
\end{corollary}

\begin{proof}
For $C_{\mathrm{enc}}=I$, we have $\nu=p$ and
$c_\omega=\Theta(p,b)$.
Using $\mathbf M\mathbf Q^{\mathsf T}
=\mathbf T+\mathbf T^{\mathsf T}$, the matching
conditions of \Cref{prop: stabilizers of composition} with $p'=0$ become
$p=v\mathbf M$ and $v\mathbf T=0$.
This gives the observable formula and, on setting
$p=0$, the detector formula, with the stated parities.
The two detector kernels coincide because
$v\mathbf M=0$ implies
$v\mathbf T=v\mathbf T^{\mathsf T}$. Finally,
\[
    \begin{aligned}
        \dim\mathcal W_{\mathrm{obs}}
        &=\ell-\operatorname{rank}\mathbf T,\\
        \dim\mathcal W_{\mathrm{det}}
        &=\ell-\operatorname{rank}[\,\mathbf T\ \mathbf M\,].
    \end{aligned}
\]
The logical-qubit count follows from
\Cref{prop:outcome-independent-stabilizer-relations}.
\end{proof}
\section{Fault tolerance of general quantum circuits}
\label{sec: Fault tolerance}
We now ask when the noisy realization of a circuit can be
corrected to reproduce its ideal action, including the
retained measurement record. In \Cref{sec: Gen iff FT}, we provide algebraic
necessary and sufficient conditions for exact recovery
under a prescribed noise model, allowing correlations
across time. These conditions recover the
Knill--Laflamme criterion and give an exact SDP test for
a fixed circuit. We then construct an SDP relaxation for
circuit synthesis.

\subsection{Necessary and sufficient conditions for fault tolerance}
\label{sec: Gen iff FT}

We first formalize the notion of fault tolerance for general quantum circuits and then establish necessary and sufficient conditions for achieving it.
Here fault tolerance means exact recoverability of the intended circuit operation including its measurement outcomes. A scalable threshold or composition with noisy recovery
requires additional hypotheses \cite{HeNguyenPattison2025Composable}.

\begin{definition}[Fault tolerance criterion]
\label{def: Fault tolerance criteria}
Let $\mathbf C$ be an $\ell$-time-step quantum circuit
with circuit map $\mathbf C_{\mathrm{cm}}$ and ideal
realization $\mathbf C_{\mathrm{id}}$, and let
$\mathbf E$ be a noise model satisfying
\Cref{def: Circuit map and noise model}.
Both the noisy and ideal realizations are channels from
$\L(Q_0)$ to $\L(\mathcal H_{\mathrm{out}})$, where
\[
    \mathcal H_{\mathrm{out}}
       \triangleq A_{1:\ell}^{\mathrm{out}}\otimes Q_\ell.
\]
The circuit is \emph{fault-tolerant} under $\mathbf E$
if there exists a CPTP recovery map
$\mathcal R:\L(\mathcal H_{\mathrm{out}})
\to\L(\mathcal H_{\mathrm{out}})$ such that
\begin{equation}
    \mathcal R\circ\mathbf E\circ\mathbf C_{\mathrm{cm}}
       =\mathbf C_{\mathrm{id}}.
    \label{eq: fault tolerance criteria}
\end{equation}
This is an equality of channels, required for every input
state. The recovery is applied noiselessly.
\end{definition}

\begin{remark}
The ideal outcome record is classical in the computational
basis. A known flip of a record bit can be corrected by
$X|o_j\rangle=|o_j\oplus1\rangle$, or by the corresponding
classical update if the record is stored classically.
\end{remark}

The following criterion is a Kraus-operator form of the
exact channel-recovery duality of
\cite[Theorem~1]{BenyOreshkov2010}, specialized to an ideal
circuit that retains its full outcome record.
A detailed treatment of this duality is given in
\cite{BenyOreshkov2011Simulation}.
We give a constructive proof in the circuit notation.

\begin{theorem}[Algebraic necessary and sufficient conditions for fault tolerance]
\label{Thm: Iff Fault tolerance}
For a circuit map $\mathbf C_{\mathrm{cm}}$ with Kraus operators
$\{\mathbf C_o\}_o$, let $\mathbf E$ be a noise
model with spacetime Kraus operators $\{\mathbf E_e\}_e$.
Then $\mathbf C$ is fault-tolerant under $\mathbf E$
if and only if there exist Hermitian matrices
$\{\Lambda^{(k)}\}_k$, indexed by the combined noisy
Kraus indices $(o,e)$, satisfying, for all $o,e,o',e'$,
\begin{equation}
    \begin{gathered}
    \mathbf C_{o'}^\dagger\mathbf E_{e'}^\dagger
          \mathbf E_e\mathbf C_o
       =\sum_k\lambda^{(k)}_{(o,e),(o',e')}
          \mathbf C_k^\dagger\mathbf I^\dagger
             \mathbf I\mathbf C_k,\\
    [\Lambda^{(k)}]_{(o,e),(o',e')}
       \triangleq\lambda^{(k)}_{(o,e),(o',e')},\\
    \Lambda^{(k)}\succeq0,\quad
    \operatorname{Tr}\Lambda^{(k)}=1
       \quad\text{for every }k.
    \label{eq: Gen Iff}
    \end{gathered}
\end{equation}
\end{theorem}

\begin{proof}
Write $\mathcal H_{\mathrm{out}}
=A^{\mathrm{out}}_{1:\ell}\otimes Q_\ell$ and set
\[
    K_k=\mathbf I\mathbf C_k,
    \qquad F_a=\mathbf E_e\mathbf C_o,
    \qquad a=(o,e).
\]
These are the ideal and noisy Kraus operators, respectively,
from $Q_0$ to $\mathcal H_{\mathrm{out}}$, with
\[
    \sum_kK_k^\dagger K_k
       =\sum_aF_a^\dagger F_a=I_{Q_0}.
\]
The equalities follow from the circuit and noise
normalization constraints in \Cref{app:causal-normalisation}.
Retaining the full record gives
$K_k=|k\rangle\otimes B_k$, where $B_k:Q_0\to Q_\ell$.
Consequently,
\begin{equation}
    K_{k'}^\dagger K_k
    =\delta_{k',k}K_k^\dagger K_k.
    \label{eq:ft-ideal-kraus-orthogonality}
\end{equation}
This holds for both adaptive and non-adaptive circuits.
In particular, the nonzero $K_k$ are Hilbert--Schmidt
orthogonal.

\emph{Necessity ($\implies$).}
Suppose a CPTP recovery with Kraus operators $\{R_\alpha\}_\alpha$
exists. Equality of the recovered and ideal channels gives
the Choi identity
\begin{equation}
    \sum_{\alpha,a}
    \vket{R_\alpha F_a}\vbra{R_\alpha F_a}
    =\sum_k\vket{K_k}\vbra{K_k}.
    \label{eq:ft-choi-equality}
\end{equation}
First consider only terms with $K_k \neq 0$. If $|w\rangle$ is orthogonal
to every $\vket{K_k}$, taking its expectation above yields
\[
    \sum_{\alpha,a}
    \bigl|\langle w\mid R_\alpha F_a\rangle\!\rangle\bigr|^2=0.
\]
Every summand is nonnegative, so each vanishes.
Thus every $R_\alpha F_a$ belongs to the span of the
nonzero $K_k$, and we may write
\begin{equation}
    R_\alpha F_a=\sum_kv_{\alpha,a,k}K_k.
    \label{eq:ft-kraus-expansion}
\end{equation}
For $d_k=\operatorname{Tr}(K_k^\dagger K_k)>0$,
Hilbert--Schmidt orthogonality gives
$\operatorname{Tr}(K_k^\dagger R_\alpha F_a)
=d_kv_{\alpha,a,k}$.
Taking the expectation of
\eqref{eq:ft-choi-equality} in $\vket{K_k}$ and dividing
by $d_k^2$ gives
\begin{equation}
    \sum_{\alpha,a}|v_{\alpha,a,k}|^2=1.
    \label{eq:ft-coefficient-normalization}
\end{equation}
Since $\sum_\alpha R_\alpha^\dagger R_\alpha=I$, using
\eqref{eq:ft-ideal-kraus-orthogonality} gives
\begin{equation}
    \begin{aligned}
        F_{a'}^\dagger F_a
        &=\sum_\alpha
          (R_\alpha F_{a'})^\dagger(R_\alpha F_a)\\
        &=\sum_k\lambda^{(k)}_{a,a'}K_k^\dagger K_k,
    \end{aligned}
    \label{eq:ft-necessary-identity}
\end{equation}
where
\begin{equation}
    \lambda^{(k)}_{a,a'}
    =\sum_\alpha
      v_{\alpha,a,k}\overline{v_{\alpha,a',k}}.
    \label{eq:ft-lambda-construction}
\end{equation}
For the column vectors $v_{\alpha,k}=(v_{\alpha,a,k})_a$,
\[
    \Lambda^{(k)}
    =\sum_\alpha v_{\alpha,k}v_{\alpha,k}^\dagger\succeq0,
    \qquad \operatorname{Tr}\Lambda^{(k)}=1,
\]
where the trace follows from
\eqref{eq:ft-coefficient-normalization}.
For $K_k=0$, any positive semidefinite $\Lambda^{(k)}$
of trace one suffices, since its operator contribution
vanishes.

\emph{Sufficiency ($\impliedby$).}
Conversely, suppose the stated matrices exist. Factor
$\Lambda^{(k)}=T^{(k)}T^{(k)\dagger}$ and write
$t_{a,j}^{(k)}=[T^{(k)}]_{a,j}$. Then
\begin{equation}
    \begin{aligned}
        \lambda^{(k)}_{a,a'}
        &=\sum_jt_{a,j}^{(k)}\overline{t_{a',j}^{(k)}},\\
        \sum_{a,j}|t_{a,j}^{(k)}|^2
        &=\operatorname{Tr}\Lambda^{(k)}=1.
    \end{aligned}
    \label{eq:ft-positive-factorization}
\end{equation}
Let $\mathcal F$ have orthonormal basis $\{|k,j\rangle\}$
and define $G_a:Q_0\to\mathcal H_{\mathrm{out}}\otimes\mathcal F$
by
\[
    G_a|\psi\rangle
    \triangleq \sum_{k,j}t_{a,j}^{(k)}K_k|\psi\rangle\otimes|k,j\rangle.
\]
The auxiliary basis $\{\ket{k,j}\}_{k,j}$ is orthonormal, so the assumed
operator identity implies
\begin{equation}
    G_{a'}^\dagger G_a
    =\sum_k\lambda^{(k)}_{a,a'}K_k^\dagger K_k
    =F_{a'}^\dagger F_a.
    \label{eq:ft-gram-equality}
\end{equation}
Thus the $F_a$ and $G_a$ have the same operator-valued
Gram matrix. Let $\mathcal S$ be the subspace spanned
by the ranges of the noisy Kraus operators:
\begin{equation}
    \mathcal S\triangleq
    \left\{
        \sum_a F_a|\psi_a\rangle:
        |\psi_a\rangle\in Q_0\ \text{for every }a
    \right\}.
\end{equation}
The vectors $|\psi_a\rangle$ are arbitrary, so the
displayed set is a linear subspace. Every noisy output
state is supported on $\mathcal S$.
On this subspace, define
\begin{equation}
    V_{\mathcal S}\Bigl(\sum_aF_a|\psi_a\rangle\Bigr)
    \triangleq\sum_aG_a|\psi_a\rangle.
    \label{eq:ft-recovery-isometry}
\end{equation}
This is independent of the chosen representation because
\begin{equation}
    \begin{aligned}
        \Bigl\|\sum_aG_a|\psi_a\rangle\Bigr\|^2
        &=\sum_{a,a'}\langle\psi_{a'}|
          F_{a'}^\dagger F_a|\psi_a\rangle\\
        &=\Bigl\|\sum_aF_a|\psi_a\rangle\Bigr\|^2.
    \end{aligned}
\end{equation}
Indeed, subtracting two representations of the same vector
gives zero on the right and therefore also on the left.
The resulting linear map $V_{\mathcal S}$ preserves norms
and is therefore an isometry.

Let $\Pi_{\mathcal S}$ project onto $\mathcal S$ and extend
$V_{\mathcal S}$ by zero, setting
$V=V_{\mathcal S}\Pi_{\mathcal S}$. Then
\[
    V^\dagger V=\Pi_{\mathcal S},
    \qquad VF_a=G_a.
\]
Choose a unit vector $|\xi\rangle\in\mathcal H_{\mathrm{out}}$
and define
\begin{equation}
    \begin{aligned}
        \mathcal R(X)
        &=\operatorname{Tr}_{\mathcal F}(VXV^\dagger)\\
        &\quad+\operatorname{Tr}[(I-\Pi_{\mathcal S})X]
          |\xi\rangle\langle\xi|.
    \end{aligned}
    \label{eq:ft-recovery-channel}
\end{equation}
The first term is completely positive. For an orthonormal
basis $\{|u_s\rangle\}_s$ of $\mathcal S^\perp$,
the second term has Kraus operators $|\xi\rangle\langle u_s|$
and is therefore also completely positive.
Moreover,
\[
    \operatorname{Tr}\mathcal R(X)
    =\operatorname{Tr}(\Pi_{\mathcal S}X)
     +\operatorname{Tr}[(I-\Pi_{\mathcal S})X]
    =\operatorname{Tr}X.
\]
Thus $\mathcal R$ is CPTP. Since every $F_a$ has range in
$\mathcal S$, the second term vanishes on noisy outputs.
Using $VF_a=G_a$ and
\eqref{eq:ft-positive-factorization}, we obtain
\begin{equation}
    \begin{aligned}
        \mathcal R\Bigl(\sum_aF_a\rho F_a^\dagger\Bigr)
        &=\sum_a\operatorname{Tr}_{\mathcal F}
          (G_a\rho G_a^\dagger)\\
        &=\sum_k\Bigl(\sum_{a,j}|t_{a,j}^{(k)}|^2\Bigr)
          K_k\rho K_k^\dagger\\
        &=\sum_kK_k\rho K_k^\dagger.
    \end{aligned}
\end{equation}
This proves \eqref{eq: fault tolerance criteria}.
\end{proof}
\begin{remark}[Operational interpretation of the coefficient matrices]
For $F_a=\mathbf E_e\mathbf C_o$, with $a=(o,e)$,
choose an environment basis $\{|a\rangle\}_a$.
The complementary channel of the noisy realization, obtained from
an associated Stinespring dilation \cite{Stinespring1955}, is
\[
    \mathbf C_{\mathbf E}^{\,c}(\rho)
      =\sum_{a,a'}\operatorname{Tr}(F_{a'}^\dagger F_a\rho)
                      |a\rangle\langle a'|.
\]
It gives the environment output when the circuit output
is discarded. \Cref{Thm: Iff Fault tolerance} expresses it as
\[
    \begin{aligned}
    \mathbf C_{\mathbf E}^{\,c}(\rho)
       &=\sum_kp_k(\rho)\Lambda^{(k)},\\
    p_k(\rho)&\triangleq\operatorname{Tr}(K_k\rho K_k^\dagger),
       \qquad K_k=\mathbf I\mathbf C_k.
    \end{aligned}
\]
Since the ideal branches retain distinct records,
\eqref{eq:ft-ideal-kraus-orthogonality} gives
\[
    \mathbf C_{\mathrm{id}}^{\,c}(\rho)
       =\sum_kp_k(\rho)|k\rangle\langle k|.
\]
Consequently,
\[
    \begin{aligned}
    \mathbf C_{\mathbf E}^{\,c}
       &=\mathcal T\circ\mathbf C_{\mathrm{id}}^{\,c},\\
    \mathcal T(X)
       &\triangleq\sum_k\langle k|X|k\rangle\Lambda^{(k)}.
    \end{aligned}
\]
Positivity and unit trace of each $\Lambda^{(k)}$ make
$\mathcal T$ a CPTP measure-and-prepare map, namely, it reads
$k$ and prepares the environment state $\Lambda^{(k)}$.
Thus exact recovery is possible precisely when the noisy
environment can be simulated from the outcome
record alone. This is the exact-recovery instance of
complementary-channel duality \cite{BenyOreshkov2010}.
An information-theoretic correctability criterion is also
given for strategic codes in
\cite[Theorem~2]{tanggara2024strategiccodeunifiedspatiotemporal}.
For a single isometric ideal branch, the complementary
channel is constant, recovering the usual
Knill--Laflamme interpretation.
\end{remark}

\begin{corollary}[Knill--Laflamme conditions \cite{knill1997theory} as a special case]
\label{cor: Knill-Laflamme conditions special case}
Let $C_{\mathrm{enc}}:Q_\L\to Q$ be an isometric encoder,
put $P=C_{\mathrm{enc}}C_{\mathrm{enc}}^\dagger$, and write
$\mathbf C_{\mathrm{enc},\mathrm{id}}(\rho)
=C_{\mathrm{enc}}\rho C_{\mathrm{enc}}^\dagger$.
For a CPTP noise channel $\mathcal E$ on $Q$ with
Kraus operators $\{E_a\}_a$, a CPTP recovery
$\mathcal R$ on $Q$ satisfying
\[
    \mathcal R\circ\mathcal E\circ
    \mathbf C_{\mathrm{enc},\mathrm{id}}
    =\mathbf C_{\mathrm{enc},\mathrm{id}}
\]
exists if and only if
\begin{equation}
    PE_a^\dagger E_bP=\alpha_{ab}P
    \qquad(\forall a,b)
    \label{eq: Knill-Laflamme conditions}
\end{equation}
for some $\alpha_{ab}\in\mathbb C$.
\end{corollary}

\begin{proof}
With ideal Kraus operator $C_{\mathrm{enc}}$ and noisy
Kraus operators $E_aC_{\mathrm{enc}}$,
\Cref{Thm: Iff Fault tolerance} gives the conditions
\[
    C_{\mathrm{enc}}^\dagger E_a^\dagger E_b
    C_{\mathrm{enc}}=\alpha_{ab}I_{Q_\L},
\]
which are equivalent to
\eqref{eq: Knill-Laflamme conditions} because
$C_{\mathrm{enc}}$ is an isometry.
For any unit vector $|\psi\rangle\in\operatorname{im}
C_{\mathrm{enc}}$, the coefficients satisfy
$\alpha_{ab}=\langle E_a\psi|E_b\psi\rangle$.
Thus $\alpha=[\alpha_{ab}]$ is positive semidefinite
and has unit trace by trace preservation of $\mathcal E$,
so the coefficient requirements of \Cref{Thm: Iff Fault tolerance} are automatic.
\end{proof}

\begin{corollary}[Semidefinite feasibility formulation]
\label{cor:ft-sdp}
Fix a circuit $\mathbf C$ and a noise model $\mathbf E$
satisfying the hypotheses of
\Cref{Thm: Iff Fault tolerance}. Let $m$ be the number
of combined indices $a=(o,e)$, and define
\[
\begin{aligned}
    F_a&\triangleq\mathbf E_e\mathbf C_o,
    \qquad a=(o,e),\\
    M_k&\triangleq
    \mathbf C_k^\dagger\mathbf I^\dagger
    \mathbf I\mathbf C_k.
\end{aligned}
\]
Then $\mathbf C$ is fault-tolerant under $\mathbf E$
if and only if the following semidefinite feasibility
problem, with Hermitian $m\times m$ matrix variables
$\{\Lambda^{(k)}\}_k$, is feasible:
\begin{equation}
\begin{aligned}
    \mathrm{find}\quad
    &\{\Lambda^{(k)}\}_k\\
    \mathrm{subject\; to}\quad
    &\Lambda^{(k)}\succeq0,
    \qquad \operatorname{Tr}\Lambda^{(k)}=1
    \quad\forall k,\\
    &F_{a'}^\dagger F_a
    =
    \sum_k[\Lambda^{(k)}]_{a,a'}M_k
    \quad\forall a,a'.
\end{aligned}
\label{eq:ft-sdp}
\end{equation}
The real and imaginary parts of the operator equalities
are imposed entrywise.
\end{corollary}

\begin{proof}
Since the circuit and noise model are fixed, the operators
$F_a$ and $M_k$ are known. All equality constraints in
\eqref{eq:ft-sdp} are therefore affine in the matrix
variables, and the remaining constraints are positive
semidefiniteness conditions. Thus \eqref{eq:ft-sdp} is
an SDP feasibility problem. Its equivalence to exact
CPTP recovery follows from
\Cref{Thm: Iff Fault tolerance} upon identifying
$\lambda_{a,a'}^{(k)}=[\Lambda^{(k)}]_{a,a'}$.
\end{proof}
\subsubsection{SDP for fault-tolerant circuit synthesis}
\label{sec:general-ft-synthesis}

We seek a circuit and a recovery that together realize a
prescribed operation under a fixed noise model. Fix the depth
$\ell$ and the registers of \Cref{def:quantum_circuit}, with
\[
    Q_r\cong\mathbb C^{2^{n_r}},\qquad
    A_r\cong\mathbb C^{2^{m_r}},\qquad
    A_r^{\mathrm{out}}\cong\mathbb C^{2^{m_r'}}.
\]
Here $Q_r$ is the data register, $A_r$ is the preparation
register and $A_r^{\mathrm{out}}$ is the retained outcome
register. The indices run over $0\leq r\leq\ell$ with $\ell$ fixed, and for every $r \in \mathbb Z^+_{\ell}$ we also fix
$m_r+n_{r-1}=m_r'+n_r$. The spaces and their dimensions are fixed throughout the search.

Let $\mathfrak C_\Theta=\{\mathbf C_\theta:\theta\in\Theta\}$
be a prescribed circuit family. For each component $r$ and
preceding outcome history $h=o_{1:r-1}$, its operators satisfy
\[
    \begin{gathered}
        C_{\theta,o_r|h}:Q_{r-1}\longrightarrow
                   A_r^{\mathrm{out}}\otimes Q_r,\\
        \sum_{o_r}C_{\theta,o_r|h}^\dagger
                         C_{\theta,o_r|h}=I_{Q_{r-1}}.
    \end{gathered}
\]
The parameter $\theta$ labels the permitted choices of
component operations, gate supports, gate orderings and
dependence on previous measurement outcomes within the
fixed depth and register structure. (\emph{Note: Explicit knowledge about the circuit family parametrized by $\Theta$
may yield additional affine constraints or linear matrix
inequalities that strengthen the SDP search in \Cref{prop:joint-ft-sdp}.}) Fixing $\theta$
specifies every component operator $C_{\theta,o_r|h}$
and hence the complete circuit $\mathbf C_\theta$.
The recovery $\mathcal R_\theta$ is chosen for that
circuit and may therefore vary with $\theta$. Set
\[
    \mathcal H_{\mathrm{out}}
       =A_{1:\ell}^{\mathrm{out}}\otimes Q_\ell,
\]
and use the external input copy $Q_0^{\mathrm{ext}}\cong Q_0$
of \Cref{app:causal-normalisation}. Fix a noise model
$\mathbf E$ satisfying \Cref{def: Circuit map and noise model}
on the common ports, and a target CPTP map
\[
    \mathbf C_{\mathrm{tar}}:
    \L(Q_0^{\mathrm{ext}})\longrightarrow
    \L(\mathcal H_{\mathrm{out}}).
\]
We ask whether for some $\theta\in\Theta$ there exists a CPTP recovery map
$\mathcal R_\theta:\L(\mathcal H_{\mathrm{out}})
\to\L(\mathcal H_{\mathrm{out}})$ such that
\begin{equation}
    \begin{aligned}
        \mathbf C_{\theta,\mathrm{id}}
            &=\mathbf C_{\mathrm{tar}},\\
        \mathcal R_\theta\circ\mathbf E\circ
            \mathbf C_{\theta,\mathrm{cm}}
            &=\mathbf C_{\mathrm{tar}}.
    \end{aligned}
    \label{eq:ft-family-synthesis}
\end{equation}
These equalities concern the complete operation, including
the quantum output and measurement record. The recovery may
depend on the chosen circuit and use its record, but must
work for every input without access to unobserved faults or
the noise memory. It is noiseless, and its implementation
cost is excluded from the circuit resource budget.
\Cref{fig:ft-synthesis-comb} illustrates the search question.

\begin{figure*}[t]
\centering
\begingroup
\tikzset{
  ft circuit/.style={
    draw=blue!55!black,fill=blue!6,
    rounded corners=1pt
  },
  ft noise/.style={
    draw=black!60,fill=black!7,
    rounded corners=1pt
  },
  ft recovery/.style={
    draw=teal!65!black,fill=teal!7,
    rounded corners=1pt
  }
}
\resizebox{\textwidth}{!}{%
\begin{quantikz}[
  row sep=0.72cm,
  column sep=0.30cm,
  wire types={b,b},
  execute at end picture={
    \node[below=5mm,text=blue!55!black]
      at ($(ftcfirst.south)!0.5!(ftclast.south)$)
      {$\mathbf C_\theta:\ \theta\in\Theta;
        \quad\ell \text{ fixed}$};
    \node[above=4mm,text=black!70]
      at ($(ftnfirst.north)!0.5!(ftnlast.north)$)
      {Fixed noise comb $\mathbf E$};
    \node[above=1mm,font=\scriptsize]
      at (ftcfirst.center |- ftnfirst.north)
      {$N_1$};
    \node[above=1mm,font=\scriptsize]
      at (ftclast.center |- ftnlast.north)
      {$N_\ell$};
  }
]
  \setwiretype{n}
  & \gate[2,style={ft noise,alias=ftnfirst}]
      {E_{e_0}}
  & \setwiretype{b}
  & \gate[2,style={ft noise}]
      {E_{e_1}}
  & \cdots\setwiretype{b}
  & \gate[2,style={ft noise}]
      {E_{e_{\ell-1}}}\setwiretype{b}
  &
  & \gate[2,style={ft noise,alias=ftnlast}]
      {E_{e_\ell}}
  & \setwiretype{n}
  & \setwiretype{n}
  \\
  \lstick{$Q_0^{\mathrm{ext}}$}
  &
  & \gate[style={ft circuit,alias=ftcfirst}]
      {C_{\theta,o_1}}
  &
  & \cdots\setwiretype{b}
  & \setwiretype{b}
  & \gate[style={ft circuit,alias=ftclast}]
      {C_{\theta,o_\ell\mid o_{1:\ell-1}}}
  &
  & \gate[style={ft recovery}]
      {\mathcal R_\theta}
  & \rstick{$\mathcal H_{\mathrm{out}}$}
\end{quantikz}%
}
\endgroup
\par\smallskip
\(\displaystyle
  \mathbf C_{\theta,\mathrm{id}}
    =\mathbf C_{\mathrm{tar}},
  \qquad
  \mathcal R_\theta\circ\mathbf E\circ
    \mathbf C_{\theta,\mathrm{cm}}
    =\mathbf C_{\mathrm{tar}}.
\)
\caption{Fault-tolerant circuit synthesis with fixed
depth $\ell$, register dimensions, target operation and
noise comb $\mathbf E$. The circuit $\mathbf C_\theta$ and
its recovery $\mathcal R_\theta$ are chosen jointly subject
to the displayed ideal-target and exact-recovery conditions. The local boxes show a fixed
noise-Kraus index sequence $e=(e_0,\ldots,e_\ell)$
and circuit outcome sequence $o=(o_1,\ldots,o_\ell)$;
summing over $e$ and $o$ gives the noisy channel
before recovery. The upper wire carries the noise memory;
record registers and classical control are suppressed.
Noise acts on the ports specified by the model. For adaptive
circuits, the outcome history selects the components as in
\Cref{rem:noise_model_explicit}.}
\label{fig:ft-synthesis-comb}
\end{figure*}
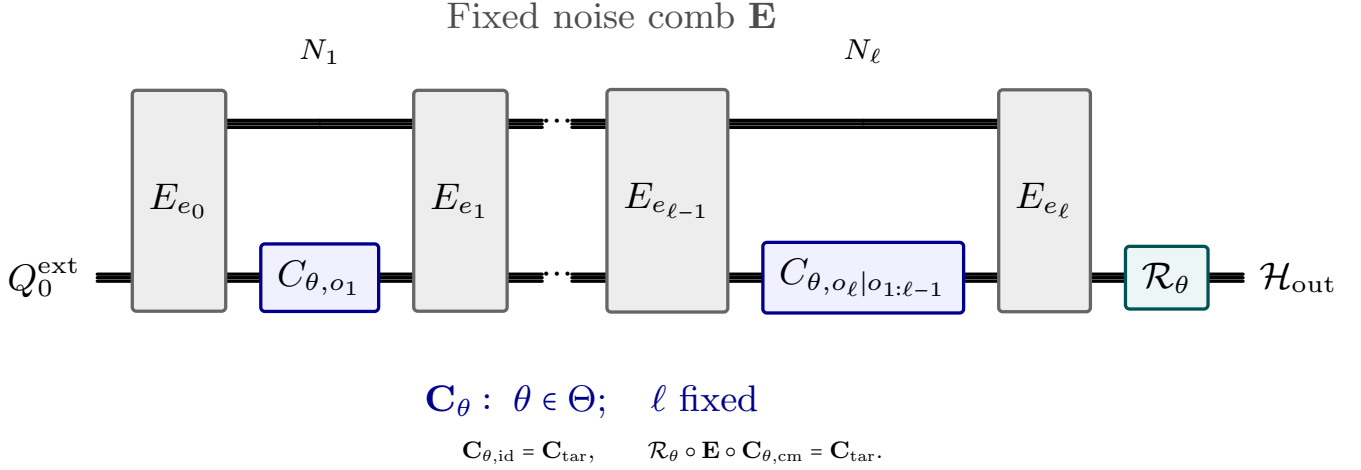

We first describe the circuit search variables. Use the
open-wire prefix spaces
\[
    \mathcal H_r=
       \bigotimes_{i\in\mathbb Z^+_{r}}
       \left(A_i^{\mathrm{out}}\otimes Q_i^{\mathrm{out}}
                            \otimes Q_{i-1}^{\mathrm{in}}\right)
\]
of \Cref{app:causal-normalisation}, with $\mathcal H_0=\mathbb C$. Set
$d_0=\dim Q_0^{\mathrm{ext}}$. An actual circuit has prefix
normalization operators, for $1\leq r\leq\ell$,
\[
    \mathbf C_\theta^{(r)}
       =\sum_{o_{1:r}}\mathbf C_{\theta,o_{1:r}}
                         \mathbf C_{\theta,o_{1:r}}^\dagger,
    \qquad
    \mathbf C_\theta^{(0)}=I_{Q_0^{\mathrm{ext}}},
\]
where the prefix Kraus operators are defined in \Cref{app:causal-normalisation}. Now, introduce Hermitian variables
$\mathbf C^{(r)}\in\L(\mathcal H_r\otimes Q_0^{\mathrm{ext}})$
and impose causal normalization constraints (see \Cref{app:causal-normalisation}):
\begin{equation}
    \begin{aligned}
        \mathbf C^{(0)}&=I_{Q_0^{\mathrm{ext}}},\\
        \mathbf C^{(r)}&\succeq0,\\
        \operatorname{Tr}_{A_r^{\mathrm{out}}Q_r^{\mathrm{out}}}
             \mathbf C^{(r)}
          &=I_{Q_{r-1}^{\mathrm{in}}}\otimes\mathbf C^{(r-1)}
              \quad(1\leq r\leq\ell).
    \end{aligned}
    \label{eq:layout-comb-normalisation}
\end{equation}
We also retain
\begin{equation}
    \begin{aligned}
        \mathbf C^{(\ell)}
          &=\frac{\operatorname{Tr}_{Q_0^{\mathrm{ext}}}
                         \mathbf C^{(\ell)}}{d_0}
                         \otimes I_{Q_0^{\mathrm{ext}}},\\
        \mathbf C^{(\ell)}
          &=\sum_oP_o\mathbf C^{(\ell)}P_o,\\
        (\mathbf C^{(\ell)})^{\mathsf T_{1:r}}
          &\succeq0\qquad(1\leq r<\ell).
    \end{aligned}
    \label{eq:synthesis-circuit-structure}
\end{equation}
Here $P_o=|o\rangle\langle o|_{A_{1:\ell}^{\mathrm{out}}}
\otimes I$, and $\mathsf T_{1:r}$ transposes the complete
input and output factors of components $1,\ldots,r$.
These constraints express the external identity factor,
classical records and partial-transpose positivity.

Let $\mathfrak S_\Theta$ be the set of Hermitian prefix tuples
satisfying \eqref{eq:layout-comb-normalisation} and
\eqref{eq:synthesis-circuit-structure}, together with any
specified finite collection of affine equalities and linear
matrix inequalities valid for every circuit in
$\mathfrak C_\Theta$. These additional constraints encode
necessary properties of the chosen family.

\begin{proposition}[Circuit search set]
\label{prop:synthesis-search-space}
The set $\mathfrak S_\Theta$ is convex, compact and
SDP-representable, and
\begin{equation}
    \left\{(\mathbf C_\theta^{(0)},\ldots,
           \mathbf C_\theta^{(\ell)}):\theta\in\Theta\right\}
       \subseteq\mathfrak S_\Theta.
    \label{eq:synthesis-family-inclusion}
\end{equation}
Every tuple in this set satisfies
\begin{equation}
    \operatorname{Tr}\mathbf C^{(r)}
       =d_0\prod_{i=0}^{r-1}\dim Q_i,
       \qquad 0\leq r\leq\ell,
    \label{eq:synthesis-prefix-trace}
\end{equation}
where the product is one for $r=0$.
\end{proposition}

\begin{proof}
For an actual circuit $\mathbf C_\theta$, causal normalization
follows from \Cref{app:causal-normalisation}. Its Kraus
representation gives the external identity factor and block
diagonality in the retained outcome record. Each branch operator
$\mathbf C_{\theta,o}\mathbf C_{\theta,o}^\dagger$
is a tensor product of positive semidefinite operators on the
complete component factors and $I_{Q_0^{\mathrm{ext}}}$.
Transposing any collection of complete component factors
therefore preserves positivity. Summing over the branches
establishes the partial-transpose constraints and hence
\eqref{eq:synthesis-circuit-structure}. The additional family
constraints hold by assumption, proving
\eqref{eq:synthesis-family-inclusion}.

The defining constraints are a finite collection of affine
equalities and linear matrix inequalities in the Hermitian
prefix variables. Thus $\mathfrak S_\Theta$ is convex and
SDP-representable. It is also closed, since the affine maps
are continuous and the positive semidefinite cone is closed.

It remains to establish boundedness. Taking traces in
\eqref{eq:layout-comb-normalisation} gives
\[
    \operatorname{Tr}\mathbf C^{(r)}
       =(\dim Q_{r-1})
          \operatorname{Tr}\mathbf C^{(r-1)}.
\]
Starting from $\mathbf C^{(0)}=I_{Q_0^{\mathrm{ext}}}$ and
iterating proves \eqref{eq:synthesis-prefix-trace}.

To bound each prefix, use the Frobenius norm
$\|A\|_{\mathrm F}
=\sqrt{\operatorname{Tr}(A^\dagger A)}$.
Since $\mathbf C^{(r)}\succeq0$, its eigenvalues
$\lambda_j$ are nonnegative, and therefore
\[
    \begin{aligned}
        \|\mathbf C^{(r)}\|_{\mathrm F}^2
          &=\operatorname{Tr}\bigl[(\mathbf C^{(r)})^2\bigr]
            =\sum_j\lambda_j^2\\
          &\leq\left(\sum_j\lambda_j\right)^2\\
          &=\left(d_0\prod_{i=0}^{r-1}\dim Q_i\right)^2.
    \end{aligned}
\]
The number of prefixes and all register dimensions are fixed,
so these bounds make the entire prefix tuple bounded.
Hence $\mathfrak S_\Theta$ is closed and bounded in a
finite-dimensional real vector space, and is therefore compact.
\end{proof}

\begin{remark}
    Membership in $\mathfrak S_\Theta$ need not provide a circuit
decomposition in $\mathfrak C_\Theta$. This is the first
relaxation of the search problem, namely, we only retain necessary matrix constraints instead
of requiring an admissible circuit realization.
\end{remark}

We express the target and recovery conditions using the
unnormalized Choi operator
\begin{equation}
    J_{\mathcal N}
       =(\mathcal N\otimes\operatorname{id})
             (\vket I\vbra I)
       =\sum_a\vket{K_a}\vbra{K_a}
    \label{eq:synthesis-choi-correspondence}
\end{equation}
for $\mathcal N(\rho)=\sum_aK_a\rho K_a^\dagger$. Equality of Choi operators is equivalent to equality of the
corresponding maps \cite{chiribella2009theoretical}. We can
therefore formulate the search for a CPTP recovery map as
a search for its Choi operator, subject to positivity and
the partial-trace constraint enforcing trace preservation.

Let $R_0\cong Q_0^{\mathrm{ext}} \cong Q_0$ be a reference register,
and put
\[
    \Phi_0=\vket{I_{Q_0}}\vbra{I_{Q_0}},\qquad
    J_{\mathrm{tar}}
       =(\mathbf C_{\mathrm{tar}}\otimes\operatorname{id}_{R_0})
          (\Phi_0),
\]
where $\Phi_0$ acts on $Q_0^{\mathrm{ext}}\otimes R_0$.
Define the fixed linear map
\begin{equation}
    J_{\mathbf E}(Y)
       =(\mathbf E\otimes\operatorname{id}_{R_0})
       \left[
          \frac{\operatorname{Tr}_{Q_0^{\mathrm{ext}}}Y}{d_0}
          \otimes\Phi_0
       \right]
    \label{eq:comb-fault-choi-map}
\end{equation}
from $\L(\mathcal H_\ell\otimes Q_0^{\mathrm{ext}})$ to
$\L(\mathcal H_{\mathrm{out}}\otimes R_0)$. The normalized
partial trace gives
\[
    \mathbf C_{\theta,\mathrm{cm}}(\rho)
       =\frac{\operatorname{Tr}_{Q_0^{\mathrm{ext}}}
                   \mathbf C_\theta^{(\ell)}}{d_0}\otimes\rho.
\]
Define $J_{\mathrm{id}}(\cdot)$ by replacing $\mathbf E$ with $\mathbf I$ in the definition of $J_{\mathbf E}(\cdot)$.
For actual circuits,
\begin{equation}
    \begin{aligned}
        J_{\mathbf E}(\mathbf C_\theta^{(\ell)})
           &=J_{\mathbf E\circ\mathbf C_{\theta,\mathrm{cm}}},\\
        J_{\mathrm{id}}(\mathbf C_\theta^{(\ell)})
           &=J_{\mathbf C_{\theta,\mathrm{id}}}.
    \end{aligned}
    \label{eq:comb-operation-choi-identification}
\end{equation}
Both maps also apply to relaxed matrices in $\mathfrak S_\Theta$ for which no
circuit realization is known.

Let $\mathcal H'_{\mathrm{out}}\cong\mathcal H_{\mathrm{out}}$
label the recovery output. A recovery $\mathcal R_\theta$
for the circuit $\mathbf C_\theta$ is represented by its
unnormalized Choi operator
\[
    J_{\mathcal R_\theta}\in
    \L(\mathcal H'_{\mathrm{out}}\otimes
       \mathcal H_{\mathrm{out}}).
\]
The inverse Choi formula~\cite{chiribella2009theoretical}
determines $\mathcal R_\theta$ linearly from
$J_{\mathcal R_\theta}$ as
\[
    \mathcal R_\theta(A)
       =\operatorname{Tr}_{\mathcal H_{\mathrm{out}}}
          \left[
             J_{\mathcal R_\theta}
             \bigl(I_{\mathcal H'_{\mathrm{out}}}
                    \otimes A^{\mathsf T}\bigr)
          \right].
\]
This map is CPTP precisely when
\begin{equation}
    J_{\mathcal R_\theta}\succeq0,\qquad
    \operatorname{Tr}_{\mathcal H'_{\mathrm{out}}}
        J_{\mathcal R_\theta}
       =I_{\mathcal H_{\mathrm{out}}}.
    \label{eq:synthesis-recovery-choi}
\end{equation}
All transposes use fixed bases. We identify
$\mathcal H'_{\mathrm{out}}$ with $\mathcal H_{\mathrm{out}}$
in these bases when comparing the recovered operation
with the target.

For a circuit $\mathbf C_\theta$ that implements the
target ideally, exact recovery is equivalent to
\eqref{eq:synthesis-recovery-choi} and
\begin{equation}
    (\mathcal R_\theta\otimes\operatorname{id}_{R_0})
       \bigl(J_{\mathbf E}(\mathbf C_\theta^{(\ell)})\bigr)
       =J_{\mathrm{tar}}.
    \label{eq:synthesis-exact-recovery}
\end{equation}
For a fixed circuit, this equation is linear in
$J_{\mathcal R_\theta}$ and gives an exact recovery SDP.
When the circuit also varies, it contains products of
unknown entries of $J_{\mathcal R_\theta}$ and
$\mathbf C^{(\ell)}$. The equation is therefore
bilinear, and the joint problem is generally nonconvex
even after replacing the circuit constraints by
$\mathfrak S_\Theta$.

To express this bilinear equation as a linear constraint,
introduce a Hermitian variable
\[
    Z\in\L\bigl(
        \mathcal H'_{\mathrm{out}}\otimes
        \mathcal H_{\mathrm{out}}\otimes
        \mathcal H_\ell\otimes Q_0^{\mathrm{ext}}\bigr).
\]
For a genuine circuit--recovery pair, we would have
$\mathbf C^{(\ell)}=\mathbf C_\theta^{(\ell)}$ and
\begin{equation}
    Z=J_{\mathcal R_\theta}\otimes\mathbf C^{(\ell)}.
    \label{eq:synthesis-product-operator}
\end{equation}
Define a fixed linear map $\mathcal K_{\mathbf E}$
on product matrices by
\begin{equation}
    \mathcal K_{\mathbf E}(J_{\mathcal R_\theta}\otimes Y)
       =(\mathcal R_\theta\otimes\operatorname{id}_{R_0})
          \bigl(J_{\mathbf E}(Y)\bigr).
    \label{eq:synthesis-linear-contraction}
\end{equation}
For this definition,
$J_{\mathcal R_\theta}\in
\L(\mathcal H'_{\mathrm{out}}\otimes\mathcal H_{\mathrm{out}})$
and
$Y\in\L(\mathcal H_\ell\otimes Q_0^{\mathrm{ext}})$
may be arbitrary operators.
The corresponding linear map $\mathcal R_\theta$ is
given by the inverse Choi formula \cite{chiribella2009theoretical} applied to $J_{\mathcal{R}_\theta}$, without imposing
the CPTP constraints.
The right-hand side is bilinear in
$J_{\mathcal R_\theta}$ and $Y$. Its values on a product
matrix basis therefore define a unique linear map on
all $Z$, with output in
$\L(\mathcal H'_{\mathrm{out}}\otimes R_0)$.
This construction is independent of the chosen
$\theta$. For a genuine circuit-recovery pair $(\mathbf C_\theta, \mathcal{R}_\theta)$, the recovery equation
becomes
\begin{equation}
    \mathcal K_{\mathbf E}(Z)=J_{\mathrm{tar}},
    \label{eq:synthesis-product-contraction}
\end{equation}
which is linear in $Z$.

The product requirement
\eqref{eq:synthesis-product-operator} remains nonconvex.
We relax it by retaining necessary semidefinite
constraints. Let $\mathsf T_{\mathrm{rec}}$ transpose
the complete recovery factor
$\mathcal H'_{\mathrm{out}}\otimes\mathcal H_{\mathrm{out}}$,
leaving the circuit factors
$\mathcal H_\ell\otimes Q_0^{\mathrm{ext}}$ unchanged.
Every genuine pair satisfies
\begin{equation}
    \begin{aligned}
        Z&\succeq0,\\
        Z^{\mathsf T_{\mathrm{rec}}}
           &=J_{\mathcal R_\theta}^{\mathsf T}
                \otimes\mathbf C^{(\ell)}
             \succeq0,\\
        \operatorname{Tr}_{\mathcal H'_{\mathrm{out}}}Z
           &=I_{\mathcal H_{\mathrm{out}}}
                \otimes\mathbf C^{(\ell)}.
    \end{aligned}
    \label{eq:synthesis-joint-constraints}
\end{equation}
The second line follows because the full transpose
$J_{\mathcal R_\theta}^{\mathsf T}$ is positive
semidefinite, and the last follows from trace
preservation of $\mathcal R_\theta$.
These constraints are necessary for the product form
but do not enforce it. Replacing the product requirement
by them therefore gives the second relaxation in our search problem.

Combining these constraints with the circuit search
set gives the following SDP.

\begin{theorem}[SDP relaxation for fault-tolerant synthesis]
\label{prop:joint-ft-sdp}
For the prescribed registers, circuit family
$\mathfrak C_\Theta$, noise model $\mathbf E$ and target
$\mathbf C_{\mathrm{tar}}$, consider
\begin{equation}
    \begin{aligned}
        \mathrm{find}\quad
           &\{\mathbf C^{(r)}\}_{r=0}^{\ell},\ Z\\
        \mathrm{subject\ to}\quad
           &(\mathbf C^{(0)},\ldots,\mathbf C^{(\ell)})
                 \in\mathfrak S_\Theta,\\
           &J_{\mathrm{id}}(\mathbf C^{(\ell)})
                 =J_{\mathrm{tar}},\\
           &Z\succeq0,\qquad
                 Z^{\mathsf T_{\mathrm{rec}}}\succeq0,\\
           &\operatorname{Tr}_{\mathcal H'_{\mathrm{out}}}Z
                 =I_{\mathcal H_{\mathrm{out}}}
                    \otimes\mathbf C^{(\ell)},\\
           &\mathcal K_{\mathbf E}(Z)=J_{\mathrm{tar}}.
    \end{aligned}
    \label{eq:joint-ft-sdp}
\end{equation}
This is an SDP. Every circuit--recovery pair satisfying
\eqref{eq:ft-family-synthesis} gives a feasible point.
Consequently, infeasibility rules out every circuit in
$\mathfrak C_\Theta$ that implements the target ideally
and admits a CPTP recovery under $\mathbf E$.
The recovery may depend on the selected circuit.
\end{theorem}

\begin{proof}
By \Cref{prop:synthesis-search-space},
$\mathfrak S_\Theta$ is SDP-representable.
The remaining constraints are affine equalities and
linear matrix inequalities, since
$J_{\mathrm{id}}$, $\mathcal K_{\mathbf E}$, partial
trace and partial transpose are fixed linear maps.
Thus \eqref{eq:joint-ft-sdp} is an SDP.

Suppose that a circuit $\mathbf C_\theta$ and a CPTP
recovery $\mathcal R_\theta$ satisfy
\eqref{eq:ft-family-synthesis}. Set
\[
    \mathbf C^{(r)}=\mathbf C_\theta^{(r)},
    \qquad
    Z=J_{\mathcal R_\theta}
         \otimes\mathbf C_\theta^{(\ell)}.
\]
The circuit tuple belongs to $\mathfrak S_\Theta$
by \eqref{eq:synthesis-family-inclusion}, and the
ideal-target equality gives
\[
    J_{\mathrm{id}}(\mathbf C_\theta^{(\ell)})
       =J_{\mathrm{tar}}.
\]
Since $\mathcal R_\theta$ is CPTP,
\eqref{eq:synthesis-joint-constraints} gives positivity
of $Z$, positivity of its indicated partial transpose
and the required partial trace. Finally, the definition
of $\mathcal K_{\mathbf E}$ and exact recovery give
\[
    \begin{aligned}
        \mathcal K_{\mathbf E}(Z)
          &=(\mathcal R_\theta\otimes\operatorname{id}_{R_0})
             \bigl(J_{\mathbf E}
                    (\mathbf C_\theta^{(\ell)})\bigr)\\
          &=J_{\mathrm{tar}}.
    \end{aligned}
\]
Hence every such pair gives a feasible point.
The infeasibility statement follows by contraposition.
\end{proof}

\begin{remark}[Feasibility and recovery]
A feasible point need not identify an admissible circuit.
Suppose that its final normalization operator admits a
realization
$\mathbf C^{(\ell)}=\mathbf C_\theta^{(\ell)}$
for some $\theta\in\Theta$. For this circuit, define
$\mathcal R_\theta$ through the Choi operator
\begin{equation}
    J_{\mathcal R_\theta}
       =\frac{
          \operatorname{Tr}_{\mathcal H_\ell
                             Q_0^{\mathrm{ext}}}Z}
          {\operatorname{Tr}\mathbf C^{(\ell)}}.
    \label{eq:synthesis-extracted-recovery}
\end{equation}
The denominator is fixed and strictly positive by
\eqref{eq:synthesis-prefix-trace}. Since $Z\succeq0$,
we have $J_{\mathcal R_\theta}\succeq0$. Moreover,
the SDP partial-trace constraint gives
\[
    \begin{aligned}
        \operatorname{Tr}_{\mathcal H'_{\mathrm{out}}}
             J_{\mathcal R_\theta}
          &=
          \frac{
             \operatorname{Tr}_{\mathcal H_\ell
                                Q_0^{\mathrm{ext}}}
                \bigl(I_{\mathcal H_{\mathrm{out}}}
                      \otimes\mathbf C^{(\ell)}\bigr)}
               {\operatorname{Tr}\mathbf C^{(\ell)}}\\
          &=I_{\mathcal H_{\mathrm{out}}}.
    \end{aligned}
\]
Thus $\mathcal R_\theta$ is CPTP. Sufficient conditions for synthesis are therefore
\begin{equation}
    \begin{aligned}
        \mathbf C^{(\ell)}
           &=\mathbf C_\theta^{(\ell)}
               \quad\text{for some }\theta\in\Theta, \quad \mathrm{and}\\
        Z&=J_{\mathcal R_\theta}\otimes\mathbf C^{(\ell)},
    \end{aligned}
    \label{eq:synthesis-realization-certificate}
\end{equation}
with $J_{\mathcal R_\theta}$ given by
\eqref{eq:synthesis-extracted-recovery}.
The first condition supplies an admissible circuit.
The second turns the final SDP equality into
\[
    (\mathcal R_\theta\otimes\operatorname{id}_{R_0})
       \bigl(J_{\mathbf E}
               (\mathbf C_\theta^{(\ell)})\bigr)
       =J_{\mathrm{tar}}.
\]
Together with the ideal-target constraint, equality of
Choi operators then gives
\[
    \mathbf C_{\theta,\mathrm{id}}
       =\mathbf C_{\mathrm{tar}},
    \qquad
    \mathcal R_\theta\circ\mathbf E\circ
       \mathbf C_{\theta,\mathrm{cm}}
       =\mathbf C_{\mathrm{tar}}.
\]
Feasibility alone guarantees neither an admissible
circuit realization nor that the extracted CPTP map
recovers its ideal operation. Positivity,
partial-transpose positivity and even separability
do not by themselves establish the product identity
in \eqref{eq:synthesis-realization-certificate}.

When $\mathfrak C_\Theta$ contains every implementation
allowed by a given resource budget, and the fixed noise
model applies throughout, infeasibility excludes all
such implementations. Recovery may vary with the
circuit, but remains noiseless and its implementation
cost is excluded from that budget.
\end{remark}

\begin{remark}[Fixed recovery]
If the same CPTP recovery is prescribed for every
circuit, then $\mathcal R_\theta=\mathcal R_*$
for all $\theta\in\Theta$, with fixed Choi operator
$J_{\mathcal R_*}$. We may then impose the exact
product constraint
\[
    Z=J_{\mathcal R_*}\otimes\mathbf C^{(\ell)},
\]
which is linear in $\mathbf C^{(\ell)}$.
Eliminating $Z$ gives the fixed-recovery SDP, with
correction condition
\[
    (\mathcal R_*\otimes\operatorname{id}_{R_0})
       \bigl(J_{\mathbf E}(\mathbf C^{(\ell)})\bigr)
       =J_{\mathrm{tar}}.
\]
Both positivity constraints and the partial-trace
constraint on $Z$ hold automatically because
$\mathcal R_*$ is CPTP and
$\mathbf C^{(\ell)}\succeq0$. If the SDP for variable recovery is infeasible then the SDP for fixed recovery is automatically infeasible.
\end{remark}


\section{Fault tolerance of non-adaptive stabilizer circuits}
\label{sec: Fault tolerance of stabilizer circuits}

In this section, we characterize fault-tolerance of non-adaptive stabilizer circuits under
stochastic Pauli noise, with a noiseless encoding circuit, all outcomes retained, and noiseless recovery
on the output data and measurement ancillas.
\subsection{Stochastic Pauli noise model}
\label{subsec:Stoch Pauli noise}
We begin by formalizing the notion of spacetime locations on which faults occur before defining the stochastic Pauli noise model.

\begin{definition}[Circuit Pauli fault space and fault locations]
\label{def:binary-fault-space}
Fix an abstraction level $\zeta\in\{0,1\}$ and consider
a non-adaptive stabilizer circuit $\mathbf C
    \triangleq
    \{C^{(1)}_{o_1},\ldots,C^{(\ell)}_{o_\ell}\}
        _{o\in\mathbb F_2^{\ell_M}}$, of depth $\ell$, with component operators 
\[
    C^{(t)}_{o_t}:Q_{t-1}\longrightarrow
    A_t^{\mathrm{out}}\otimes Q_t,
    \qquad t\in \mathbb Z^+_\ell,
\]
as in \Cref{def: stab circuit}.
Here $Q_t$ contains $n_t$ data qubits,
$A_t^{\mathrm{out}}$ contains $m'_t$ ancilla qubits,
and set
\[
    o=[o_1,\ldots,o_\ell] \in \mathbb F_2^{\ell_M},
    \qquad o_t\in\mathbb F_2^{m'_t},
    \qquad \ell_M=\sum_{t=1}^{\ell}m'_t,
\]
with $A_0^{\mathrm{out}}=\mathbb C$ and $m'_0=0$. Furthermore, for $t \in \{0,\cdots , \ell\}$, fix decompositions
\[
    A_t^{\mathrm{out}}
    \cong\bigotimes_{s=1}^{m'_t}Q_{s,t},
    \qquad
    Q_t\cong
\bigotimes_{s=m'_t+1}^{m'_t+n_t}Q_{s,t},
    \qquad Q_{s,t}\cong\mathbb C^2.
\]
Define the permitted \emph{fault-location set} $\Lambda^{(\zeta)}$ and the \emph{spacetime fault space} $\mathcal E^{(\zeta)}$ of $\mathbf{C}$ as
\begin{equation}
\begin{aligned}
    \Lambda^{(\zeta)}
    &\triangleq
    \{(s,t):0\leq t\leq\ell,\ s\in\mathbb Z^+_{m'_t+n_t}\}, \\
     \mathcal E^{(\zeta)}
    &\triangleq
    \left\{
        e= [e_{s,t}]_{(s,t)\in\Lambda^{(\zeta)}}:
    e_{s,t}\in\mathbb F_2^2
    \right\}.
\end{aligned}
    \label{eq:pauli-fault-loc-and-space}
\end{equation}
With componentwise addition and scalar multiplication,
\[
    \mathcal E^{(\zeta)}\cong\mathbb F_2^{d_{\zeta}},
    \qquad d_{\zeta}\triangleq2|\Lambda^{(\zeta)}|.
\]
\end{definition}

\begin{definition}[Stochastic Pauli noise model]
\label{def:stochastic-pauli-noise}
For the circuit and fault space of
\Cref{def:binary-fault-space}, a
\emph{stochastic Pauli noise model} $\mathbf E$ has multi-step
Kraus operators
\[
    \{\sqrt{\mu_{e}}\,\mathbf E_{e}\}
        _{e\in\mathcal E^{(\zeta)}},
    \qquad
    \mu_{e}\geq0,
    \qquad
    \sum_{e\in\mathcal E^{(\zeta)}}\mu_{e}=1.
\]
For each fault pattern $e\in \mathcal{E}^{(\zeta)}$, we define its record
and data Paulis by
\[
    \begin{aligned}
        P_t^A(e)
        &\triangleq
        \bigotimes_{s=1}^{m'_t}P(e_{s,t})
        \in\mathcal P(A_t^{\mathrm{out}}),
        &&t\in\mathbb Z^+_\ell,\\
        P_t^Q(e)
        &\triangleq
        \bigotimes_{s=m'_t+1}^{m'_t+n_t}
            P(e_{s,t})
        \in\mathcal P(Q_t),
        &&0\leq t\leq\ell.
    \end{aligned}
\]
The explicit form of the multi-step operators is defined as
\begin{equation}
    \begin{aligned}
        \mathbf E_{e}
        \triangleq
        \bigotimes_{t=1}^{\ell}
            P_t^A(e)
        \otimes P_\ell^Q(e) \otimes
        \vbra{
        P_{\ell-1}^Q(e)^\star}
        \otimes\cdots\otimes
        \vbra{P_0^Q(e)^\star}.
    \end{aligned}
    \label{eq:pauli-fault-contraction}
\end{equation}
\end{definition}

\begin{proposition}[Channels for fixed Pauli faults]
\label{prop:fixed-pauli-fault-channels}
For the circuit and noise model of
\Cref{def:binary-fault-space,def:stochastic-pauli-noise},
let $\mathbf C_o$ be the circuit-map Kraus operators
of \Cref{def: Circuit map and noise model}.
For each $e\in\mathcal E^{(\zeta)}$, define
\begin{equation}
    F_{e,o}\triangleq\mathbf E_{e}\mathbf C_o,
    \qquad
    \mathcal N_{e}(\rho)
    \triangleq
    \sum_{o\in\mathbb F_2^{\ell_M}}
        F_{e,o}\rho F_{e,o}^\dagger.
    \label{eq:ft-fixed-fault-channel}
\end{equation}
Each $\mathcal N_{e}:
\L(Q_0)\to\L(A_{1:\ell}^{\mathrm{out}}\otimes Q_\ell)$
is CPTP, and $\mathcal N_0=\mathbf C_{\mathrm{id}}$.
The stochastic noisy realization $\mathbf C_{\mathbf E} = \sum_{e}\mu_{e}\mathcal N_{e}$.
\end{proposition}

\begin{proof}
By \Cref{rem:noise_model_explicit}, $F_{e,o}$
is the composed noisy branch operator with the specified Pauli faults.
Summing successively over $o_\ell,\ldots,o_1$,
component completeness relations and Pauli unitarity give
\[
    \sum_oF_{e,o}^\dagger F_{e,o}
    =I_{Q_0}.
\]
Thus the Kraus form defines a CPTP map.
Finally, $\mathbf E_0=\mathbf I$ gives the ideal
realization, and the weighted Kraus operators give
the stated mixture.
\end{proof}

\subsection{Fault signatures and effect}
\label{subsec:fault-signatures}

Fix $\zeta\in\{0,1\}$ and a noiseless stabilizer encoder
$C_{\mathrm{enc}}:Q_\L\to Q_0$, where $Q_\L \cong \C^{2^{k_{\mathrm{in}}}}$. Throughout the remainder of
this section, include this encoder in the fixed-fault
channels of \Cref{prop:fixed-pauli-fault-channels}, so that
\[
    F_{e,o}
    =\mathbf E_{e}\mathbf C_oC_{\mathrm{enc}},
    \qquad
    F_{0,o}=\mathbf I\mathbf C_oC_{\mathrm{enc}}\triangleq\widetilde {\mathbf {C}}_o.
\]
Denote the resulting channels again by $\mathcal N_{e}$.
Write $A=A_{1:\ell}^{\mathrm{out}}$ and
$H_{\mathrm{out}}^{(\zeta)}=A\otimes Q_\ell$, with
$n_R=\ell_M$ record qubits and $n_{\mathrm{out}}=n_\ell$
data qubits. Let $\mathcal W^{(\zeta)} \triangleq \mathcal W_{\widetilde{\mathbf C}}$
and retain its subspaces
$\mathcal W_{\mathrm{det}}$, $\mathcal W_{\mathrm{obs}}$,
and $\mathcal W_{\mathrm{out}}$ from
\Cref{def:common-stabilizer-relations}.
For $\omega=[b,p,p']\in\mathcal W^{(\zeta)}$, put
\begin{equation}
    R_\omega\triangleq
    (-1)^{c_\omega}Z_A^b\otimes P(p'),
    \qquad
    R_\omega\widetilde C_o=\widetilde C_oP(p).
    \label{eq:ft-signed-output-relation}
\end{equation}

\begin{proposition}[Bi-linear fault relations]
\label{lem:ft-signature-completeness}
There exists a unique bilinear map
\[
    \chi_{\zeta}:\mathcal E^{(\zeta)}\times\mathcal W^{(\zeta)}
    \longrightarrow\mathbb F_2
\]
such that, for every $e \in \mathcal{E}^{(\zeta)}$, $\omega=[b,p,p'] \in \mathcal{W}^{(\zeta)}$,
and outcome string $o \in \mathbb F_2^{\ell_M}$,
\begin{equation}
    R_\omega F_{e,o}
    =(-1)^{\chi_{\zeta}(e,\omega)}
       F_{e,o}P(p).
    \label{eq:ft-noisy-common-relation}
\end{equation}
Moreover, $\mathcal N_{e}=\mathcal N_{\mathbf f}$
if and only if
$\chi_{\zeta}(e,\cdot)=\chi_{\zeta}(\mathbf f,\cdot)$.
\end{proposition}

\begin{proof}
Clifford conjugation propagates Pauli labels linearly,
while
\[
    \Pi_t^{P(m)}P(u)
    =P(u)\Pi_{t+\langle u,m\rangle_s}^{P(m)}
\]
shifts measurement outcomes. Consequently, applying the Clifford conjugation rule of
\Cref{Prop: Symplectic Clifford} and the measurement identity above successively to the composed
branches of \Cref{rem:noise_model_explicit} gives
$\mathbb F_2$-linear maps
\[
    \beta^{(\zeta)}:\mathcal E^{(\zeta)}
    \longrightarrow\mathbb F_2^{2(n_R+n_{\mathrm{out}})},
    \qquad
    s:\mathcal E^{(\zeta)}
    \longrightarrow\mathbb F_2^{n_R}
\]
such that
\[
    F_{e,o}
    \doteq P(\beta^{(\zeta)}(e))
            \widetilde C_{o+s(e)}
    \qquad(\forall\,e,o),
\]
where $\doteq$ denotes equality up to an overall phase. Summing over all outcomes gives
\begin{equation}
    \mathcal N_{e}(\rho)
    =
    P(\beta^{(\zeta)}(e))\mathcal N_0(\rho)
    P(\beta^{(\zeta)}(e))^\dagger.
    \label{eq: stab noise propagation}
\end{equation}
If $r(\omega)$ is the binary label of $R_\omega$, then
\begin{equation}
    \chi_{\zeta}(e,\omega)
    =\langle\beta^{(\zeta)}(e),r(\omega)\rangle_s
    \label{eq:ft-signature-commutation}
\end{equation}
is bilinear and satisfies
\eqref{eq:ft-noisy-common-relation}.
Trace preservation guarantees a nonzero branch for
each $e$, establishing uniqueness. Finally, for $\omega=[b,p,p']\in\mathcal W^{(\zeta)}$, put
\[
    g(\omega)\triangleq
    R_\omega\otimes P(p)^{\mathsf T},
    \qquad
    J_{e}\triangleq
    \sum_o\vket{F_{e,o}}\vbra{F_{e,o}}.
\]
Then define the normalized stabilizer state
\[
    |\psi\rangle
    =2^{-k_{\mathrm{in}}/2}
      \sum_o\vket{\widetilde C_o}.
\]
For $\Pi_o=|o\rangle\langle o|_A\otimes I$, its
stabilizer expansion therefore gives
\[
    J_0
    =2^{k_{\mathrm{in}}}
      \sum_o\Pi_o|\psi\rangle\langle\psi|\Pi_o
    =2^{-(n_R+n_{\mathrm{out}})}
      \sum_{\omega\in\mathcal W^{(\zeta)}}g(\omega).
\]
Conjugating $J_0$ by the propagated Pauli
$P(\beta^{(\zeta)}(e))$ now yields
\begin{equation}
        J_{e}
    =2^{-(n_R+n_{\mathrm{out}})}
      \sum_{\omega\in\mathcal W^{(\zeta)}}
      (-1)^{\chi_{\zeta}(e,\omega)}g(\omega). \label{eq: noisy choi}
\end{equation}
The $g(\omega)$ are mutually orthogonal, so their
coefficients are unique. Thus the Choi correspondence
\cite{chiribella2009theoretical} gives
\[
    \begin{aligned}
        \mathcal N_{e}=\mathcal N_{f}
        &\iff J_{e}=J_{f}\\
        &\iff \chi_{\zeta}(e,\omega)
             =\chi_{\zeta}( f,\omega)
             \quad\forall\,\omega\in\mathcal W^{(\zeta)}.
    \end{aligned}
\]
\end{proof}

The following two corollaries enable us to compute $\chi_{\zeta}$ explicitly for any level $\zeta\in \{0,1\}$.

\begin{corollary}[Evaluation at level $0$]
\label{prop:ft-level0-signature}
Consider the circuit of
\Cref{prop: Code conditioned stabilizers L0},
with $U=U_N\cdots U_1$.
Resolve preparation, the $N$ gates, and final measurement
into $\ell=N+2$ components.
Write $e_0=[x_0,z_0]$ for the fault label on the physical
input, $e_1$ for the label after preparation,
$e_{j+1}$ for the label after $U_j$, and $e_{N+2}$
for the label after measurement.
Define
\begin{equation}
    \begin{aligned}
        u_0&=[0^m,x_0,0^m,z_0]+e_1,\\
        u_j&=u_{j-1}F_{U_j}+e_{j+1},
        \qquad j\in\mathbb Z^+_N.
    \end{aligned}
    \label{eq:ft-level0-propagation}
\end{equation}
For $\omega=[b,p,p'] \in \mathcal{W}^{(0)}$ with $p'=[x',z']$,
\begin{equation}
    \chi_0(e,\omega)
    =
    \bigl\langle
        u_N+e_{N+2},[0^{n_R},x',b,z']
    \bigr\rangle_s.
    \label{eq:ft-level0-signature}
\end{equation}
\end{corollary}

\begin{proof}
\Cref{lem:ft-signature-completeness} enables us to propagate all faults to the output data-record registers. Moving the resulting Pauli through its
projector shifts the outcome index, leaving output
label $u_N+e_{N+2}$. Its commutation sign with
$R_\omega$ gives the formula.
\end{proof}

\begin{corollary}[Evaluation at level $1$]
\label{prop:ft-level1-signature}
Consider the measurement sequence of
\Cref{prop: code-conditioned stabilizers level 1}.
Let $e_r$ be the label of $P_r^Q(e)$ and
$[a_r,d_r]$ the label of $P_r^A(e)$.
For $\omega=[b,p,p']  \in \mathcal{W}^{(1)}$, let $\nu_r$ be its matched
data labels from \Cref{eq:L1-prefix-and-w}.
Then
\begin{equation}
    \chi_1(e,\omega)
    =
    \sum_{r=0}^{\ell}
        \langle e_r,\nu_{r+1}\rangle_s
    +\sum_{r=1}^{\ell}a_rb_r.
    \label{eq:ft-level1-signature}
\end{equation}
In particular, record $Z$ faults do not affect the signature.
\end{corollary}

\begin{proof}
The matched relation acts by $P(\nu_{r+1})$ on the
data after measurement $r$ and by $Z^{b_r}$ on its
record. Propagating these operators through the faults
contributes the displayed symplectic and record signs.
The composition rule of
\Cref{prop: stabilizers of composition} adds these signs.
\end{proof}

\begin{definition}[Fault signature]
\label{def:ft-error-signature}
The \emph{fault signature} of
$e\in\mathcal E^{(\zeta)}$ is the linear functional
\[
    \chi_{\zeta}(e,\cdot):
    \mathcal W^{(\zeta)}\longrightarrow\mathbb F_2.
\]
A relation $\omega \in \mathcal{X} \subseteq \mathcal{W}^{(\zeta)}$ is defined to be \emph{violated} 
when $\chi_{\zeta}(e,\omega)=1$.
\end{definition}

\begin{definition}[Operational meaning of fault signatures]
\label{rem:ft-operational-signatures}
For a subspace $\mathcal X\subseteq\mathcal W^{(\zeta)}$,
define its indicator $\delta_{\mathcal X}$ as
\begin{equation}
    \delta_{\mathcal X}(e)
    \triangleq
    \begin{cases}
        0,&\chi_{\zeta}(e,\omega)=0
           \text{ for every }\omega\in\mathcal X,\\
        1,&\text{otherwise}.
    \end{cases}
    \label{eq:ft-violation-indicator}
\end{equation}
Abbreviate the indicators for
$\mathcal W_{\mathrm{det}}$, $\mathcal W_{\mathrm{out}}$,
$\mathcal W_{\mathrm{obs}}$, and $\mathcal W^{(\zeta)}$ (c.f.
\Cref{def:common-stabilizer-relations})
as $\delta_{\mathrm{det}}$, $\delta_{\mathrm{out}}$,
$\delta_{\mathrm{obs}}$, and $\delta_{\mathcal W}$. The operational meaning of these indicators is presented in \Cref{tab:ft-operational-signatures}.
\end{definition}

\begin{table}[t]
\centering
\small
\setlength{\tabcolsep}{3pt}
\renewcommand{\arraystretch}{1.15}
\begin{tabular}{@{}ccccp{0.48\columnwidth}@{}}
\toprule
$\delta_{\mathrm{det}}$
& $\delta_{\mathrm{out}}$
& $\delta_{\mathrm{obs}}$
& $\delta_{\mathcal W}$
& Operational meaning\\
\midrule
1 & 1 & $*$ & 1
& A record parity detects the fault.\\
0 & 1 & $*$ & 1
& A noiseless output-data check detects the fault.\\
0 & 0 & 1 & 1
& A measured-input value is corrupted.\\
0 & 0 & 0 & 1
& Preserved logical information is corrupted.\\
0 & 0 & 0 & 0
& The channel is unchanged.\\
\bottomrule
\end{tabular}
\caption{
Classification of signature violations. An asterisk denotes that any value of that indicator does not change the operational meaning depicted in that row.
}
\label{tab:ft-operational-signatures}
\end{table}

\begin{definition}[Detection and logical-effect matrices]
\label{def:ft-signature-matrices}
Let $\epsilon_1,\ldots,\epsilon_{d_{\zeta}}$ be the coordinate
basis of $\mathcal E^{(\zeta)}$, and define
\begin{equation}
    \gamma_{\zeta}(\omega)
    \triangleq
    [\,\chi_{\zeta}(\epsilon_1,\omega),\ldots,
       \chi_{\zeta}(\epsilon_{d_{\zeta}},\omega)\,].
    \label{eq:ft-signature-row}
\end{equation}
Thus
$\gamma_{\zeta}(\omega)e^{\mathsf T}
=\chi_{\zeta}(e,\omega)$. Choose a basis $\omega_1,\ldots,\omega_h$ of
$\mathcal W_{\mathrm{out}}$ and extend it to a basis
$\omega_1,\ldots,\omega_q$ of $\mathcal W^{(\zeta)}$.
The \emph{detection matrix} and
\emph{logical-effect matrix} are then defined by
\begin{equation}
    H_{\zeta}\triangleq
    \begin{bmatrix}
        \gamma_{\zeta}(\omega_1)\\
        \vdots\\
        \gamma_{\zeta}(\omega_h)
    \end{bmatrix},
    \qquad
    L_{\zeta}\triangleq
    \begin{bmatrix}
        \gamma_{\zeta}(\omega_{h+1})\\
        \vdots\\
        \gamma_{\zeta}(\omega_q)
    \end{bmatrix}.
    \label{eq:ft-HL-matrices}
\end{equation}
For each fault pattern $e$, its
\emph{syndrome} and \emph{logical-effect signature} are defined as
\[
    \sigma_{\zeta}(e)\triangleq H_{\zeta}e^{\mathsf T},
    \qquad
    \lambda_{\zeta}(e)\triangleq L_{\zeta}e^{\mathsf T}.
\]
The syndrome records all input-independent output-check
violations. Together, the two vectors specify the full
fault signature and hence the channel
$\mathcal N_{e}$.
\end{definition}

\begin{remark}
    Changing the chosen relation bases gives
\begin{equation}
    H'_{\zeta}=U_HH_{\zeta},
    \qquad
    L'_{\zeta}=U_LL_{\zeta}+RH_{\zeta},
    \label{eq:ft-change-representatives}
\end{equation}
with $U_H,U_L$ invertible.
Thus zero syndrome and equality of logical-effect
signatures at a fixed syndrome are independent of
these choices.
\end{remark}

\begin{proposition}[Detection and logical effect matrices and indicators violation]
\label{prop:ft-matrix-indicators}
For every $e\in\mathcal E^{(\zeta)}$,
\[
    \begin{aligned}
        H_{\zeta}e^{\mathsf T}=0
        &\iff \delta_{\mathrm{out}}(e)=0,\\
        H_{\zeta}e^{\mathsf T}
        =L_{\zeta}e^{\mathsf T}=0
        &\iff \delta_{\mathcal W}(e)=0.
    \end{aligned}
\]
Consequently, the four signature cases imply the
indicator values in \Cref{tab:ft-matrix-meaning}.
\end{proposition}

\begin{proof}
The rows of $H_{\zeta}$ evaluate the signature on a basis
of $\mathcal W_{\mathrm{out}}$, while the rows of
$H_{\zeta}$ and $L_{\zeta}$ together evaluate it on a basis
of $\mathcal W^{(\zeta)}$. Linearity gives both equivalences.

If both signatures vanish, all four indicators vanish.
If only the syndrome vanishes, then
$\delta_{\mathrm{det}}=\delta_{\mathrm{out}}=0$
and $\delta_{\mathcal W}=1$, while
$\delta_{\mathrm{obs}}$ may be zero or one.
If the syndrome is nonzero, then
$\delta_{\mathrm{out}}=\delta_{\mathcal W}=1$
regardless of whether the logical-effect signature
vanishes; neither $\delta_{\mathrm{det}}$ nor
$\delta_{\mathrm{obs}}$ is determined in general.
\end{proof}

\begin{table}[t]
\centering
\small
\setlength{\tabcolsep}{3pt}
\renewcommand{\arraystretch}{1.4}
\begin{tabular}{@{}ccc@{}}
\toprule
$H_{\zeta}e^{\mathsf T}$
& $L_{\zeta}e^{\mathsf T}$
& Indicator violation\\
\midrule
$0$ & $0$
& $\begin{gathered}
    \delta_{\mathrm{det}}=\delta_{\mathrm{out}}=0\\
    \delta_{\mathrm{obs}}=\delta_{\mathcal W}=0
   \end{gathered}$\\
$0$ & $\neq0$
& $\begin{gathered}
    \delta_{\mathrm{det}}=\delta_{\mathrm{out}}=0\\
    \delta_{\mathrm{obs}}=*,\quad\delta_{\mathcal W}=1
   \end{gathered}$\\
$\neq0$ & $0$
& $\begin{gathered}
    \delta_{\mathrm{out}}=\delta_{\mathcal W}=1\\
    \delta_{\mathrm{det}}=*,\quad\delta_{\mathrm{obs}}=*
   \end{gathered}$\\
$\neq0$ & $\neq0$
& $\begin{gathered}
    \delta_{\mathrm{out}}=\delta_{\mathcal W}=1\\
    \delta_{\mathrm{det}}=*,\quad\delta_{\mathrm{obs}}=*
   \end{gathered}$\\
\bottomrule
\end{tabular}
\caption{
An asterisk denotes an indicator not determined by
the displayed signature conditions on $H_{\zeta}e^{\mathsf T}$ and $L_{\zeta}e^{\mathsf T}$. 
}
\label{tab:ft-matrix-meaning}
\end{table}

\subsection{Exact recovery and circuit distance}
\label{subsec:ft-recovery-distance}

We begin by establishing necessary and sufficient conditions for fault tolerance of non-adaptive stabilizer circuits. Then following the footsteps of \cite{beverland2024faulttolerancestabilizerchannels,kliuchnikov2023stabilizercircuitverification} we define circuit distance in terms of the detection and logical effect matrices.

\begin{theorem}[Exact recovery]
\label{thm:ft-HL-criterion}
Let
\[
    \mu_{e}\geq0,
    \qquad
    \sum_{e\in\mathcal E^{(\zeta)}}\mu_{e}=1,
    \qquad
    \mathcal N_\mu
    \triangleq
    \sum_{e\in\mathcal E^{(\zeta)}}
        \mu_{e}\mathcal N_{e}.
\]
A CPTP recovery $\mathcal R$ on $H_{\mathrm{out}}^{(\zeta)}$
satisfying
$\mathcal R\circ\mathcal N_\mu=\mathcal N_0$
exists if and only if, for every
$e, f\in\mathcal E^{(\zeta)}$
with $\mu_{e},\mu_{ f}>0$,
\begin{equation}
    H_{\zeta}e^{\mathsf T}=H_{\zeta}f^{\mathsf T}
    \quad\Longrightarrow\quad
    L_{\zeta}e^{\mathsf T}=L_{\zeta} f^{\mathsf T}.
    \label{eq:ft-HL-pairwise-criterion}
\end{equation}
Operationally, faults having the same syndrome must have the same logical
effect for a common correction based on their syndrome to succeed.
\end{theorem}

\begin{proof}
($\implies$) Suppose recovery exists and pick arbitrary $e ,  f \in \mathcal{E}^{(\zeta)}$ with $H_{\zeta}e^{\mathsf T}=H_{\zeta} f^{\mathsf T}$ and $\mu_{e},\mu_{ f} \neq 0$. We must show that $L_{\zeta}e^{\mathsf T}=L_{\zeta} f^{\mathsf T}$. Now assume $L_{\zeta}e^{\mathsf T}\neq L_{\zeta}  f^{\mathsf T}$. Then by \Cref{def:ft-signature-matrices}, there exists $\omega_i \in \{\omega_{h+1},\cdots, \omega_q\}$ such that $\chi_{\zeta}(e,\omega_i) \neq \chi_{\zeta}( f,\omega_i)$. Now letting $\omega_i = [b_i,p_i,p'_i]$, \Cref{eq:ft-signed-output-relation} gives for all $j$, 
\begin{equation*}
    \begin{gathered}
        P(p_i)\widetilde C^\dagger_j \widetilde C_j = \widetilde C^\dagger_jR_{\omega_i}\widetilde C_j = \widetilde C^\dagger_j\widetilde C_jP(p_i) \\ \implies [\widetilde C^\dagger_j \widetilde C_j,P(p_i)] = 0.
    \end{gathered}
\end{equation*}
Now by \Cref{Thm: Iff Fault tolerance}, existence of recovery implies 
\begin{equation*}
    \begin{gathered}
        F^\dagger_{ f,o'}F_{e,o} \in \operatorname{span} \{\widetilde C^\dagger_j\widetilde C_j\}, \quad \forall o,o' \\\implies [F^\dagger_{ f,o'}F_{e,o},P(p_i)] = 0.
    \end{gathered}
\end{equation*}
Again, by \Cref{eq:ft-noisy-common-relation} we have
\begin{equation*}
   \begin{aligned}
       P(p_i)F^\dagger_{ f,o'}F_{e,o} &=  (-1)^{\chi_{\zeta}( f,\omega_i)}F^\dagger_{ f,o'}R_{\omega_i}F_{e,o} \\
       &= (-1)^{\chi_{\zeta}( f,\omega_i) + \chi_{\zeta}(e,\omega_i)}F^\dagger_{ f,o'}F_{e,o} P(p_i).
   \end{aligned}
\end{equation*}
However, $\chi_{\zeta}(e,\omega_i)\neq\chi_{\zeta}( f,\omega_i)$ implies $\{F^\dagger_{f,o'}F_{e,o},P(p_i)\} = 0$. This forces $F^\dagger_{ f,o'}F_{e,o} = 0$. Put $\sigma=H_{\zeta}e^{\mathsf T}
=H_{\zeta} f^{\mathsf T}$ and taking $\omega_j \in \{\omega_1,\cdots , \omega_h\}$ from \Cref{def:ft-signature-matrices} define
\[
    \Pi_\sigma
    \triangleq
    \prod_{j=1}^{h}
        \frac{I+(-1)^{\sigma_j}R_{\omega_j}}{2} \neq 0.
\]
Since $\operatorname{Tr}(P(p)^{\mathsf T})
=2^{k_{\mathrm{in}}}\delta_{p,0}$, tracing out the
input registers in $J_{e},J_{ f}$ in
\Cref{eq: noisy choi} gives, for
$ a\in\{e, f\}$,
\begin{equation*}
    \begin{aligned}
        \sum_oF_{ a,o}F_{ a,o}^\dagger
        &=2^{k_{\mathrm{in}}-n_R-n_{\mathrm{out}}}
          \sum_{\omega\in\mathcal W_{\mathrm{out}}}
          (-1)^{\chi_{\zeta}( a,\omega)}R_\omega\\
        &=c\Pi_\sigma , \quad c\triangleq2^{k_{\mathrm{in}}+h-n_R-n_{\mathrm{out}}}>0.
    \end{aligned}
\end{equation*}
The last equality follows by expanding the product
defining $\Pi_\sigma$.
Trace preservation gives
$c\operatorname{Tr}\Pi_\sigma=2^{k_{\mathrm{in}}}>0$.
However, the vanishing cross-products imply
\[
    \begin{aligned}
        c^2\Pi_\sigma
        &=\sum_{o',o}
          F_{ f,o'}
          \bigl(F_{ f,o'}^\dagger
                F_{e,o}\bigr)
          F_{e,o}^\dagger\\
        &=0,
    \end{aligned}
\]
a contradiction. Hence
$L_{\zeta}e^{\mathsf T}=L_{\zeta} f^{\mathsf T}$.($\impliedby$) The converse directly follows from \Cref{lem:ft-signature-completeness}. Now assume \Cref{eq:ft-HL-pairwise-criterion} holds. For each syndrome $\sigma$ attained by a fault of
positive probability, choose such a fault
$e_\sigma$ and set
\[
    U_\sigma
    \triangleq
    P(\beta^{(\zeta)}(e_\sigma))^\dagger.
\]
Set $U_\sigma=I$ for all remaining syndromes and define
\[
    \mathcal R(\tau)
    \triangleq
    \sum_{\sigma\in\mathbb F_2^h}
        U_\sigma\Pi_\sigma\tau\Pi_\sigma U_\sigma^\dagger.
\]
The projectors $\Pi_\sigma$ are mutually orthogonal
and sum to $I$, so $\mathcal R$ is CPTP. Fix $e$ with $\mu_{e}>0$ and put
$\sigma=H_{\zeta}e^{\mathsf T}$.
The assumed condition gives
$L_{\zeta}e^{\mathsf T}
=L_{\zeta}e_\sigma^{\mathsf T}$, hence
$\mathcal N_{e}=\mathcal N_{e_\sigma}$
by \Cref{lem:ft-signature-completeness}.
Moreover, since $\omega_1,\ldots,\omega_h$ have zero input labels,
\eqref{eq:ft-noisy-common-relation} gives
$\Pi_\sigma F_{e,o}=F_{e,o}$.
Therefore, using \eqref{eq: stab noise propagation},
\[
    \begin{aligned}
        \mathcal R(\mathcal N_{e}(\rho))
        &=U_\sigma\mathcal N_{e}(\rho)U_\sigma^\dagger\\
        &=U_\sigma\mathcal N_{e_\sigma}(\rho)
          U_\sigma^\dagger\\
        &=\mathcal N_0(\rho).
    \end{aligned}
\]
Thus one recovery corrects every positive-probability
fault. By linearity and $\sum_{e}\mu_{e}=1$,
$\mathcal R\circ\mathcal N_\mu=\mathcal N_0$.
\end{proof}

\begin{definition}[Spacetime Pauli weight]
\label{def:spacetime-pauli-weight}
\label{def:spacetime Pauli weight}
For $e\in\mathcal E^{(\zeta)}$, we define its
\emph{spacetime weight} as
\begin{equation}
    \operatorname{wt}_{\mathrm{ST}}(e)
    \triangleq
    \bigl|\{(s,t)\in\Lambda^{(\zeta)}:
                  e_{s,t}\neq[0,0]\}\bigr|.
    \label{eq:ft-spacetime-weight}
\end{equation}
\end{definition}

\begin{definition}[Circuit distance]
\label{def:ft-circuit-distance}
Fix $\zeta\in\{0,1\}$ and a binary linear fault subspace
$\mathcal V^{(\zeta)}\subseteq\mathcal E^{(\zeta)}$.
The \emph{circuit distance} of $\widetilde{\mathbf C}$
relative to $\mathcal V^{(\zeta)}$ is
\begin{equation}
    d_{\mathrm{circ},\zeta}
    \triangleq
    \min_{\substack{
        e\in\mathcal V^{(\zeta)}\\
        H_{\zeta}e^{\mathsf T}=0\\
        L_{\zeta}e^{\mathsf T}\neq0}}
        \operatorname{wt}_{\mathrm{ST}}(e),
    \label{eq:ft-HL-distance}
\end{equation}
with $\min\varnothing=\infty$.
Operationally, it is the minimum weight of an allowed spacetime fault
with trivial syndrome and nontrivial logical effect.
\end{definition}

\begin{proposition}[Exact recovery of faults with bounded weight]
\label{cor:ft-correction-radius}
For every integer $t\geq0$ with
$2t<d_{\mathrm{circ},\zeta}$, one CPTP recovery corrects
all patterns $e\in\mathcal V^{(\zeta)}$ satisfying
$\operatorname{wt}_{\mathrm{ST}}(e)\leq t$,
and hence every stochastic mixture of these patterns.
In particular, for finite distance, all faults up to
weight
\[
    \left\lfloor\frac{d_{\mathrm{circ},\zeta}-1}{2}\right\rfloor
\]
are correctable.
\end{proposition}

\begin{proof}
Since $\mathcal V^{(\zeta)}$ is a vector space, closure under addition gives, for all $e ,  f \in \mathcal V^{(\zeta)}$ satisfying $\operatorname{wt}_{\mathrm{ST}}(e),\operatorname{wt}_{\mathrm{ST}}( f)\leq t$,
$e+ f\in\mathcal V^{(\zeta)}$ and
\[
    \operatorname{wt}_{\mathrm{ST}}(e+ f)
    \leq
    \operatorname{wt}_{\mathrm{ST}}(e)
    +\operatorname{wt}_{\mathrm{ST}}( f)
    \leq2t<d_{\mathrm{circ},\zeta}.
\]
If their syndromes agree, the distance definition
therefore forces
$L_{\zeta}(e+ f)^{\mathsf T}=0$.
Their logical-effect signatures consequently agree,
and \Cref{thm:ft-HL-criterion} supplies a common recovery.
\end{proof}

In \Cref{app:honeycomb}, we explicitly
compute this distance for the finite honeycomb circuit
under its specified fault locations.


\subsection{Weight enumerators}
\label{sec:weight-enumerators}

For each non-adaptive stabilizer circuit in a prescribed family, we use the signature
matrices $H_{\zeta},L_{\zeta}$ (\Cref{def:ft-signature-matrices}) to count
Pauli faults by their weights in blocks of spacetime locations,
detection syndrome, and logical effect. These counts define
nonnegative integer-valued distributions. We derive
Fourier--MacWilliams identities using character duality for
additive codes
\cite{MacWilliams1963Weights,GluesingLuerssen2015Fourier}
and combine them with circuit and boundary constraints.
Related quantum and circuit enumerators appear in
\cite{ShorLaflamme1997MacWilliams,Rains1998QuantumWeightEnumerators,
CaoLackey2024TensorEnumerators,KuklianskyLackey2024CircuitEnumerators}.
In \Cref{sec:linear-programming-bounds}, we relax the counts
to real variables and construct linear programs that give
upper bounds on circuit distance, lower bounds on undetected
logical failure probability, and upper bounds on pseudo-threshold
throughout the family under stochastic Pauli noise model.

\subsubsection{Circuit families and weight distributions}

We begin by defining a family of non-adaptive stabilizer circuits implementing the same linear map on some given input space.

\begin{definition}[Non-adaptive stabilizer circuit family and fault space]
\label{def:we-circuit-family}
Fix $\zeta\in\{0,1\}$ and a family
$\{\mathbf C_\theta:\theta\in\Theta\}$ of non-adaptive stabilizer
circuits on common input and output registers, with a fixed
noiseless encoder $C_{\mathrm{enc}}$ and the complete measurement record
retained. All members implement the same quantum channel, namely, $\forall \theta \in \Theta$ it holds
\begin{equation}
 \sum_o\widetilde C_{\theta,o}\rho\widetilde C_{\theta,o}^{\dagger}
 \triangleq \mathcal N_0(\rho),\qquad
 \widetilde C_{\theta,o}\triangleq C_{\theta,o}C_{\mathrm{enc}}.
 \label{eq:we-common-ideal-channel}
\end{equation}
Each circuit $\mathbf C_\theta$ has permitted spacetime locations and a partition $\Pi_{\theta}$ given by
\begin{equation}
 \begin{gathered}
 \Lambda_\theta^{(\zeta)}\triangleq \bigsqcup_{b=1}^{B_\theta}I_{\theta,b},\quad
 \Pi_\theta\triangleq \{I_{\theta,1},\ldots,I_{\theta,B_\theta}\}, \quad \forall b\in \mathbb Z^+_{B_{\theta}} , I_{\theta,b}\neq \varnothing,\\
N_{\theta,b}\triangleq |I_{\theta,b}|,\qquad
 |\Lambda_\theta^{(\zeta)}| \triangleq N_\theta\triangleq \sum_bN_{\theta,b}\leq N^{\Theta}_{\max}\in \mathbb Z<\infty.
 \end{gathered}
 \label{eq:we-location-partition}
\end{equation}
The spacetime fault space for each member is given by
$\mathcal E_\theta^{(\zeta)} \cong \mathbb F_2^{2N_{\theta}}$
with dimension $d_{\zeta,\theta}\triangleq 2N_\theta$.
The permitted patterns form a binary subspace
$\mathcal V_\theta^{(\zeta)}\subseteq\mathcal E_\theta^{(\zeta)}$.
\end{definition}
\begin{definition}[Block weight and weight profile]
\label{def:weight-profile}
For $e\in\mathcal E_\theta^{(\zeta)}$ and
$b\in\mathbb Z^+_{B_\theta}$, define its \emph{block weight} by
\begin{equation}
 \operatorname{wt}_{\theta,b}(e)
 \triangleq
 \left|\{(s,t)\in I_{\theta,b}:
                    e_{s,t}\ne0\}\right|
 \leq N_{\theta,b}.
 \label{eq: Block weight}
\end{equation}
Its \emph{weight profile} under $\Pi_\theta$ is
\begin{equation}
 \operatorname{wt}_{\Pi_\theta}(e)
 \triangleq
 \bigl(\operatorname{wt}_{\theta,1}(e),\ldots,
       \operatorname{wt}_{\theta,B_\theta}(e)\bigr).
 \label{eq: weight profile}
\end{equation}
The set of all weight profiles is defined as
\begin{equation}
 J_{\Pi_\theta}\triangleq
 \left\{
 (w_1,\ldots,w_{B_\theta})\in\mathbb Z^{B_\theta}:
 0\leq w_b\leq N_{\theta,b}\ \forall b\in\mathbb Z^+_{B_\theta}
 \right\}.
\end{equation}
\end{definition}
\begin{definition}[Syndrome-trivial and channel-trivial spaces]
For circuit $\mathbf C_\theta$, construct the detection
matrix $H_{\zeta,\theta}$ and logical-effect matrix
$L_{\zeta,\theta}$ by \Cref{def:ft-signature-matrices}, with
columns indexed by the coordinates of
$\mathcal E_\theta^{(\zeta)}$, and define
\[
 M_{\zeta,\theta}\triangleq
 \begin{bmatrix}
 H_{\zeta,\theta}\\
 L_{\zeta,\theta}
 \end{bmatrix}.
\]
For the permitted binary linear fault space
$\mathcal V_\theta^{(\zeta)}\subseteq\mathcal E_\theta^{(\zeta)}$,
define the \emph{syndrome-trivial space} and
\emph{channel-trivial space} by
\begin{equation}
 \begin{aligned}
 \mathcal K_{\zeta,\theta}
   &\triangleq
   \{e\in\mathcal V_\theta^{(\zeta)}:
                H_{\zeta,\theta}e^{\mathsf T}=0\},\\
 \mathcal T_{\zeta,\theta}
   &\triangleq
   \{e\in\mathcal V_\theta^{(\zeta)}:
                M_{\zeta,\theta}e^{\mathsf T}=0\}.
 \end{aligned}
 \label{eq:we-fault-subspaces}
\end{equation}
Define their character duals by
\begin{equation}
 \begin{aligned}
 \mathcal K_{\zeta,\theta}^{\perp}
 &\triangleq
 \left\{x\in\mathcal E_\theta^{(\zeta)}:
 xe^{\mathsf T}=0
 \ \forall\,e\in\mathcal K_{\zeta,\theta}\right\},\\
 \mathcal T_{\zeta,\theta}^{\perp}
 &\triangleq
 \left\{x\in\mathcal E_\theta^{(\zeta)}:
 xe^{\mathsf T}=0
 \ \forall\,e\in\mathcal T_{\zeta,\theta}\right\},
 \end{aligned}
\end{equation}
where $xe^{\mathsf T}$ is the ordinary binary dot product.
\end{definition}

\begin{definition}[Weight enumerators and generating functions]
\label{def:Weight-enumerators-circuit}
For $\mathbf w\in J_{\Pi_\theta}$, define the weight enumerators $U_{\theta,\mathbf w},T_{\theta,\mathbf w},\widehat U_{\theta,\mathbf w},\widehat T_{\theta,\mathbf w}$ as
\begin{equation}
 \begin{aligned}
 U_{\theta,\mathbf w}
   &\triangleq
   |\{e\in\mathcal K_{\zeta,\theta}:
      \operatorname{wt}_{\Pi_\theta}(e)=\mathbf w\}|,\\
 T_{\theta,\mathbf w}
   &\triangleq
   |\{e\in\mathcal T_{\zeta,\theta}:
      \operatorname{wt}_{\Pi_\theta}(e)=\mathbf w\}|,\\
 \widehat U_{\theta,\mathbf w}
   &\triangleq
   |\{e\in\mathcal K_{\zeta,\theta}^{\perp}:
      \operatorname{wt}_{\Pi_\theta}(e)=\mathbf w\}|,\\
 \widehat T_{\theta,\mathbf w}
   &\triangleq
   |\{e\in\mathcal T_{\zeta,\theta}^{\perp}:
      \operatorname{wt}_{\Pi_\theta}(e)=\mathbf w\}|.
 \end{aligned}
 \label{eq:we-nested-distributions}
\end{equation}
The difference $U_{\theta,\mathbf w}-T_{\theta,\mathbf w}$
counts harmful undetected faults with weight profile
$\mathbf w$. Let $\sigma,\lambda$ be binary column vectors with
lengths equal to the row counts of $H_{\zeta,\theta}$ and
$L_{\zeta,\theta}$, respectively. Define the
\emph{signature-resolved weight coefficients} and their
generating function by
\begin{align}
 A_{\theta,\mathbf w}^{(\zeta)}(\sigma,\lambda)
 &\triangleq
 \left|\left\{
 e\in\mathcal V_\theta^{(\zeta)}:
 \begin{array}{l}
 H_{\zeta,\theta}e^{\mathsf T}=\sigma,\\
 L_{\zeta,\theta}e^{\mathsf T}=\lambda,\\
 \operatorname{wt}_{\Pi_\theta}(e)=\mathbf w
 \end{array}
 \right\}\right|,
 \label{eq:we-signature-coefficients}\\
 W_{\theta;\sigma,\lambda}^{(\zeta),\Pi_\theta}
       (u,v)
 &\triangleq
 \sum_{\mathbf w\in J_{\Pi_\theta}}
 A_{\theta,\mathbf w}^{(\zeta)}(\sigma,\lambda)
 \prod_{b=1}^{B_\theta}
 u_b^{N_{\theta,b}-w_b}v_b^{w_b}.
 \label{eq:we-signature-polynomial}
\end{align}
Here $u=(u_1,\ldots,u_{B_\theta})$ and
$v=(v_1,\ldots,v_{B_\theta})$ are tuples of
formal variables. Unattainable signatures have zero
coefficients. In particular,
\[
 U_{\theta,\mathbf w}
   =\sum_\lambda A_{\theta,\mathbf w}^{(\zeta)}(0,\lambda),
 \qquad
 T_{\theta,\mathbf w}
   =A_{\theta,\mathbf w}^{(\zeta)}(0,0),
\]
where the sum runs over all logical-effect signatures.
\end{definition}

\begin{proposition}[Dual spaces and cardinalities]
\label{lem:we-dual-row-spaces}
For $A\in\{H_{\zeta,\theta},M_{\zeta,\theta}\}$, we denote
$A|_{\mathcal V_\theta^{(\zeta)}}$ as a linear map
$e\mapsto Ae^{\mathsf T}$ with domain
$\mathcal V_\theta^{(\zeta)}$ and define its rank as the dimension
of its image. Taking all duals in $\mathcal E_\theta^{(\zeta)}$,
\begin{align}
 \mathcal K_{\zeta,\theta}^{\perp}
   &=(\mathcal V_\theta^{(\zeta)})^\perp
       +\operatorname{row}H_{\zeta,\theta}, \notag\\
 \mathcal T_{\zeta,\theta}^{\perp}
   &=(\mathcal V_\theta^{(\zeta)})^\perp
       +\operatorname{row}M_{\zeta,\theta},
 \label{eq:we-dual-inclusion}\\
 |\mathcal K_{\zeta,\theta}|
   &=2^{\dim\mathcal V_\theta^{(\zeta)}
       -\operatorname{rank}
        (H_{\zeta,\theta}|_{\mathcal V_\theta^{(\zeta)}})}, \notag\\
 |\mathcal T_{\zeta,\theta}|
   &=2^{\dim\mathcal V_\theta^{(\zeta)}
       -\operatorname{rank}
        (M_{\zeta,\theta}|_{\mathcal V_\theta^{(\zeta)}})},
 \label{eq:we-nested-cardinalities}\\
 |\mathcal K_{\zeta,\theta}^{\perp}|
   &=\frac{2^{d_{\zeta,\theta}}}{|\mathcal K_{\zeta,\theta}|},
 \notag\\
 |\mathcal T_{\zeta,\theta}^{\perp}|
   &=\frac{2^{d_{\zeta,\theta}}}{|\mathcal T_{\zeta,\theta}|}.
 \label{eq:we-dual-cardinalities}
\end{align}
\end{proposition}

\begin{proof}
By definition,
\[
 \begin{aligned}
 \mathcal K_{\zeta,\theta}
   &=\mathcal V_\theta^{(\zeta)}
       \cap(\operatorname{row}H_{\zeta,\theta})^\perp,\\
 \mathcal T_{\zeta,\theta}
   &=\mathcal V_\theta^{(\zeta)}
       \cap(\operatorname{row}M_{\zeta,\theta})^\perp.
 \end{aligned}
\]
The orthogonal complement of an intersection is the sum
of the orthogonal complements, and taking the orthogonal
complement twice returns the original subspace.
Applying these facts give \Cref{eq:we-dual-inclusion}. The spaces $\mathcal K_{\zeta,\theta}$ and
$\mathcal T_{\zeta,\theta}$ are the kernels of
$H_{\zeta,\theta}|_{\mathcal V_\theta^{(\zeta)}}$ and
$M_{\zeta,\theta}|_{\mathcal V_\theta^{(\zeta)}}$, respectively.
Rank--nullity therefore gives their dimensions.
Since a binary linear space has cardinality two raised
to its dimension, \Cref{eq:we-nested-cardinalities} follows. Finally, non-degeneracy of the binary dot product gives
\[
 \begin{aligned}
 \dim\mathcal K_{\zeta,\theta}^{\perp}
   &=d_{\zeta,\theta}-\dim\mathcal K_{\zeta,\theta},\\
 \dim\mathcal T_{\zeta,\theta}^{\perp}
   &=d_{\zeta,\theta}-\dim\mathcal T_{\zeta,\theta}.
 \end{aligned}
\]
Exponentiating these equalities with base two proves
\Cref{eq:we-dual-cardinalities}.
\end{proof}

\subsubsection{MacWilliams identities and constraints}
We next derive constraints on these enumerators by combining
the character-duality approach to MacWilliams identities
\cite{GluesingLuerssen2015Fourier,CaoLackey2024TensorEnumerators,
KuklianskyLackey2024CircuitEnumerators}
with the circuit signature characterization.
\begin{proposition}[Split MacWilliams identities]
\label{prop:we-signature-macwilliams}
\label{cor:we-split-macwilliams}
\label{lem:we-krawtchouk-coefficients}
For integers $q\geq2$, $n\geq0$, and $0\leq i,j\leq n$,
the \emph{$q$-ary Krawtchouk polynomials}
\cite{GluesingLuerssen2015Fourier} are defined as
\begin{equation}
 \mathbb K_j^{(q)}(i;n)
 \triangleq
 \sum_{a=0}^{j}(-1)^a(q-1)^{j-a}
       \binom ia\binom{n-i}{j-a},
 \label{eq:we-krawtchouk}
\end{equation}
where the summand vanishes whenever $a>i$ or $j-a>n-i$.
For $\mathbf j,\mathbf w\in J_{\Pi_\theta}$, define
\begin{equation}
 \mathbb K_{\mathbf j}^{\Pi_\theta}(\mathbf w)
 \triangleq
 \prod_{b=1}^{B_\theta}
 \mathbb K_{j_b}^{(4)}(w_b;N_{\theta,b}).
 \label{eq:we-split-krawtchouk}
\end{equation}
For every circuit $\mathbf C_\theta$, the split MacWilliams identities are given by
\begin{equation}
 \begin{aligned}
 \widehat U_{\theta,\mathbf j}
   &=\frac{1}{|\mathcal K_{\zeta,\theta}|}
     \sum_{\mathbf w\in J_{\Pi_\theta}}
       U_{\theta,\mathbf w}
       \mathbb K_{\mathbf j}^{\Pi_\theta}(\mathbf w),\\
 \widehat T_{\theta,\mathbf j}
   &=\frac{1}{|\mathcal T_{\zeta,\theta}|}
     \sum_{\mathbf w\in J_{\Pi_\theta}}
       T_{\theta,\mathbf w}
       \mathbb K_{\mathbf j}^{\Pi_\theta}(\mathbf w).
 \end{aligned}
 \label{eq:we-split-macwilliams}
\end{equation}
\end{proposition}

\begin{proof}
For every $x\in\mathcal E_\theta^{(\zeta)}$, it holds
\begin{equation}
 \label{eq:we-character-orthogonality}
 \sum_{e\in\mathcal K_{\zeta,\theta}}
       (-1)^{xe^{\mathsf T}}
 =
 \begin{cases}
 |\mathcal K_{\zeta,\theta}|,&
       x\in\mathcal K_{\zeta,\theta}^{\perp},\\
 0,&\text{otherwise}.
 \end{cases}
\end{equation}
For each $(s,t)\in\Lambda_\theta^{(\zeta)}$,
\begin{equation}
 \sum_{x_{s,t}\in\mathbb F_2^2\setminus\{0\}}
 (-1)^{x_{s,t}e_{s,t}^{\mathsf T}}
 =
 \begin{cases}
 3,&e_{s,t}=0,\\
 -1,&e_{s,t}\ne0.
 \end{cases}
 \label{eq:we-local-character-sum}
\end{equation}
Using \Cref{eq:we-character-orthogonality} and then factoring
the sum over the location labels gives
\[
\begin{aligned}
 &\sum_{\mathbf j\in J_{\Pi_\theta}}
   \widehat U_{\theta,\mathbf j}
   \prod_bu_b^{N_{\theta,b}-j_b}v_b^{j_b}\\
 &=\frac{1}{|\mathcal K_{\zeta,\theta}|}
   \sum_{e\in\mathcal K_{\zeta,\theta}}
   \sum_{x\in\mathcal E_\theta^{(\zeta)}}
   (-1)^{xe^{\mathsf T}}
   \prod_b
   u_b^{N_{\theta,b}-\operatorname{wt}_{\theta,b}(x)}
   v_b^{\operatorname{wt}_{\theta,b}(x)}\\
 &=\frac{1}{|\mathcal K_{\zeta,\theta}|}
   \sum_{e\in\mathcal K_{\zeta,\theta}}
   \prod_b\prod_{(s,t)\in I_{\theta,b}}
   \left(
     u_b+v_b
     \sum_{x_{s,t}\in\mathbb F_2^2\setminus\{0\}}
     (-1)^{x_{s,t}e_{s,t}^{\mathsf T}}
   \right)\\
 &=\frac{1}{|\mathcal K_{\zeta,\theta}|}
   \sum_{\mathbf w\in J_{\Pi_\theta}}
   U_{\theta,\mathbf w}
   \prod_b
   (u_b+3v_b)^{N_{\theta,b}-w_b}(u_b-v_b)^{w_b}.
\end{aligned}
\]
The last equality uses \Cref{eq:we-local-character-sum}
and groups $e$ by its weight profile.
By the binomial theorem and
\Cref{eq:we-krawtchouk,eq:we-split-krawtchouk},
\[
 \prod_b
 (u_b+3v_b)^{N_{\theta,b}-w_b}(u_b-v_b)^{w_b}
 =
 \sum_{\mathbf j\in J_{\Pi_\theta}}
 \mathbb K_{\mathbf j}^{\Pi_\theta}(\mathbf w)
 \prod_bu_b^{N_{\theta,b}-j_b}v_b^{j_b}.
\]
Substitution and comparison of coefficients give
\[
 \widehat U_{\theta,\mathbf j}
 =
 \frac{1}{|\mathcal K_{\zeta,\theta}|}
 \sum_{\mathbf w\in J_{\Pi_\theta}}
 U_{\theta,\mathbf w}
 \mathbb K_{\mathbf j}^{\Pi_\theta}(\mathbf w).
\]
Replacing $\mathcal K_{\zeta,\theta}$ by
$\mathcal T_{\zeta,\theta}$ proves the second identity in \Cref{eq:we-split-macwilliams}.
\end{proof}

\begin{proposition}[Counting and distance constraints]
\label{prop:we-enumerator-constraints}
For $\mathbf w\in J_{\Pi_\theta}$, write
\[
 \Gamma_{\theta,\mathbf w}
 \triangleq
 \prod_{b=1}^{B_\theta}
 3^{w_b}\binom{N_{\theta,b}}{w_b},
 \qquad
 |\mathbf w|\triangleq\sum_{b=1}^{B_\theta}w_b.
\]
For every circuit $\mathbf C_\theta$, the four weight
distributions are nonnegative integer-valued and satisfy
\begin{equation}
 \begin{aligned}
 0&\leq T_{\theta,\mathbf w}
   \leq U_{\theta,\mathbf w}
   \leq\Gamma_{\theta,\mathbf w},\\
 0&\leq\widehat U_{\theta,\mathbf w}
   \leq\widehat T_{\theta,\mathbf w}
   \leq\Gamma_{\theta,\mathbf w},
 \end{aligned}
 \label{eq:lp-capacity}
\end{equation}
\begin{equation}
 U_{\theta,\mathbf0}=T_{\theta,\mathbf0}
 =\widehat U_{\theta,\mathbf0}
 =\widehat T_{\theta,\mathbf0}=1,
 \label{eq:lp-zero}
\end{equation}
and
\begin{equation}
 \begin{aligned}
 \sum_{\mathbf w\in J_{\Pi_\theta}}U_{\theta,\mathbf w}
   &=|\mathcal K_{\zeta,\theta}|,&
 \sum_{\mathbf w\in J_{\Pi_\theta}}T_{\theta,\mathbf w}
   &=|\mathcal T_{\zeta,\theta}|,\\
 \sum_{\mathbf w\in J_{\Pi_\theta}}\widehat U_{\theta,\mathbf w}
   &=\frac{2^{d_{\zeta,\theta}}}{|\mathcal K_{\zeta,\theta}|},&
 \sum_{\mathbf w\in J_{\Pi_\theta}}\widehat T_{\theta,\mathbf w}
   &=\frac{2^{d_{\zeta,\theta}}}{|\mathcal T_{\zeta,\theta}|}.
 \end{aligned}
 \label{eq:lp-cardinalities}
\end{equation}
Furthermore, for every integer $D\geq1$,
\begin{equation}
 d_{\mathrm{circ},\zeta}(\theta)\geq D
 \quad\Longleftrightarrow\quad
 U_{\theta,\mathbf w}=T_{\theta,\mathbf w}
 \quad\text{for all }0<|\mathbf w|<D.
 \label{eq:we-distance-equalities}
\end{equation}
\end{proposition}

\begin{proof}
The inclusions
\[
 \mathcal T_{\zeta,\theta}\subseteq\mathcal K_{\zeta,\theta},
 \qquad
 \mathcal K_{\zeta,\theta}^{\perp}
 \subseteq\mathcal T_{\zeta,\theta}^{\perp}
\]
give the inequalities between the distributions.
There are exactly $\Gamma_{\theta,\mathbf w}$ labels
of profile $\mathbf w$ in $\mathcal E_\theta^{(\zeta)}$, namely,
in each block, choose $w_b$ locations and one of three
nonzero labels at each chosen location.
This proves \Cref{eq:lp-capacity}.
Each subspace contains the zero vector, the unique label
of zero weight, proving \Cref{eq:lp-zero}.
Summing over profiles counts all elements of each
subspace; \Cref{lem:we-dual-row-spaces} then gives
\Cref{eq:lp-cardinalities}.

By \Cref{eq:we-fault-subspaces,eq:we-nested-distributions},
$U_{\theta,\mathbf w}-T_{\theta,\mathbf w}$ counts
the faults in
$\mathcal K_{\zeta,\theta}\setminus\mathcal T_{\zeta,\theta}$
with profile $\mathbf w$.
The signature characterization of circuit distance
in \Cref{def:ft-circuit-distance} directly gives \Cref{eq:we-distance-equalities}.
\end{proof}

\begin{proposition}[Stabilizer constraints at the boundaries]
\label{sec:lp-boundary-constraints}
\label{prop:lp-boundary-enumerators}
Let $V_{\mathrm{code}}$, $V_{\mathrm{out},\theta}$, and
$\widetilde V_{\mathrm{out}}$ be the output-stabilizer
label spaces of the encoder, unconditioned circuit, and
code-conditioned circuit, respectively
(\Cref{Cor: Logical dim Pres,prop:outcome-independent-stabilizer-relations}).
Define
\[
\resizebox{\linewidth}{!}{$\displaystyle
 \begin{aligned}
 A_j^{\mathrm{in}}
   &\triangleq
   |\{p\in V_{\mathrm{code}}:\operatorname{wt}P(p)=j\}|,&
 B_j^{\mathrm{in}}
   &\triangleq
   |\{p\in V_{\mathrm{code}}^{\perp_s}:
                         \operatorname{wt}P(p)=j\}|,\\
 A_{\theta,j}^{\mathrm{out}}
   &\triangleq
   |\{p\in V_{\mathrm{out},\theta}:
                         \operatorname{wt}P(p)=j\}|,&
 B_{\theta,j}^{\mathrm{out}}
   &\triangleq
   |\{p\in V_{\mathrm{out},\theta}^{\perp_s}:
                         \operatorname{wt}P(p)=j\}|,\\
 \widetilde A_j^{\mathrm{out}}
   &\triangleq
   |\{p\in\widetilde V_{\mathrm{out}}:
                         \operatorname{wt}P(p)=j\}|,&
 \widetilde B_j^{\mathrm{out}}
   &\triangleq
   |\{p\in\widetilde V_{\mathrm{out}}^{\perp_s}:
                         \operatorname{wt}P(p)=j\}|.
 \end{aligned}
$}
\]
for $0\leq j\leq n_{\mathrm{in}}$ () in the input pair $(A^{\mathrm{in}}_j,B^{\mathrm{in}}_j)$
and $0\leq j\leq n_{\mathrm{out}}$ in the output pairs $(A^{\mathrm{out}}_{\theta,j},B^{\mathrm{out}}_{\theta,j})$;$(\widetilde A^{\mathrm{out}}_j,\widetilde B^{\mathrm{out}}_j)$. Each pair $(A_j,B_j)$ satisfies
\begin{equation}
 \begin{gathered}
 B_j=2^{-r}\sum_{i=0}^{n}
       A_i\mathbb K_j^{(4)}(i;n)
       \qquad(0\leq j\leq n),\\
 A_0=B_0=1,\qquad
 0\leq A_j\leq B_j\leq3^j\binom nj
       \qquad(0\leq j\leq n),\\
 \sum_{j=0}^{n}A_j=2^r,\qquad
 \sum_{j=0}^{n}B_j=2^{2n-r},
 \end{gathered}
 \label{eq:lp-output-macwilliams}
\end{equation}
where $r$ is the dimension of the corresponding stabilizer
space, and $n=n_{\mathrm{in}}$ for the input pair or
$n=n_{\mathrm{out}}$ for either output pair.
If the corresponding boundary code has distance at least
$d_0$, then also $A_j=B_j$ for $0<j<d_0$.
\end{proposition}

\begin{proof}
Each boundary space is isotropic, so the standard
stabilizer weight-enumerator constraints apply
\cite{calderbank1998quantum,Rains1998QuantumWeightEnumerators}.
\end{proof}

\begin{proposition}[Boundary constraints and Singleton bound]
\label{prop:lp-input-boundary}
\label{prop:lp-final-output}
\label{cor:boundary-singleton}
Let $\alpha\in\{\mathrm{in},\mathrm{out}\}$ denote the
input or output boundary, immediately after the encoder
or the final operation, respectively.
Suppose all data-qubit locations at boundary $\alpha$
form a block $I_{\theta,\alpha}\in\Pi_\theta$, and every
Pauli pattern supported on that block belongs to
$\mathcal V_\theta^{(\zeta)}$.
Let $\iota_{\theta,\alpha}(j)$ denote the weight profile
with weight $j$ in block $I_{\theta,\alpha}$ and zero elsewhere. Using the boundary distributions of
\Cref{prop:lp-boundary-enumerators}, the input case
$\alpha=\mathrm{in}$ gives
\begin{equation}
 T_{\theta,\iota_{\theta,\mathrm{in}}(j)}
   \geq A_j^{\mathrm{in}},\qquad
 U_{\theta,\iota_{\theta,\mathrm{in}}(j)}
   \geq B_j^{\mathrm{in}}
 \quad(0\leq j\leq n_{\mathrm{in}}).
 \label{eq:lp-input-coupling}
\end{equation}
The output case $\alpha=\mathrm{out}$ gives, for
$0\leq j\leq n_{\mathrm{out}}$,
\begin{equation}
 \begin{aligned}
 T_{\theta,\iota_{\theta,\mathrm{out}}(j)}
   &=\widetilde A_j^{\mathrm{out}}
     \geq A_{\theta,j}^{\mathrm{out}},\\
 U_{\theta,\iota_{\theta,\mathrm{out}}(j)}
   &=\widetilde B_j^{\mathrm{out}}
     \leq B_{\theta,j}^{\mathrm{out}}.
 \end{aligned}
 \label{eq:lp-output-coupling}
\end{equation}
If all final-data Pauli fault patterns belong to
$\mathcal V_\theta^{(\zeta)}$ and
$k_{\mathrm{pres}}
=n_{\mathrm{out}}-\dim\widetilde V_{\mathrm{out}}>0$,
then, for any partition $\Pi_\theta$,
\begin{equation}
 n_{\mathrm{out}}-k_{\mathrm{pres}}
 \geq2\bigl(d_{\mathrm{circ},\zeta}(\theta)-1\bigr).
 \label{eq:boundary-singleton}
\end{equation}
\end{proposition}

\begin{proof}
Input labels in $V_{\mathrm{code}}$ are channel-trivial.
Those in $V_{\mathrm{code}}^{\perp_s}$ act as logical
input Paulis and preserve every ideal output check.
Counting these faults proves \Cref{eq:lp-input-coupling}.

By \Cref{Cor: Logical dim Pres,prop:outcome-independent-stabilizer-relations},
an output Pauli fault is undetected precisely when its
label lies in $\widetilde V_{\mathrm{out}}^{\perp_s}$,
and channel-trivial precisely when its label lies in
$\widetilde V_{\mathrm{out}}$.
Counting gives the output equalities.
Composing unconditioned output relations with the encoder
gives $V_{\mathrm{out},\theta}\subseteq\widetilde V_{\mathrm{out}}$;
this inclusion and its reversed symplectic-dual inclusion
give the output inequalities. Consequently,
\begin{equation}
 d_{\mathrm{circ},\zeta}(\theta)\leq d_{\mathrm{out}}
 \triangleq
 \min_{p\in\widetilde V_{\mathrm{out}}^{\perp_s}
                  \setminus\widetilde V_{\mathrm{out}}}
       \operatorname{wt}P(p),
 \label{eq:lp-output-distance}
\end{equation}
with $\min\varnothing=\infty$, whenever all final-data
faults are permitted.
For $k_{\mathrm{pres}}>0$, the quantum Singleton bound
\cite[Sec.~7.1]{gottesman1997stabilizer} gives
\[
 n_{\mathrm{out}}-k_{\mathrm{pres}}
 \geq2(d_{\mathrm{out}}-1)
 \geq2\bigl(d_{\mathrm{circ},\zeta}(\theta)-1\bigr).
\]
\end{proof}

\begin{remark}
    For a given circuit family, the constraints in \Cref{cor:boundary-singleton} are applicable only under the stated assumptions therein.
\end{remark}
\subsection{Linear programming bounds}
\label{sec:linear-programming-bounds}
We now use the constraints derived in the earlier section to obtain LP upper bounds on circuit
distance and pseudo-threshold, and lower bounds on undetected logical
failure probability.

\subsubsection{Distance bounds}
\label{sec:lp-distance}
We first group circuits whose enumerator constraints have the same
numerical data. We formalize this below.

\begin{definition}[Circuit classes]
\label{def:lp-circuit-classes}
Partition the parameter set into finitely many nonempty classes,
\begin{equation}
 \Theta=\bigsqcup_{c\in\mathfrak C}\Theta_c,
 \label{eq:lp-class-partition}
\end{equation}
with common block sizes and cardinalities within each class, namely,
\begin{equation}
 \begin{gathered}
 (N_{\theta,1},\ldots,N_{\theta,B_\theta})
   =(N_{c,1},\ldots,N_{c,B_c}),\\
 |\mathcal K_{\zeta,\theta}|=\kappa_c,\qquad
 |\mathcal T_{\zeta,\theta}|=\tau_c
 \qquad(\theta\in\Theta_c).
 \end{gathered}
 \label{eq:lp-class-data}
\end{equation}
The applicability, block indices, and numerical data of the
boundary constraints used are also common within each class.
Such a finite partition exists because
$N_\theta\leq N_{\max}^{\Theta}$ and the boundary registers
are fixed. Write $N_c=\sum_{b=1}^{B_c}N_{c,b}$ and denote the corresponding
profile set, capacities, and Krawtchouk coefficients by
$J_c$, $\Gamma_{c,\mathbf w}$, and
$\mathbb K_{\mathbf j}^{c}(\mathbf w)$ respectively.
\end{definition}

\begin{definition}[Spacetime linear program]
\label{def:lp-spacetime}
For an integer $D\geq1$, we define the \emph{spacetime distance linear program} as the feasibility LP
$\mathrm{STLP}_\Theta(D)$ which has four real arrays
\[
 U^{(c)}_{\mathbf w},\quad T^{(c)}_{\mathbf w},\quad
 \widehat U^{(c)}_{\mathbf w},\quad
 \widehat T^{(c)}_{\mathbf w}
 \qquad(c\in\mathfrak C,\ \mathbf w\in J_c),
\]
and weights
\begin{equation}
 \eta_c\geq0,\qquad
 \sum_{c\in\mathfrak C}\eta_c=1.
 \label{eq:lp-class-weights}
\end{equation}
For each class $c$ in \Cref{def:lp-circuit-classes}, impose the MacWilliams identities
\Cref{eq:we-split-macwilliams} using its class data, and
the distance constraints
\[
 U^{(c)}_{\mathbf w}=T^{(c)}_{\mathbf w}
 \qquad(\mathbf w\in J_c,\ 0<|\mathbf w|<D).
\]
These equalities retain their form. The counting constraints
\Cref{eq:lp-capacity,eq:lp-zero,eq:lp-cardinalities}
become
\begin{equation}
 \begin{gathered}
 U^{(c)}_{\mathbf0}=T^{(c)}_{\mathbf0}
 =\widehat U^{(c)}_{\mathbf0}
 =\widehat T^{(c)}_{\mathbf0}=\eta_c,\\
 0\leq T^{(c)}_{\mathbf w}\leq U^{(c)}_{\mathbf w}
       \leq\eta_c\Gamma_{c,\mathbf w},\\
 0\leq\widehat U^{(c)}_{\mathbf w}
       \leq\widehat T^{(c)}_{\mathbf w}
       \leq\eta_c\Gamma_{c,\mathbf w}
       \qquad(\mathbf w\in J_c),\\
 \begin{aligned}
 \sum_{\mathbf w\in J_c}U^{(c)}_{\mathbf w}
   &=\eta_c\kappa_c,&
 \sum_{\mathbf w\in J_c}T^{(c)}_{\mathbf w}
   &=\eta_c\tau_c,\\
 \sum_{\mathbf w\in J_c}\widehat U^{(c)}_{\mathbf w}
   &=\eta_c\frac{4^{N_c}}{\kappa_c},&
 \sum_{\mathbf w\in J_c}\widehat T^{(c)}_{\mathbf w}
   &=\eta_c\frac{4^{N_c}}{\tau_c}.
 \end{aligned}
 \end{gathered}
 \label{eq:lp-scaled-constraints}
\end{equation}
Include the boundary constraints
\Cref{eq:lp-input-coupling,eq:lp-output-coupling}
when applicable, multiplying every boundary count on their
right-hand sides by $\eta_c$. The boundary data satisfy
\Cref{eq:lp-output-macwilliams} and any prescribed boundary-code
distance. When \Cref{eq:boundary-singleton} applies, also impose
\[
 \eta_c(n_{\mathrm{out}}-k_{\mathrm{pres}})
 \geq2\eta_c(D-1).
\]
\end{definition}

\begin{theorem}[LP distance bound]
\label{thm:lp-obstruction}
Every circuit with $d_{\mathrm{circ},\zeta}(\theta)\geq D$
gives a feasible point of $\mathrm{STLP}_\Theta(D)$.
Consequently,
\begin{equation}
 \begin{gathered}
 \mathrm{STLP}_\Theta(D)\text{ infeasible}\\
 \Longrightarrow\quad
 d_{\mathrm{circ},\zeta}(\theta)\leq D-1
 \quad\text{for every }\theta\in\Theta.
 \end{gathered}
 \label{eq:lp-obstruction}
\end{equation}
Define
\begin{equation}
 \begin{aligned}
 d_{\mathrm{LP}}\triangleq\min\bigl\{D-1:\;&
 D\in\{1,\ldots,N_{\max}^{\Theta}+1\},\\
 &\mathrm{STLP}_\Theta(D)\text{ infeasible}\bigr\},
 \end{aligned}
 \label{eq:lp-distance-bound}
\end{equation}
with $\min\varnothing=\infty$. Then
$d_{\mathrm{circ},\zeta}(\theta)\leq d_{\mathrm{LP}}$
for every family member. The bound is finite whenever
$\kappa_c>\tau_c$ in every class.
\end{theorem}

\begin{proof}
For a circuit $\theta\in\Theta_c$ of distance at least $D$,
set $\eta_c=1$ and assign its four enumerators to the
class-$c$ arrays. Set all other weights and arrays to zero.
By
\Cref{prop:we-signature-macwilliams,prop:we-enumerator-constraints,prop:lp-input-boundary},
these assignments satisfy every applicable constraint.
Thus infeasibility excludes every circuit of distance
at least $D$, proving the bounds.

At $D=N_{\max}^{\Theta}+1$, the distance constraints
give $U^{(c)}_{\mathbf w}=T^{(c)}_{\mathbf w}$ for every
profile. Summing and using the cardinality constraints yields
$\eta_c(\kappa_c-\tau_c)=0$ for every class.
If $\kappa_c>\tau_c$ throughout for every $c$, all weights vanish,
contradicting $\sum_{c\in\mathfrak C}\eta_c=1$.
\end{proof}

\subsubsection{Performance bounds}
\label{sec:iid-performance}
Circuit distance identifies the smallest harmful undetected fault.
Performance under a specified noise model also depends on the number
of fault patterns of each weight and signature. We use the enumerators
to bound undetected logical failure probability and pseudo-threshold
throughout the family. In this subsection, all Pauli fault patterns on
$\Lambda_\theta^{(\zeta)}$ are permitted, so
$\mathcal V_\theta^{(\zeta)}=\mathcal E_\theta^{(\zeta)}$
for every $\theta\in\Theta$. Additionally, in what follows, we present our results for i.i.d. depolarizing noise and leave the generalization to stochastic noise model for future work.

\begin{definition}[i.i.d. depolarizing noise]
\label{def:iid-noise}
For $0\leq p\leq1$, the \emph{i.i.d. depolarizing noise model}
$\mathbf E_\theta(p)$ is the special case of
\Cref{def:stochastic-pauli-noise} with probabilities
\begin{equation}
 \mu_{\theta,e}(p)
 =\pi_{\theta,\mathbf w}(p)
 \triangleq
 (1-p)^{N_\theta-|\mathbf w|}(p/3)^{|\mathbf w|},
 \label{eq:iid-uniform-probability}
\end{equation}
where $e\in\mathcal E_\theta^{(\zeta)}$ and
$\mathbf w=\operatorname{wt}_{\Pi_\theta}(e)$.
Thus the Pauli labels at the locations in
$\Lambda_\theta^{(\zeta)}$ are independent, each assigning
probability $1-p$ to $I$ and $p/3$ to each of $X,Y,Z$.
On a classical record, this gives bit-flip probability $2p/3$.
\end{definition}

\begin{definition}[Attainable signatures and deterministic decoder]
\label{def:iid-signatures}
\label{def:iid-decoder}
The attainable signature set is
\begin{equation}
 \Sigma_{\zeta,\theta}
 \triangleq
 \left\{
 \bigl(H_{\zeta,\theta}e^{\mathsf T},
       L_{\zeta,\theta}e^{\mathsf T}\bigr):
 e\in\mathcal E_\theta^{(\zeta)}
 \right\}.
 \label{eq:iid-attainable-signatures}
\end{equation}
A fixed deterministic syndrome decoder assigns
$\widehat{e}_{\zeta,\theta}(\sigma)\in\mathcal E_\theta^{(\zeta)}$
to every $\sigma$ for which
$(\sigma,\lambda)\in\Sigma_{\zeta,\theta}$ for some $\lambda$, with
\begin{equation}
 \begin{gathered}
 H_{\zeta,\theta}
 \widehat{e}_{\zeta,\theta}(\sigma)^{\mathsf T}
 =\sigma,\qquad
 \widehat{e}_{\zeta,\theta}(0)=0,\\
 \delta_\theta(\sigma)
 \triangleq
 L_{\zeta,\theta}
 \widehat{e}_{\zeta,\theta}(\sigma)^{\mathsf T}.
 \end{gathered}
 \label{eq:iid-decoder}
\end{equation}
It applies correction
$P(\beta_\theta^{(\zeta)}(\widehat{e}_{\zeta,\theta}(\sigma)))^\dagger$ at the circuit data-record output
as in \Cref{thm:ft-HL-criterion}, with noiseless
output-check measurements and recovery.
\end{definition}

\begin{definition}[Undetected logical failure probability]
\label{def:iid-performance-metrics}
\label{def:iid-undetected}
Under \Cref{def:iid-noise}, the unconditional probability
of a harmful undetected fault is
\begin{equation}
 P_{\mathrm{ud},\theta}^{(\zeta)}(p)
 \triangleq
 \left.
 \sum_{\substack{(0,\lambda)\in\Sigma_{\zeta,\theta}\\
                  \lambda\ne0}}
 W_{\theta;0,\lambda}^{(\zeta),\Pi_\theta}(u,v)
 \right|_{\substack{u_b=1-p\\v_b=p/3}}.
 \label{eq:iid-undetected-bound}
\end{equation}
Here and below, the substitutions apply to every
$b\in\mathbb Z^+_{B_\theta}$.
\end{definition}

\begin{definition}[Decoded failure probability]
\label{def:iid-decoded-failure}
Under \Cref{def:iid-noise}, the probability of a fault pattern for which, after the decoder in \Cref{def:iid-decoder} applies a correction, the resulting channel is different from $\mathcal{N}_0$ is
\begin{equation}
 P_{\mathrm{fail},\theta}^{(\zeta),\delta_\theta}(p)
 \triangleq
 \left.
 \sum_{\substack{(\sigma,\lambda)\in\Sigma_{\zeta,\theta}\\
                  \lambda\ne\delta_\theta(\sigma)}}
 W_{\theta;\sigma,\lambda}^{(\zeta),\Pi_\theta}
       (u,v)
 \right|_{\substack{u_b=1-p\\v_b=p/3}}.
 \label{eq:iid-failure-probability}
\end{equation}
\end{definition}

\begin{definition}[Pseudo-threshold]
\label{def:iid-pseudothreshold}
Fix $0<p_{\max}\leq1$. For an integer $j\geq0$,
define the pseudo-threshold relative to the target
$10^{-j}p$ by
\begin{equation}
 \begin{aligned}
 p_{\mathrm{ps},\theta}^{(\zeta),\delta_\theta}(j)
 \triangleq\sup\bigl\{x\in[0,p_{\max}]:{}&\\[-3pt]
 P_{\mathrm{fail},\theta}^{(\zeta),\delta_\theta}(p)
 \leq10^{-j}p
 \quad\text{for every }0<p<x&\bigr\}.
 \end{aligned}
 \label{eq:iid-pseudothreshold}
\end{equation}
The case $j=0$ gives the usual pseudo-threshold
\cite[Sec.~1.2]{ChamberlandBeverland2018Flags};
$j\geq1$ requires suppression by a factor of at least
$10^j$ relative to the physical error rate.
The value is zero if no positive interval qualifies.
\end{definition}

\begin{definition}[Circuit performance linear program]
\label{def:iid-performance-lp}
Fix integers $D\geq1$, $j\geq0$, and
$p\in[0,p_{\max}]$. Put
$\pi_{c,\mathbf w}(p)
=(1-p)^{N_c-|\mathbf w|}(p/3)^{|\mathbf w|}$ and define
\begin{equation}
 \underline P_\Theta(D,p)
 \triangleq
 \min_{\mathrm{STLP}_\Theta(D)}
 \sum_{c\in\mathfrak C}\sum_{\mathbf w\in J_c}
 (U^{(c)}_{\mathbf w}-T^{(c)}_{\mathbf w})
 \pi_{c,\mathbf w}(p),
 \label{eq:iid-lp-objective}
\end{equation}
with an empty minimum equal to $+\infty$.
The feasibility LP $\mathrm{PerfLP}_\Theta(D,p,j)$ consists
of $\mathrm{STLP}_\Theta(D)$ and the single additional constraint
\begin{equation}
 \sum_{c\in\mathfrak C}\sum_{\mathbf w\in J_c}
 (U^{(c)}_{\mathbf w}-T^{(c)}_{\mathbf w})
 \pi_{c,\mathbf w}(p)
 \leq10^{-j}p.
 \label{eq:iid-performance-lp}
\end{equation}
\end{definition}

\begin{theorem}[LP performance bounds]
\label{prop:iid-failure}
\label{prop:iid-performance-lp}
Fix integers $D\geq1$ and $j\geq0$.
Every $\theta\in\Theta$ with
$d_{\mathrm{circ},\zeta}(\theta)\geq D$ and every decoder
in \Cref{def:iid-decoder} satisfy, for $0\leq p\leq p_{\max}$,
\begin{equation}
 \underline P_\Theta(D,p)
 \leq P_{\mathrm{ud},\theta}^{(\zeta)}(p)
 \leq P_{\mathrm{fail},\theta}^{(\zeta),\delta_\theta}(p).
 \label{eq:iid-lp-lower-bound}
\end{equation}
Whenever
$P_{\mathrm{fail},\theta}^{(\zeta),\delta_\theta}(p)\leq10^{-j}p$,
the circuit family's enumerators give a feasible point of
$\mathrm{PerfLP}_\Theta(D,p,j)$.
Consequently,
\begin{equation}
 \begin{gathered}
 p_{\mathrm{ps},\theta}^{(\zeta),\delta_\theta}(j)
 \leq p_{\mathrm{ps}}^{\mathrm{LP}}(D,j),\\
 \begin{aligned}
 p_{\mathrm{ps}}^{\mathrm{LP}}(D,j)
 \triangleq\sup\bigl\{x\in[0,p_{\max}]:{}&\\[-3pt]
 \mathrm{PerfLP}_\Theta(D,p,j)\text{ is feasible}
 \quad\text{for every }0<p<x&\bigr\}.
 \end{aligned}
 \end{gathered}
 \label{eq:iid-lp-pseudothreshold}
\end{equation}
\end{theorem}

\begin{proof}
The feasible point supplied by \Cref{thm:lp-obstruction}
has objective $P_{\mathrm{ud},\theta}^{(\zeta)}(p)$,
by expanding \Cref{eq:we-signature-polynomial}.
Minimization gives the first inequality in
\Cref{eq:iid-lp-lower-bound}.
Since $\delta_\theta(0)=0$, every harmful undetected fault
also causes decoded failure, giving the second.
Thus this point satisfies the additional performance constraint
whenever the decoder achieves the prescribed target.

Every comparison interval in \Cref{def:iid-pseudothreshold}
therefore gives an interval on which the performance LP
is feasible. Taking supremums proves
\Cref{eq:iid-lp-pseudothreshold}.
\end{proof}
We illustrate these LP bounds with flag-based syndrome-extraction
circuits, expressing the $t$-flag condition as additional linear
constraints in \Cref{subsec:flag-syndrome-extraction}.


\section{Discussion and outlook}
\label{sec:discussion-outlook}

The tools developed in this work provide a general framework for analyzing
fault tolerance while retaining the full circuit operation, including
quantum outputs and measurement records. We therefore foresee broad
applications in fault-tolerant circuit verification and design. Some interesting avenues for future work are as follows:
\begin{itemize}
   \item \emph{Adaptive stabilizer circuits.} Consider the simple adaptive circuit in
\Cref{remark: Adaptive case}. It is shown there that the global parity check
$o_2+o_3=0$ does not capture the outcome-conditional parity check $(1+o_1)o_2=0$ over $\mathbb F_2$. Applying our construction of the common relation space
$\mathcal W_{\mathbf C}$ in
\Cref{def:common-stabilizer-relations} to adaptive circuits captures only Pauli relations valid across all branches
with fixed input and output Pauli labels and signs
determined by fixed record parities. It therefore need
not capture such conditional checks. An extension of our framework could retain the branch spaces $V_{K_o}$
defined in \Cref{Cor: Logical dim Pres}, together with
their relation signs and the conditions selecting each history.
The composition rule of \Cref{prop: stabilizers of composition}
could then update these relations as the preceding outcomes
determine the next operation. An extension that includes faults on the classical controller must also account for corrupted control outcomes that change these choices, beyond the retained-record noise convention in \Cref{rem:noise_model_explicit}.
Existing verification methods for circuits with classical
control
\cite{kliuchnikov2023stabilizercircuitverification,
FangYing2024Symbolic,ChenEtAl2025FaultTolerance}
and composable detector error models
\cite{kliuchnikov2026composing} provide starting points. A remaining challenge is to combine these approaches with
our fault-signature and weight-enumerator tools while retaining
the conditional checks without explicitly enumerating all
measurement histories, whose number can grow exponentially
with the number of measurements.
    \item \emph{Non-stabilizer circuits.}
Our recovery conditions in
\Cref{Thm: Iff Fault tolerance} already apply to non-stabilizer
circuits, but our explicit fault signatures and weight
enumerators rely on stabilizer structure. For example,
$TXT^\dagger=(X+Y)/\sqrt{2}$, so a single Pauli fault can become
a sum of Pauli terms that affect subsequent checks differently.
Replacing this coherent sum by a random choice between
an $X$ fault and a $Y$ fault, each with probability $1/2$,
would generally change the error channel. Extending our weight-enumerator LP approach would require
quantities that describe faults of each spacetime weight
while retaining the amplitudes and relative phases of the
Pauli terms they produce. The next step would be to derive
linear equalities and inequalities that these quantities
must satisfy whenever a circuit achieves a specified level
of fault protection. Circuits with few $T$ gates offer a concrete starting point, namely, expanding each $T$ as a linear combination of $I$ and $Z$
expresses the circuit operator as a sum of Clifford circuit
operators, whose contributions must be combined coherently.
Existing tensor and circuit enumerators
\cite{CaoLackey2024TensorEnumerators,
KuklianskyLackey2024CircuitEnumerators}
and simulation methods using stabilizer decompositions
\cite{BravyiGosset2016Simulation} provide tools for investigating
this extension.
\item \emph{Connectivity and other hardware constraints.} Our framework can accommodate restricted qubit connectivity, native gate sets, and limits on circuit depth or ancilla resources by restricting the admissible circuit family. This provides a basis for investigating how hardware constraints limit achievable fault tolerance. We leave the development of corresponding SDP and LP bounds and explicit circuit constructions for future work.
    \item \emph{Noise models.} In \Cref{sec:iid-performance} we have derived performance bounds for i.i.d. depolarizing noise model. The next step would be to generalize this to stochastic Pauli noise models.
    \item \emph{Tractable solutions for SDP/LP and constructions.} Another interesting future direction would be to explore both stabilizer and non-stabilizer circuit families for which our LP and SDP would have analytically or numerically tractable solutions, providing a direct advantage over brute force search.
    \item \emph{Closed forms.} Computing neat closed form expressions for level $1$ non-adaptive circuits with varying qubit registers over time and multiple measurement ancilla can lead to more general analysis of floquet codes.
    \item \emph{Decoders.} In our performance bounds we have only explored a basic look-up table decoder. The next step would be to obtain bounds and constructions for efficient and scalable decoders.
\end{itemize}

\begin{acknowledgments}
S. Ghosh and H. Boche acknowledge support from
the BMFTR under the ``Souver\"an. Digital.
Vernetzt.'' program for project 6G-life, 16KIS2414. S. Ghosh and H. Boche acknowledge support from the BMFTR
Quantum Programme Q-STARS, 16KIS2601, and project QUARKS, 
16KIS1998K.

H. Boche acknowledges funding by
the DFG as part of Germany’s Excellence
Strategy-EXC 2050/2-Project ID 390696704-
Cluster of Excellence ``Centre for Tactile Internet
with Human-in-the-Loop'' (CeTI) of Technische
Universit\"at Dresden and within the Gottfried
Wilhelm Leibniz Prize Programme under Grant BO 1734/20-1 and in part
by DFG Projects under Grant BO 1734/35-1.

A.T. is supported by the CQT PhD Scholarship, the Google PhD Fellowship Program and the CQT Young Researcher Career Development Grant.
\end{acknowledgments}

\paragraph*{Data availability.}
The data specifying the finite honeycomb circuit and the analytic
CSS Hamming-code benchmark are contained in this article.

\paragraph*{AI use.} 
The authors used OpenAI GPT-5.6 Sol to assist with writing the Python code for the honeycomb Floquet example in Appendix B. All the proofs were done by the authors while their correctness were cross-checked with OpenAI GPT-5.6 Sol and GPT-6 Astra.
\onecolumngrid
\appendix
\crefalias{section}{appendix}



\section{Causal normalization of circuit map and noise model}
\label{app:causal-normalisation}

The constraints used in \Cref{def: Circuit map and noise model}
follow from local trace preservation and causality. We derive
them by discarding the latest output and repeating the
calculation at each earlier step. We then show that connecting
the circuit and noise gives a trace-preserving map.
The register dimensions may depend on the time step, as in
\Cref{def:quantum_circuit}.
All vectorizations use the fixed bases of the preliminaries
and are unnormalized. A CP map
$C(\rho)=\sum_a C_a\rho C_a^\dagger$ has Choi operator
$J_C=\sum_a\vket{C_a}\vbra{C_a}\succeq0$
\cite{chiribella2009theoretical}.
For an operator $D:X\to Y$, write
$\vket D=\sum_{i,j}D_{ij}|i\rangle_Y\otimes|j\rangle_X$.
Tracing the output gives
\begin{equation}
    \begin{aligned}
        \operatorname{Tr}_Y(\vket D\vbra D)
        &=\sum_{j,k}\left(\sum_iD_{ij}D_{ik}^{\star}\right)
                         |j\rangle\langle k|\\
        &=(D^\dagger D)^{\mathsf T}.
    \end{aligned}
    \label{eq:app-comb-local-trace}
\end{equation}
Consequently, for Kraus operators $\{D_a\}_a$, the
condition $\sum_aD_a^\dagger D_a=I_X$ is equivalent to
$\operatorname{Tr}_Y\sum_a\vket{D_a}\vbra{D_a}=I_X$.
We use this identity and its index form below.

\subsection{Circuit map constraints}

To distinguish open wires, let $Q_r^{\mathrm{out}}$ denote
the data output of component $r$, and let
$Q_r^{\mathrm{in}}$ denote the input of component $r+1$.
For $1\le r<\ell$, these are distinct copies of $Q_r$
until connected. The final data output is denoted by
$Q_\ell^{\mathrm{out}}$. Likewise, the external input
$Q_0^{\mathrm{ext}}$ is distinct from the first component's
input copy $Q_0^{\mathrm{in}}$. For each fixed history $h=o_{1:r-1}$, the component
Kraus operators define a trace-preserving channel
\begin{equation}
    C_{o_r|h}:Q_{r-1}^{\mathrm{in}}\longrightarrow
             A_r^{\mathrm{out}}\otimes Q_r^{\mathrm{out}},
    \qquad
    \sum_{o_r}C_{o_r|h}^\dagger C_{o_r|h}
        =I_{Q_{r-1}^{\mathrm{in}}}.
    \label{eq:app-comb-component-tp}
\end{equation}
It is this local completeness relation that yields
the circuit-map constraints. For $0\le r\le\ell$, define the prefix space containing
all open input and output legs of the first $r$ circuit
components by
\[
    \mathcal H_r
    \triangleq
    \bigotimes_{i\in\mathbb Z^+_{r}}
    \left(
        A_i^{\mathrm{out}}\otimes Q_i^{\mathrm{out}}
        \otimes Q_{i-1}^{\mathrm{in}}
    \right).
\]
The factors are ordered from $i=r$ to $i=1$, and the
empty tensor product gives $\mathcal H_0=\mathbb C$.
For each outcome prefix $o_{1:r}$, define
\begin{equation}
    \begin{aligned}
        \mathbf C_{o_{1:r}}
        &\triangleq
        \vket{C_{o_r|o_{1:r-1}}}\otimes\cdots
        \otimes\vket{C_{o_1}}
        \otimes I_{Q_0^{\mathrm{ext}}},\\
        \mathbf C_{\varnothing}
        &\triangleq I_{Q_0^{\mathrm{ext}}}.
    \end{aligned}
    \label{eq:app-comb-prefixes}
\end{equation}
Thus $\mathbf C_{o_{1:r}}$ maps $Q_0^{\mathrm{ext}}$
to $\mathcal H_r\otimes Q_0^{\mathrm{ext}}$.
At $r=\ell$, these are precisely the circuit-map Kraus
operators $\mathbf C_o$. The normalization operators
in \Cref{def: Circuit map and noise model} are
\begin{equation}
    \mathbf C^{(r)}
    \triangleq
    \sum_{o_{1:r}}
        \mathbf C_{o_{1:r}}\mathbf C_{o_{1:r}}^\dagger
    \succeq0,
    \qquad
    \mathbf C^{(0)}=I_{Q_0^{\mathrm{ext}}}.
    \label{eq:app-comb-normalisation-operators}
\end{equation}
For $h=o_{1:r-1}$, the prefix factorizes as
$\mathbf C_{o_{1:r}}=\vket{C_{o_r|h}}\otimes\mathbf C_h$.
Grouping the sum by $h$ and applying
\eqref{eq:app-comb-local-trace} to the last component gives
\begin{equation}
    \begin{aligned}
        \operatorname{Tr}_{A_r^{\mathrm{out}}
                              Q_r^{\mathrm{out}}}
             \mathbf C^{(r)}
        &=
        \sum_h
        \left(\sum_{o_r}C_{o_r|h}^\dagger C_{o_r|h}
        \right)^{\mathsf T}
        \otimes\mathbf C_h\mathbf C_h^\dagger\\
        &=
        I_{Q_{r-1}^{\mathrm{in}}}
        \otimes\mathbf C^{(r-1)}.
    \end{aligned}
    \label{eq:app-comb-circuit-normalisation}
\end{equation}
For $r=1$, the history is empty and
$\mathbf C_h=I_{Q_0^{\mathrm{ext}}}$, giving the stated
boundary condition. Since the completeness relation
holds for every history, the argument also applies
to all adaptive choices of the components. These are the deterministic quantum-comb normalization constraints
with the external input retained
\cite{chiribella2009theoretical}.
Discarding the latest output leaves an identity on
that component's input, tensored with the earlier
prefix normalization operator. Thus later trace-preserving
operations do not alter the earlier process.
The full circuit map satisfies
\begin{equation}
    \mathbf C_{\mathrm{cm}}(\rho)
       =\sum_o\mathbf C_o\rho\mathbf C_o^\dagger,
    \qquad
    \mathbf C_{\mathrm{cm}}(I_{Q_0^{\mathrm{ext}}})
       =\mathbf C^{(\ell)}.
    \label{eq:app-comb-circuit-map}
\end{equation}
Here $\mathbf C^{(\ell)}$ is the normalization operator
appearing in \Cref{def: Circuit map and noise model} and it is not the Choi operator
of $\mathbf C_{\mathrm{cm}}$. Local trace preservation does not imply that this open
circuit map is trace-preserving. Indeed, the external
input is an identity tensor factor in every $\mathbf C_o$,
so $\sum_o\mathbf C_o^\dagger\mathbf C_o$ is proportional
to $I_{Q_0^{\mathrm{ext}}}$. Taking traces in
\eqref{eq:app-comb-circuit-normalisation} determines
the proportionality factor
\begin{equation}
    \begin{aligned}
        \sum_o\mathbf C_o^\dagger\mathbf C_o
        &=
        \left(\prod_{r=0}^{\ell-1}\dim Q_r\right)
        I_{Q_0^{\mathrm{ext}}},\\
        \operatorname{Tr}\mathbf C_{\mathrm{cm}}(\rho)
        &=
        \left(\prod_{r=0}^{\ell-1}\dim Q_r\right)
        \operatorname{Tr}\rho.
    \end{aligned}
    \label{eq:app-comb-open-trace}
\end{equation}
The dimension factors arise from the unnormalized
vectorizations. After contraction with a noise model
satisfying the stated constraints, the resulting map is
trace-preserving, as shown in
\Cref{subsec: Trace-pres after contraction}.

\subsection{Noise model constraints}

Consider the local noise operators of
\Cref{rem:noise_model_explicit}, with
$N_0=N_{\ell+1}=\mathbb C$.
The initial operation includes preparation of the memory,
and the final Kraus family includes tracing it out.
More explicitly, the local operators have domains and
codomains
\[
    \begin{aligned}
        E_{e_0}&:Q_0^{\mathrm{ext}}
             \longrightarrow Q_0^{\mathrm{in}}\otimes N_1,\\
        E_{e_r}&:A_r^{\mathrm{out}}\otimes Q_r^{\mathrm{out}}
                    \otimes N_r
             \longrightarrow A_r^{\mathrm{out}}
                    \otimes Q_r^{\mathrm{in}}\otimes N_{r+1}
                    \quad(1\leq r<\ell),\\
        E_{e_\ell}&:A_\ell^{\mathrm{out}}\otimes Q_\ell^{\mathrm{out}}
                    \otimes N_\ell
             \longrightarrow A_\ell^{\mathrm{out}}
                    \otimes Q_\ell.
    \end{aligned}
\]
The input and output copies of each record register are
distinct tensor factors, although both are denoted by
$A_r^{\mathrm{out}}$. Noise may act on these registers
as well as on the data.
Their local completeness relations are
\begin{equation}
    \sum_{e_0}E_{e_0}^\dagger E_{e_0}
       =I_{Q_0^{\mathrm{ext}}},
    \qquad
    \sum_{e_r}E_{e_r}^\dagger E_{e_r}
       =I_{A_r^{\mathrm{out}}Q_r^{\mathrm{out}}N_r}
       \quad(1\le r\le\ell).
    \label{eq:app-comb-noise-local-tp}
\end{equation}

For $0\le r<\ell$, define the noise-prefix operators
\begin{equation}
    \begin{aligned}
        \mathbf E_{e_{0:r}}:
        Q_r^{\mathrm{in}}\otimes
        \mathcal H_r\otimes Q_0^{\mathrm{ext}}
        \longrightarrow
        A^{\mathrm{out}}_{1:r}
        \otimes N_{r+1}.
    \end{aligned}
    \label{eq:app-comb-noise-prefix}
\end{equation}
The component blocks in the domain are ordered from
$i=r$ to $i=1$; empty tensor products equal $\mathbb C$.
Writing $I_{\mathrm{in}}$ (resp. $I_{\mathrm{out}}$) for the
identities on the displayed domain (resp. codomain) of \Cref{eq:app-comb-noise-prefix}, we have
\begin{equation}
    \begin{aligned}
        \mathbf E_{e_{0:r}}
        \triangleq
        \left(
            I_{\mathrm{out}}
            \otimes\vbra{E_{e_r}^{\star}}
            \otimes\cdots\otimes\vbra{E_{e_0}^{\star}}
        \right)\cdot
        \left(
            \vket{I_{\mathrm{out}}}
            \otimes\bigotimes_{s=1}^{r}\vket{I_{N_s}}
            \otimes I_{\mathrm{in}}
        \right).
    \end{aligned}
    \label{eq:app-comb-noise-prefix-construction}
\end{equation}
Each $\vket{I_{N_s}}$ connects the $N_s$ output of
$E_{e_{s-1}}$ to the $N_s$ input of $E_{e_s}$.
The ket $\vket{I_{\mathrm{out}}}$ retains the record
outputs and the remaining memory, while
$I_{\mathrm{in}}$ leaves the prefix inputs open.
Tensor factors are paired according to these register
labels and reordered as in
\eqref{eq:app-comb-noise-prefix}. Define the positive normalization operators
\begin{equation}
    \mathbf E^{(r)}
       \triangleq
       \sum_{e_{0:r}}
          \mathbf E_{e_{0:r}}^\dagger\mathbf E_{e_{0:r}}
       \succeq0.
    \label{eq:app-comb-noise-effects}
\end{equation}
For $1\le r<\ell$, write $h' =e_{0:r-1}$.
For fixed $h'$ and $e_r$, the prefix operator is obtained
by treating the $Q_r^{\mathrm{in}}$ output of $E_{e_r}$
as an open input index. Tracing this index in
$\mathbf E_{e_{0:r}}^\dagger\mathbf E_{e_{0:r}}$
therefore sums over that output index. The other output
indices are already summed in the operator product.
Together these sums give $E_{e_r}^\dagger E_{e_r}$,
with no transpose. Grouping the sum by $h'$ and applying
local completeness gives
\begin{equation}
    \begin{aligned}
        \operatorname{Tr}_{Q_r^{\mathrm{in}}}\mathbf E^{(r)}
        &=
        \sum_{h'}
        \left(
            I_{A_r^{\mathrm{out}}Q_r^{\mathrm{out}}}
            \otimes\mathbf E_{h'}
        \right)^\dagger
        \left[
            I_{A_{1:r-1}^{\mathrm{out}}}
            \otimes\sum_{e_r}E_{e_r}^\dagger E_{e_r}
        \right]
        \left(
            I_{A_r^{\mathrm{out}}Q_r^{\mathrm{out}}}
            \otimes\mathbf E_{h'}
        \right)\\
        &=
        I_{A_r^{\mathrm{out}}Q_r^{\mathrm{out}}}
        \otimes\sum_{h'}\mathbf E_{h'}^\dagger\mathbf E_{h'}\\
        &=
        I_{A_r^{\mathrm{out}}Q_r^{\mathrm{out}}}
        \otimes\mathbf E^{(r-1)}.
    \end{aligned}
    \label{eq:app-comb-noise-recursion-proof}
\end{equation}
The identity on $A_{1:r-1}^{\mathrm{out}}$ carries the
earlier record outputs through the latest operation.
At $r=0$, the same calculation gives the boundary condition
\begin{equation}
    \operatorname{Tr}_{Q_0^{\mathrm{in}}}\mathbf E^{(0)}
       =\sum_{e_0}E_{e_0}^\dagger E_{e_0}
       =I_{Q_0^{\mathrm{ext}}}.
    \label{eq:app-comb-noise-boundary}
\end{equation}
Appending the final local noise operator yields
\begin{equation}
    \mathbf E_e
       =
       \left(
           I_{A_{1:\ell-1}^{\mathrm{out}}}\otimes E_{e_\ell}
       \right)
       \left(
           I_{A_\ell^{\mathrm{out}}Q_\ell^{\mathrm{out}}}
           \otimes\mathbf E_{e_{0:\ell-1}}
       \right).
    \label{eq:app-comb-final-noise}
\end{equation}
Summing first over $e_\ell$ and applying local
completeness once more gives
\begin{equation}
    \begin{aligned}
        \sum_e\mathbf E_e^\dagger\mathbf E_e
        &=
        I_{A_\ell^{\mathrm{out}}Q_\ell^{\mathrm{out}}}
        \otimes
        \sum_{e_{0:\ell-1}}
           \mathbf E_{e_{0:\ell-1}}^\dagger
           \mathbf E_{e_{0:\ell-1}}\\
        &=
        I_{A_\ell^{\mathrm{out}}Q_\ell^{\mathrm{out}}}
        \otimes\mathbf E^{(\ell-1)}.
    \end{aligned}
    \label{eq:app-comb-noise-normalisation}
\end{equation}
These are the noise constraints of
\Cref{def: Circuit map and noise model}.
The connected memory registers allow correlations
between 
local noise operators; no independence assumption
is used. A related operational treatment of quantum
processes with memory is given by the process-tensor
framework \cite{pollock2018non}.

\subsection{Normalization after connecting circuit and noise}
\label{subsec: Trace-pres after contraction}
The following argument uses only the circuit and noise
normalization constraints. It therefore applies to any
noise map satisfying those constraints, whether or not
a local realization has been specified.
The connected map has Kraus operators
$\{\mathbf E_e\mathbf C_o\}_{e,o}$.
By \eqref{eq:app-comb-noise-normalisation},
\begin{equation}
    \sum_{e,o}
       (\mathbf E_e\mathbf C_o)^\dagger
        (\mathbf E_e\mathbf C_o)
    =
    \sum_o\mathbf C_o^\dagger
       \bigl(I_{A_\ell^{\mathrm{out}}Q_\ell^{\mathrm{out}}}
                   \otimes\mathbf E^{(\ell-1)}\bigr)
        \mathbf C_o.
    \label{eq:app-comb-connected-kraus}
\end{equation}
Fix $h=o_{1:r-1}$. Choose orthonormal bases
$\{|a\rangle\}$ of
$A_r^{\mathrm{out}}\otimes Q_r^{\mathrm{out}}$
and $\{|i\rangle\}$ of $Q_{r-1}^{\mathrm{in}}$.
Then
\[
    \vket{C_{o_r|h}}
       =\sum_{a,i}(C_{o_r|h})_{ai}
                     |a\rangle\otimes|i\rangle,
    \qquad
    \mathbf C_{o_{1:r}}
       =\vket{C_{o_r|h}}\otimes\mathbf C_h.
\]
Component completeness gives
\[
    \sum_{o_r,a}
       (C_{o_r|h})_{ai}^{\star}(C_{o_r|h})_{aj}
       =
       \left(\sum_{o_r}C_{o_r|h}^\dagger C_{o_r|h}\right)_{ij}
       =\delta_{ij}.
\]
Consequently, expanding the contraction yields
\begin{equation}
    \begin{aligned}
        &\sum_{o_r}
          \mathbf C_{o_{1:r}}^\dagger
          \left(
              I_{A_r^{\mathrm{out}}Q_r^{\mathrm{out}}}
              \otimes\mathbf E^{(r-1)}
          \right)\mathbf C_{o_{1:r}}\\
        &\quad=
          \mathbf C_h^\dagger
          \left[
              \sum_{o_r,a,i,j}
              (C_{o_r|h})_{ai}^{\star}(C_{o_r|h})_{aj}
              (\langle i|\otimes I)\,
                 \mathbf E^{(r-1)}\,
              (|j\rangle\otimes I)
          \right]\mathbf C_h\\
        &\quad=
          \mathbf C_h^\dagger
          \left[
              \sum_i
              (\langle i|\otimes I)\,
                 \mathbf E^{(r-1)}\,
              (|i\rangle\otimes I)
          \right]\mathbf C_h\\
        &\quad=
          \mathbf C_h^\dagger
          \left(
              \operatorname{Tr}_{Q_{r-1}^{\mathrm{in}}}
                    \mathbf E^{(r-1)}
          \right)\mathbf C_h.
    \end{aligned}
    \label{eq:app-comb-connected-trace}
\end{equation}
Here $I$ acts on
$\mathcal H_{r-1}\otimes Q_0^{\mathrm{ext}}$,
and the final equality is the definition of the partial trace.
For $r>1$, the noise recursion gives
\[
    \operatorname{Tr}_{Q_{r-1}^{\mathrm{in}}}
          \mathbf E^{(r-1)}
       =
       I_{A_{r-1}^{\mathrm{out}}Q_{r-1}^{\mathrm{out}}}
          \otimes\mathbf E^{(r-2)}.
\]
Thus summing \eqref{eq:app-comb-connected-trace} over $h$
reduces the expression from $r$ components to $r-1$.
Repeating this reduction gives
\begin{equation}
    \begin{aligned}
        \sum_{e,o}
           (\mathbf E_e\mathbf C_o)^\dagger
            (\mathbf E_e\mathbf C_o)
        &=
        \sum_o\mathbf C_o^\dagger
        \left(
            I_{A_\ell^{\mathrm{out}}Q_\ell^{\mathrm{out}}}
            \otimes\mathbf E^{(\ell-1)}
        \right)\mathbf C_o\\
        &=
        \cdots
        =
        \sum_{o_1}\mathbf C_{o_1}^\dagger
        \left(
            I_{A_1^{\mathrm{out}}Q_1^{\mathrm{out}}}
            \otimes\mathbf E^{(0)}
        \right)\mathbf C_{o_1}\\
        &=
        \mathbf C_{\varnothing}^\dagger
        \left(
            \operatorname{Tr}_{Q_0^{\mathrm{in}}}
                \mathbf E^{(0)}
        \right)\mathbf C_{\varnothing}\\
        &=I_{Q_0^{\mathrm{ext}}}.
    \end{aligned}
    \label{eq:app-comb-connected-tp}
\end{equation}
The first equality uses
\eqref{eq:app-comb-noise-normalisation}, while the last uses
$\mathbf C_{\varnothing}=I_{Q_0^{\mathrm{ext}}}$
and \eqref{eq:app-comb-noise-boundary}.
Hence the connected map is CPTP.
Identity noise gives the ideal contraction as a special case.

\section{18-qubit instance of the Hastings--Haah honeycomb Floquet code}
\label{app:honeycomb}
\label{ex:honeycomb-level-1}

Consider an $18$-qubit instance of the Hastings--Haah honeycomb
floquet code \cite{hastings2021dynamically}. The corresponding measurement schedule matrix is given by
\begin{equation}
 \mathbf M=\bigl[(\mathbf M^{(0)})^{\mathsf T}\;
 (\mathbf M^{(1)})^{\mathsf T}\;
 (\mathbf M^{(2)})^{\mathsf T}\;
 (\mathbf M^{(0)})^{\mathsf T}\;
 (\mathbf M^{(1)})^{\mathsf T}\;
 (\mathbf M^{(2)})^{\mathsf T}\;
 (\mathbf M^{(0)})^{\mathsf T}\bigr]^{\mathsf T}
 \in\mathbb F_2^{63\times36}.
 \label{eq:honeycomb-M}
\end{equation}
For $a\in\{0,1,2\}$, the rows of
$\mathbf M^{(a)}\in\mathbb F_2^{9\times36}$ encode the
nine Pauli checks listed in column $\mathbf M^{(a)}$ of
\Cref{fig:honeycomb-finite-lattice}.
Specifically, a check on qubits $i,j$ is encoded by
\begin{equation}
 X_iX_j\longmapsto[f_i+f_j,0],\qquad
 Y_iY_j\longmapsto[f_i+f_j,f_i+f_j],\qquad
 Z_iZ_j\longmapsto[0,f_i+f_j],
\end{equation}
where $f_i$ is the $i$th standard basis row vector of
$\mathbb F_2^{18}$ and $0$ denotes the zero row vector
of length $18$. Round $t\in\{0,\ldots,6\}$ uses block
$\mathbf M^{(t\bmod3)}$, with its checks measured in
row order $s=1,\ldots,9$. Define
\begin{equation}
 \mathbf m_{9t+s}
 \triangleq
 \bigl(\mathbf M^{(t\bmod3)}\bigr)_{s,:}
 \in\mathbb F_2^{36},
 \qquad
 M_{9t+s}\triangleq P(\mathbf m_{9t+s}),
\end{equation}
where $(\cdot)_{s,:}$ denotes the $s$th row.
Thus $\mathbf m_r$ is the binary Pauli label of the
$r$th measurement, with $r=9t+s$.
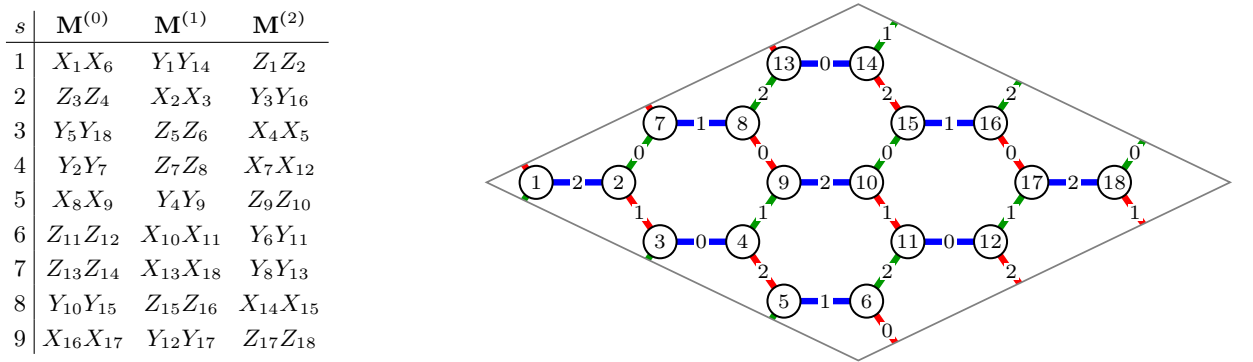
\begin{figure}[htbp]
\centering
\begin{minipage}[c]{.43\linewidth}
\centering
$\begin{array}{c|ccc}
 s&\mathbf M^{(0)}&\mathbf M^{(1)}&\mathbf M^{(2)}\\\hline
 1&X_1X_6&Y_1Y_{14}&Z_1Z_2\\
 2&Z_3Z_4&X_2X_3&Y_3Y_{16}\\
 3&Y_5Y_{18}&Z_5Z_6&X_4X_5\\
 4&Y_2Y_7&Z_7Z_8&X_7X_{12}\\
 5&X_8X_9&Y_4Y_9&Z_9Z_{10}\\
 6&Z_{11}Z_{12}&X_{10}X_{11}&Y_6Y_{11}\\
 7&Z_{13}Z_{14}&X_{13}X_{18}&Y_8Y_{13}\\
 8&Y_{10}Y_{15}&Z_{15}Z_{16}&X_{14}X_{15}\\
 9&X_{16}X_{17}&Y_{12}Y_{17}&Z_{17}Z_{18}
\end{array}$
\end{minipage}\hfill
\begin{minipage}[c]{.55\linewidth}
\centering
\begingroup
\def\HCA#1#2{({1.5*((#1)+(#2))},{.8660254*(-(#1)+(#2))})}
\def\HCB#1#2{({1.5*((#1)+(#2))+1},{.8660254*(-(#1)+(#2))})}
\def\HCDomainPath{(-.6,0)--(3.9,-2.598076)--(8.4,0)--(3.9,2.598076)--cycle}
\resizebox{\linewidth}{!}{%
\begin{tikzpicture}[x=1cm,y=.83cm,
 hcx/.style={draw=red,line width=2.2pt,line cap=round},
 hcy/.style={draw=green!60!black,line width=2.2pt,line cap=round},
 hcz/.style={draw=blue,line width=2.2pt,line cap=round},
 hclabel/.style={fill=white,inner sep=.5pt,font=\scriptsize},
 hcvertex/.style={circle,draw,line width=.7pt,fill=white,
   minimum size=4mm,inner sep=0pt,font=\scriptsize}]
 \begin{scope}
  \clip \HCDomainPath;
  \foreach \i in {-1,...,3}{\foreach \j in {-1,...,3}{
   \pgfmathtruncatemacro{\cx}{mod(mod(\i+2*\j,3)+3,3)}
   \pgfmathtruncatemacro{\cy}{mod(\cx+1,3)}
   \pgfmathtruncatemacro{\cz}{mod(\cx+2,3)}
   \draw[hcx] \HCA{\i}{\j}
    -- node[hclabel]{\cx} \HCB{\i-1}{\j};
   \draw[hcy] \HCA{\i}{\j}
    -- node[hclabel]{\cy} \HCB{\i}{\j-1};
   \draw[hcz] \HCA{\i}{\j}
    -- node[hclabel]{\cz} \HCB{\i}{\j};
  }}
 \end{scope}
 \draw[gray,line width=.6pt] \HCDomainPath;
 \foreach \i in {0,1,2}{\foreach \j in {0,1,2}{
  \pgfmathtruncatemacro{\qa}{2*(3*\j+\i+1)-1}
  \pgfmathtruncatemacro{\qb}{\qa+1}
  \node[hcvertex] at \HCA{\i}{\j} {\qa};
  \node[hcvertex] at \HCB{\i}{\j} {\qb};
 }}
\end{tikzpicture}%
}
\endgroup
\end{minipage}
\caption{Complete measurement schedule. Left: Pauli checks specifying
rows of the three binary blocks of $\mathbf M$. Right: Honeycomb lattice with opposite sides identified. Circles denote qubit numbers. Edge numbers
give the round type $0,1,2$, repeated as $0,1,2,0,1,2,0$.
Red, green, and blue edges denote $XX$, $YY$, and $ZZ$, respectively.}
\label{fig:honeycomb-finite-lattice}
\end{figure}
From \Cref{def:L1-schedule-matrices} it follows,
\begin{equation}
 \mathbf Q=\mathbf M\Omega_{18},\qquad
 T_{rs}=\begin{cases}
 \langle\mathbf m_s,\mathbf m_r\rangle_s,&s<r,\\0,&s\ge r,
 \end{cases}
 \qquad
 \operatorname{rank}\mathbf T=36,\qquad
 \operatorname{rank}[\mathbf T\ \mathbf M]
 =\operatorname{rank}[\mathbf T^{\mathsf T}\ \mathbf M]=52.
 \label{eq:honeycomb-ranks}
\end{equation}

\begin{proposition}[Common relations and preserved logical qubits]
\label{prop:honeycomb-relations}
For $C_{\mathrm{enc}}=I$, the circuit satisfies
\begin{equation}
\begin{gathered}
 \dim\mathcal W_{\mathrm{det}}=11,\qquad
 \dim\mathcal W_{\mathrm{obs}}
 =\dim\mathcal W_{\mathrm{out}}=27,\qquad
 \dim\mathcal W_{\mathbf C}=47,\\
 V_{\mathrm{in}}=V_{\mathrm{out}},\qquad
 \dim V_{\mathrm{in}}=16,\qquad
 \dim\mathcal L_{\mathrm{POP}}=4,\qquad
 k_{\mathrm{pres}}=2.
\end{gathered}
\label{eq:honeycomb-dimensions}
\end{equation}
All detector parities vanish, and
\begin{equation}
 \Omega_{\mathbf C}=\mathcal D_{\mathbf C}^{\perp},
 \qquad |\Omega_{\mathbf C}|=2^{52}.
 \label{eq:honeycomb-realisable-outcomes}
\end{equation}
A symplectic basis of the preserved logical classes is
represented by
\begin{equation}
\begin{aligned}
 \overline X_1&=X_1X_2X_7X_8X_{13}X_{14},&
 \overline Z_1&=Y_1Y_3X_4Z_6,\\
 \overline X_2&=X_2X_7X_8X_{10}X_{11}X_{12}X_{15},&
 \overline Z_2&=Y_1Y_2Y_3Y_4Y_5Y_6.
\end{aligned}
\label{eq:honeycomb-logical-generators}
\end{equation}
A choice of their record-vector representatives is given by
\begin{equation}
\begin{gathered}
 b_{\overline X_1}=b_{\overline Z_2}=0,\qquad
 b_{\overline Z_1}=\sum_{r\in J_Z}\epsilon_r,\qquad
 b_{\overline X_2}=\sum_{r\in J_X}\epsilon_r,\\
 J_Z=\{2,11,12,19,21,28,29,38,39,46,48,55\},\\
 J_X=\{2,4,5,9,12,18,19,20,21,22,28,29,31,36,
        39,45,46,47,48,49,55,59\},
\end{gathered}
\label{eq:honeycomb-logical-records}
\end{equation}
where $\epsilon_r$ is the $r$th standard basis row vector
of $\mathbb F_2^{63}$. For each displayed $\overline P$ in \Cref{eq:honeycomb-logical-generators},
\begin{equation}
 (Z_A^{b_{\overline P}}\otimes\overline P)C_o
 =C_o\overline P.
 \label{eq:honeycomb-logical-POPs}
\end{equation}
\end{proposition}

\begin{proof}
By \Cref{prop: code-conditioned stabilizers level 1},
the common relations $\omega=[b,p,p']$ are obtained by solving the linear system
\begin{equation}
 p\mathbf Q^{\mathsf T}+b\mathbf T^{\mathsf T}=0,
 \qquad p'=p+b\mathbf M,
 \label{eq:honeycomb-relation-system}
\end{equation}
with $c_\omega=\Theta(p,b)$.
By \Cref{cor:L1-extraction}, and by setting $p=0$ in
\Cref{eq:honeycomb-relation-system} for output checks,
we obtain
\begin{equation}
\begin{array}{rll}
 \text{Detectors:}
   &d\mathbf T=0,&d\mathbf M=0,\\
 \text{Input observables:}
   &v\mathbf T=0,&v\mathbf M=p,\\
 \text{Output checks:}
   &u\mathbf T^{\mathsf T}=0,&u\mathbf M=p'.
\end{array}
\label{eq:honeycomb-extraction-systems}
\end{equation}
Gaussian elimination over $\mathbb F_2$ yields
$V_{\mathrm{in}}=V_{\mathrm{out}}$.
The ranks in \Cref{eq:honeycomb-ranks} give detector
dimension $63-52=11$ and observable and output-check
dimensions $63-36=27$. Hence
$\dim V_{\mathrm{in}}=27-11=16$, while
\Cref{cor:L1-extraction} gives
$k_{\mathrm{pres}}=18-52+36=2$.
By \Cref{prop:outcome-independent-stabilizer-relations},
$\dim\mathcal L_{\mathrm{POP}}=4$ and
$\dim\mathcal W_{\mathbf C}=11+16+16+4=47$.
Phase evaluation gives $\Theta(0,d)=0$ on a detector
basis, proving \Cref{eq:honeycomb-realisable-outcomes}. The binary labels $\overline x_i,\overline z_i$ of the
operators in \Cref{eq:honeycomb-logical-generators} satisfy
\begin{equation}
\begin{gathered}
 \overline x_i,\overline z_i\in V_{\mathrm{in}}^{\perp_s},
 \qquad
 \langle\overline x_i,\overline z_j\rangle_s=\delta_{ij},\\
 \langle\overline x_i,\overline x_j\rangle_s
 =\langle\overline z_i,\overline z_j\rangle_s=0,
 \qquad i,j\in\{1,2\}.
\end{gathered}
\end{equation}
Since $\dim(V_{\mathrm{in}}^{\perp_s}/V_{\mathrm{in}})=4$,
their classes form a symplectic basis.
Solving \Cref{eq:honeycomb-relation-system} with
$p=p'=\overline p$ yields the representatives in
\Cref{eq:honeycomb-logical-records}, which satisfy
\begin{equation}
 b_{\overline P}\mathbf M=0,\qquad
 \overline p\mathbf Q^{\mathsf T}
 +b_{\overline P}\mathbf T^{\mathsf T}=0,\qquad
 \Theta(\overline p,b_{\overline P})=0.
\end{equation}
These identities establish \Cref{eq:honeycomb-logical-POPs}.
\end{proof}

\begin{corollary}[Code-conditioned honeycomb circuit]
\label{cor:honeycomb-conditioned}
Let $C_{\mathrm{enc}}$ encode two qubits into the common
$+1$ eigenspace of the rank-$16$ stabilizer group
\begin{equation}
 \mathcal S=
 \left\{(-1)^{\Theta(v\mathbf M,v)}P(v\mathbf M):
                         v\mathbf T=0\right\},
 \qquad
 \widetilde C_o\triangleq C_oC_{\mathrm{enc}}.
 \label{eq:honeycomb-encoder}
\end{equation}
For $\widetilde{\mathbf C}$,
\begin{equation}
\begin{gathered}
 \mathcal D_{\widetilde{\mathbf C}}=\ker\mathbf T,
 \qquad
 \dim\mathcal D_{\widetilde{\mathbf C}}=27,
 \qquad
 \mathcal W_{\mathrm{obs}}=\mathcal W_{\mathrm{det}},\\
 \Omega_{\widetilde{\mathbf C}}=(\ker\mathbf T)^\perp,
 \qquad
 |\Omega_{\widetilde{\mathbf C}}|=2^{36},\\
 \dim V_{\mathrm{out}}=16,\qquad
 k_{\mathrm{pres}}=2,\qquad
 \dim\mathcal W_{\widetilde{\mathbf C}}=47.
\end{gathered}
\label{eq:honeycomb-conditioned-detectors}
\end{equation}
All detector parities vanish. Conditioning leaves
$\mathbf M,\mathbf Q,\mathbf T$ unchanged.
\end{corollary}

\begin{proof}
By \Cref{prop: code-conditioned stabilizers level 1},
an observable relation $[v,p,0]$ requires
$\nu=v\mathbf M$ and $v\mathbf T=0$.
Thus $\nu$ is an encoder stabilizer label, so $p=0$.
Conversely, every $v\in\ker\mathbf T$ gives a detector
whose parity vanishes by the choice of $\mathcal S$.
This proves the detector and observable claims;
$\operatorname{rank}\mathbf T=36$ gives the dimensions
and outcome count.
The sixteen output-check labels and four logical
classes in \Cref{prop:honeycomb-relations} remain the same after
code-conditioning. Hence
\Cref{prop:outcome-independent-stabilizer-relations}
gives $k_{\mathrm{pres}}=2$ and
$\dim\mathcal W_{\widetilde{\mathbf C}}=27+16+4=47$.
\end{proof}

\subsection{Circuit distance}
\label{app:honeycomb-circuit-distance}
Let $\mathcal V^{(1)}\subseteq\mathcal E^{(1)}$ allow arbitrary Pauli
faults on each data qubit immediately after each individual measurement
$r=1,\ldots,62$, and on each ancilla measurement qubit immediately after its creation.
All other coordinates vanish. There are $18\cdot62+63=1179$ locations
and $2358$ binary coordinates. The encoder and both data boundaries are
noiseless; all records are retained, and output checks and recovery are
noiseless. Record $Z$ faults have zero signature.

\begin{proposition}[Distance]
For the conditioned circuit and fault space $\mathcal{V}^{(1)}$, $d_{\mathrm{circ},1}=4$.
\end{proposition}
\begin{proof}
Using \Cref{def:ft-signature-matrices,prop:ft-level1-signature},
construct $H_1$ from the $27$ detectors and $16$ additional
output checks of \Cref{cor:honeycomb-conditioned}, and
$L_1$ from the logical relations in
\Cref{eq:honeycomb-logical-POPs}, ordered as
$\overline X_1,\overline Z_1,\overline X_2,\overline Z_2$. Define
\begin{equation}
 \mathcal F_1
 \triangleq \{e\in\mathcal V^{(1)}:
                  \operatorname{wt}_{\mathrm{ST}}(e)\leq1\},
 \qquad
 \mathcal F_1+\mathcal F_1
 \triangleq \{a+b:a,b\in\mathcal F_1\}.
\end{equation}
By construction, $\mathcal F_1+\mathcal F_1$ is exactly the set
of allowed spacetime fault patterns of weight at most two. Computing $(H_1e^{\mathsf T},L_1e^{\mathsf T})$ for
$e\in\mathcal F_1$ gives $299$ distinct signatures.
By linearity, adding every pair of these signatures gives
the $40472$ distinct signatures of
$\mathcal F_1+\mathcal F_1$.
Exact binary comparison verifies
\begin{equation}
 \begin{gathered}
 H_1a^{\mathsf T}=H_1b^{\mathsf T}
 \ \Longrightarrow\
 L_1a^{\mathsf T}=L_1b^{\mathsf T},\\
 a\in\mathcal F_1,\qquad b\in\mathcal F_1+\mathcal F_1.
 \end{gathered}
 \label{eq:honeycomb-signature-comparison}
\end{equation}

Every allowed pattern $e$ of weight at most three splits
as $e=a+b$, with $a\in\mathcal F_1$ and
$b\in\mathcal F_1+\mathcal F_1$.
If $H_1e^{\mathsf T}=0$, then
$H_1a^{\mathsf T}=H_1b^{\mathsf T}$, so
\Cref{eq:honeycomb-signature-comparison} gives
\[
 L_1e^{\mathsf T}
 =L_1a^{\mathsf T}+L_1b^{\mathsf T}=0.
\]
Hence no allowed fault of weight at most three has zero
syndrome and nonzero logical signature. The circuit distance
(\Cref{def:ft-circuit-distance}) therefore gives
$d_{\mathrm{circ},1}\geq4$. Conversely, apply $Y_1X_2Z_7$ immediately after measurement
$1$ and $X_{18}$ immediately after measurement $3$.
This allowed pattern $e_*$ has weight four and satisfies
\[
 H_1e_*^{\mathsf T}=0,\qquad
 L_1e_*^{\mathsf T}=(0,0,1,1)^{\mathsf T}\ne0.
\]
Thus $d_{\mathrm{circ},1}\leq4$, proving
\begin{equation}
 d_{\mathrm{circ},1}=4.
 \label{eq:honeycomb-circuit-distance}
\end{equation}
\end{proof}

\section{Syndrome extraction on CSS codes with flag circuits}
\label{subsec:flag-syndrome-extraction}

We apply the fault signature formalism of
\Cref{sec: Fault tolerance of stabilizer circuits} and the LP bounds of
\Cref{sec:linear-programming-bounds} to a family of flag circuits
prescribed by their CNOT orderings for measuring a fixed $Z$-type
stabilizer of a CSS code
\cite{ChaoReichardt2018TwoQubits,ChamberlandBeverland2018Flags,PatoEtAl2024Flag,chao2020flag}.
The analysis extends to a chosen $X$-type stabilizer measurement by using its support, reversing
all CNOT directions, and interchanging $|0\rangle$ with $|+\rangle$
in ancilla preparations and $Z$ with $X$ in ancilla measurements,
with the corresponding changes to the fault-propagation and
signature matrices. The $t$-flag condition \cite{ChamberlandBeverland2018Flags} supplies additional constraints to the distance bound LP through vanishing
fault counts. Finally, we analytically benchmark the LP distance bound for CSS Hamming
codes and verify that it is attained by the one-flag construction of
Chao and Reichardt \cite{ChaoReichardt2018TwoQubits} in \Cref{subsec:flp-css-hamming}.
\subsection{Z-measurement flag circuit family, signature matrices and fault weight profiles}

\begin{definition}[Flag-circuit family for $Z$-type stabilizer measurement]
\label{def:fw-family}
Fix an $[[n,k_L,d_{\mathrm{code}}]]$ CSS code with an isometric stabilizer encoder
$C_{\mathrm{enc}}:Q_{\mathrm L}\to Q\triangleq\bigotimes_{j=1}^{n}Q_j$,
where $Q_{\mathrm L}\cong\C^{2^{k_L}}$ and $Q_j\cong\C^2$,
and let $P_Z=\prod_{j=1}^{w}Z_{Q_j}$ be the measured
weight-$w$ stabilizer, supported on $Q_1,\ldots,Q_w$,
with $1\leq w\leq n$ and $P_ZC_{\mathrm{enc}}=C_{\mathrm{enc}}$.
Let $A\cong\C^2$ be the syndrome ancilla and
$F=\bigotimes_{a=1}^{f}F_a$ the flag register,
where $F_a\cong\C^2$ and $f\geq1$.
Write $Q_{\ne j}\triangleq\bigotimes_{r\in\mathbb Z^+_n\setminus\{j\}}Q_r$
and $F_{\ne a}\triangleq\bigotimes_{b\in\mathbb Z^+_f\setminus\{a\}}F_b$.
Put $g=w+2f$ and define linear operators
\begin{equation}
 \begin{aligned}
 \mathcal G_j
   &\triangleq\cnot_{Q_j\to A}
       \otimes I_{F\otimes Q_{\ne j}},
   &&j\in \mathbb Z^+_w,\\
 \mathcal G_{w+2a-1}=\mathcal G_{w+2a}
   &\triangleq\cnot_{F_a\to A}
       \otimes I_{F_{\ne a}\otimes Q},
   &&a\in \mathbb Z^+_f.
 \end{aligned}
 \label{eq:flag-gate-occurrences}
\end{equation}
The preparation and measurement components, common to every
ordering, are defined as
\begin{equation}
 \begin{aligned}
 C^{\mathrm{prep}}
   \triangleq |0\rangle_A\otimes
\left(\bigotimes_{a=1}^{f}|+\rangle_{F_a}\right)
      \otimes C_{\mathrm{enc}},\quad \mathrm{and} \quad  
 C_{s,\alpha}^{\mathrm{meas}}
   \triangleq |s\rangle\langle s|_A
      \otimes\bigl(|\alpha\rangle\langle\alpha|_F H^{\otimes f}\bigr)
      \otimes I_Q,
 \end{aligned}
 \label{eq:flag-preparation-readout}
\end{equation}
for $s\in\mathbb F_2$ and $\alpha\in\mathbb F_2^f$, respectively.
Let $\mathfrak P$ be a nonempty set of permutations
$\pi:\mathbb Z_g^+\to\mathbb Z_g^+$, with $\pi(i)$ specifying
the gate at time step $i$.
Define the family $\{\mathbf C_\pi:\pi\in\mathfrak P\}$ by
\begin{equation}
 \begin{gathered}
 \mathbf C_\pi\triangleq
   \{C^{\mathrm{prep}},\mathcal G_{\pi(1)},\ldots,
     \mathcal G_{\pi(g)},\{C_{s,\alpha}^{\mathrm{meas}}\}_{s,\alpha}\},\quad \mathrm{with} \quad 
 U_\pi\triangleq
   \mathcal G_{\pi(g)}\cdots\mathcal G_{\pi(1)}
   \in\operatorname{Cliff}(A\otimes F\otimes Q).
 \end{gathered}
 \label{eq:flag-circuit-components}
\end{equation}
Its ideal realization $\mathbf C_{\mathrm{id},\pi}$
has Kraus operators $\{C_{s,\alpha}^{\pi}: Q_{\mathrm L}\to A\otimes F\otimes Q\}_{s,\alpha}$ given by
\begin{equation}
  C_{s,\alpha}^{\pi}
 \triangleq C_{s,\alpha}^{\mathrm{meas}}U_\pi C^{\mathrm{prep}}= \begin{quantikz}[
     row sep={0.6cm,between origins},
     column sep=0.18cm
 ]
     \lstick{$\ket{0}_A$}
         & \qw
         & \gate[3]{U_\pi}
         & \qw
         & \meter{}
         & \rstick{$s$}\setwiretype{c}\\
     \lstick{$\ket{+}^{\otimes f}_F$}\setwiretype{b}
         & \qw
         & \qw
         & \gate{H^{\otimes f}}
         & \meter{}
         & \rstick{$\alpha$}\setwiretype{c}\\
     \lstick{$Q_{\mathrm L}$}\setwiretype{b}
         & \gate{C_{\mathrm{enc}}}
         & \qw
         & \qw
         & \qw
         & \rstick{$Q$}\setwiretype{b}
 \end{quantikz}.
 \label{eq:flag-circuit-diagram}
\end{equation}
\end{definition}
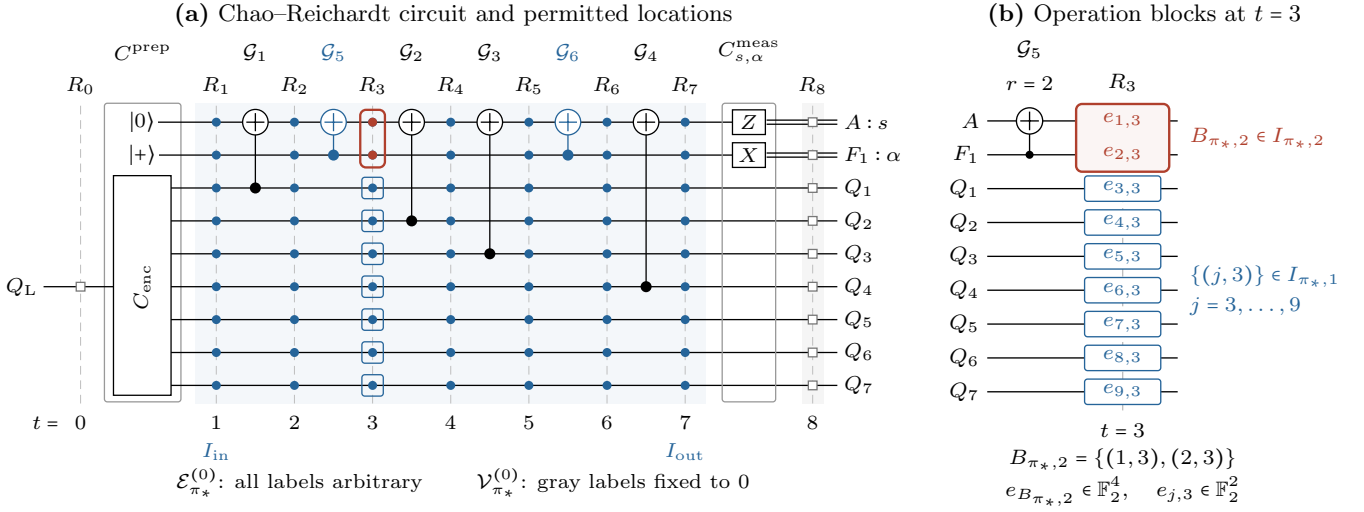
\begin{figure*}[t]
\centering
\begingroup
\definecolor{FlagBlue}{RGB}{37,100,153}
\definecolor{FlagRed}{RGB}{176,64,45}
\definecolor{FlagGray}{RGB}{120,120,120}

\begin{minipage}[t]{.665\textwidth}
\centering
\textbf{(a)} Chao--Reichardt circuit and permitted locations
\par\smallskip
\resizebox{\linewidth}{!}{%
\begin{tikzpicture}[
 x=1cm,y=1cm,line width=.45pt,
 every node/.style={font=\scriptsize,inner sep=1pt},
 guide/.style={
   draw=black!25,densely dashed,line width=.35pt
 },
 allowed/.style={
   circle,fill=FlagBlue,inner sep=1.1pt
 },
 fixed/.style={
   rectangle,draw=FlagGray,fill=white,
   minimum size=3pt,inner sep=0pt
 },
 target/.style={
   circle,draw,fill=white,
   minimum size=3.1mm,inner sep=0pt
 },
 control/.style={circle,fill,inner sep=1.35pt}
]
 \fill[FlagBlue!5] (1.40,-.22) rectangle (7.60,3.43);
 \fill[black!4] (8.77,-.22) rectangle (9.03,3.43);
 \foreach \t/\x in
 {0/0,1/1.65,2/2.60,3/3.55,4/4.50,
  5/5.45,6/6.40,7/7.35,8/8.90}{
   \draw[guide] (\x,-.22)--(\x,3.43);
   \node at (\x,3.65) {$R_{\t}$};
   \node at (\x,-.45) {$\t$};
 }
 \node[anchor=east] at (-.22,-.45) {$t=$};

 \draw[draw=black!40,rounded corners=1pt]
   (.29,-.20) rectangle (1.22,3.43);
 \node at (.755,4.05) {$C^{\mathrm{prep}}$};
 \draw (-.45,1.20)--(.40,1.20);
 \node[anchor=east] at (-.48,1.20) {$Q_{\mathrm L}$};
 \node[fixed] at (0,1.20) {};

 \foreach \y in {2.4,2.0,1.6,1.2,.8,.4,0}
   \draw (.75,\y)--(9.20,\y);
 \foreach \y in {3.2,2.8}
   \draw (.75,\y)--(8.125,\y);

 \node[fill=white] at (.75,3.20) {$|0\rangle$};
 \node[fill=white] at (.75,2.80) {$|+\rangle$};
 \draw[fill=white] (.40,-.12) rectangle (1.10,2.55);
 \node[rotate=90] at (.75,1.20) {$C_{\mathrm{enc}}$};

 \foreach \x/\y/\gate in
 {2.125/2.4/1,4.025/2.0/2,
  4.975/1.6/3,6.875/1.2/4}{
   \draw (\x,\y)--(\x,3.2);
   \node[control] at (\x,\y) {};
   \node[target] at (\x,3.2) {};
   \draw
     (\x-.11,3.2)--(\x+.11,3.2)
     (\x,3.09)--(\x,3.31);
   \node at (\x,4.05) {$\mathcal G_{\gate}$};
 }

 \foreach \x/\gate in {3.075/5,5.925/6}{
   \draw[FlagBlue] (\x,2.8)--(\x,3.2);
   \node[control,fill=FlagBlue] at (\x,2.8) {};
   \node[target,draw=FlagBlue] at (\x,3.2) {};
   \draw[FlagBlue]
     (\x-.11,3.2)--(\x+.11,3.2)
     (\x,3.09)--(\x,3.31);
   \node[text=FlagBlue] at (\x,4.05)
     {$\mathcal G_{\gate}$};
 }

 \draw[draw=black!40,rounded corners=1pt]
   (7.80,-.20) rectangle (8.45,3.43);
 \node at (8.125,4.05) {$C^{\mathrm{meas}}_{s,\alpha}$};
 \foreach \y/\basis in {3.2/Z,2.8/X}{
   \node[
     draw,fill=white,
     minimum width=4mm,minimum height=3mm
   ] at (8.125,\y) {$\basis$};
   \draw (8.34,\y+.025)--(9.20,\y+.025);
   \draw (8.34,\y-.025)--(9.20,\y-.025);
 }
 \node[anchor=west] at (9.24,3.2) {$A:s$};
 \node[anchor=west] at (9.24,2.8) {$F_1:\alpha$};
 \foreach \j/\y in
 {1/2.4,2/2.0,3/1.6,4/1.2,5/.8,6/.4,7/0}
   \node[anchor=west] at (9.24,\y) {$Q_{\j}$};

 \foreach \x in {1.65,2.60,3.55,4.50,5.45,6.40,7.35}
   \foreach \y in {3.2,2.8,2.4,2.0,1.6,1.2,.8,.4,0}
     \node[allowed] at (\x,\y) {};
 \foreach \y in {3.2,2.8,2.4,2.0,1.6,1.2,.8,.4,0}
   \node[fixed] at (8.90,\y) {};

 \draw[FlagRed,rounded corners=2pt,line width=.8pt]
   (3.40,2.65) rectangle (3.70,3.35);
 \foreach \y in {3.2,2.8}
   \node[allowed,fill=FlagRed] at (3.55,\y) {};
 \foreach \y in {2.4,2.0,1.6,1.2,.8,.4,0}
   \draw[FlagBlue,rounded corners=1pt]
     (3.42,\y-.13) rectangle (3.68,\y+.13);

 \node[text=FlagBlue] at (1.65,-.80) {$I_{\mathrm{in}}$};
 \node[text=FlagBlue] at (7.35,-.80) {$I_{\mathrm{out}}$};
 \node[align=center] at (4.65,-1.18)
   {$\mathcal E_{\pi_*}^{(0)}$: all labels arbitrary
    \qquad
    $\mathcal V_{\pi_*}^{(0)}$: gray labels fixed to $0$};
\end{tikzpicture}}
\end{minipage}\hfill
\begin{minipage}[t]{.315\textwidth}
\centering
\textbf{(b)} Operation blocks at $t=3$
\par\smallskip
\resizebox{.93\linewidth}{!}{%
\begin{tikzpicture}[
 x=1cm,y=1cm,line width=.45pt,
 every node/.style={font=\scriptsize,inner sep=1pt},
 cell/.style={
   draw=FlagBlue,fill=white,rounded corners=1pt,
   minimum width=9mm,minimum height=3mm
 },
 target/.style={
   circle,draw,fill=white,
   minimum size=3.1mm,inner sep=0pt
 }
]
 \node at (.10,3.65) {$r=2$};
 \node at (1.20,3.65) {$R_3$};
 \node at (.10,4.05) {$\mathcal G_5$};
 \draw[black!25,densely dashed]
   (1.20,-.22)--(1.20,3.43);

 \foreach \j/\y/\lab in
 {1/3.2/A,2/2.8/F_1,3/2.4/Q_1,4/2.0/Q_2,
  5/1.6/Q_3,6/1.2/Q_4,7/.8/Q_5,8/.4/Q_6,9/0/Q_7}{
   \draw (-.40,\y)--(1.85,\y);
   \node[anchor=east] at (-.45,\y) {$\lab$};
 }

 \draw (.10,2.8)--(.10,3.2);
 \fill (.10,2.8) circle (1.35pt);
 \node[target] at (.10,3.2) {};
 \draw
   (-.01,3.2)--(.21,3.2)
   (.10,3.09)--(.10,3.31);

 \draw[
   FlagRed,fill=FlagRed!5,
   rounded corners=2pt,line width=.8pt
 ] (.65,2.60) rectangle (1.75,3.40);
 \node[text=FlagRed] at (1.20,3.20) {$e_{1,3}$};
 \node[text=FlagRed] at (1.20,2.80) {$e_{2,3}$};
 \node[anchor=west,text=FlagRed] at (1.95,3.00)
   {$B_{\pi_*,2}\in I_{\pi_*,2}$};

 \foreach \j/\y in
 {3/2.4,4/2.0,5/1.6,6/1.2,7/.8,8/.4,9/0}
   \node[cell,text=FlagBlue] at (1.20,\y) {$e_{\j,3}$};
 \node[anchor=west,text=FlagBlue,align=left]
   at (1.95,1.20)
   {$\{(j,3)\}\in I_{\pi_*,1}$\\[2pt]$j=3,\ldots,9$};

 \node at (1.20,-.45) {$t=3$};
 \node at (1.20,-.80)
   {$B_{\pi_*,2}=\{(1,3),(2,3)\}$};
 \node at (1.20,-1.18)
   {$e_{B_{\pi_*,2}}\in\mathbb F_2^4,\quad
     e_{j,3}\in\mathbb F_2^2$};
\end{tikzpicture}}
\end{minipage}
\endgroup

\caption{Chao--Reichardt one-flag circuit and operation partition.
(a) Steane-code (i.e., the $\mu=3$ specialization of the CSS Hamming code) example \cite{ChaoReichardt2018TwoQubits},
with measured support relabeled $Q_1,\ldots,Q_4$ and ordering
$\pi_*=(1,5,2,3,6,4)$.
Vertical guides mark the registers $R_t$.
Each marked location carries an arbitrary Pauli label in
$\mathcal E_{\pi_*}^{(0)}$.
In $\mathcal V_{\pi_*}^{(0)}$, gray labels vanish and all other
labels remain unrestricted, including both data boundaries.
(b) The red block groups the two outputs of the second CNOT;
each blue block is a separate idle-qubit operation location.
A nonzero block label counts as one faulty operation.
Applying this grouping at every step, with singleton blocks
at $t=1$, partitions $\Lambda_{\pi_*,\mathrm{act}}^{(0)}$.}
\label{fig:flag-family-grouping-hamming}
\label{fig:ham-ZA-propagation}
\end{figure*}
\begin{definition}[Fault locations and fault model]
\label{def:flag-pauli-model}
Put $n_*=n+f+1$. For each $\pi\in\mathfrak P$, write the
component domains and codomains as
\begin{equation}
 \begin{aligned}
 C^{\mathrm{prep}}:R_0\to R_1,\quad 
 \mathcal G_{\pi(r)}:R_r\to R_{r+1}
       \quad(r\in\mathbb Z^+_g), \quad
 C_{s,\alpha}^{\mathrm{meas}}:R_{g+1}\to R_{g+2},
 \end{aligned}
 \label{eq:flag-component-spaces}
\end{equation}
where $R_0=Q_{\mathrm L}$ and
$R_t\cong A\otimes F\otimes Q$ for $t \in \mathbb Z^+_{g+2}$.
Let $n_0=k_L$ and $n_t=n_*$ for $t\in\mathbb Z^+_{g+2}$.
Label the logical input qubits by $j\in\mathbb Z^+_{k_L}$.
For $t\in\mathbb Z^+_{g+2}$, label the qubits of $R_t$ by
$j\in\mathbb Z^+_{n_*}$ in the order
$A,F_1,\ldots,F_f,Q_1,\ldots,Q_n$, namely system $A$ has label $1$,
$F_1$ has label $2$, and so on. Following the spacetime-location convention of
\Cref{def:binary-fault-space}, define the spacetime locations,
spacetime fault space, and permitted fault subspace respectively by
\begin{equation}
 \begin{gathered}
 \Lambda_\pi^{(0)}\triangleq
 \{(j,t):t\in[g+2],\ j\in\mathbb Z^+_{n_t}\},\quad
 \mathcal E_\pi^{(0)}\triangleq
 \bigl\{e=[e_{j,t}]_{(j,t)\in\Lambda_\pi^{(0)}}:
                 e_{j,t}\in\mathbb F_2^2\bigr\},\\
 \mathcal V_\pi^{(0)}\triangleq
 \bigl\{e\in\mathcal E_\pi^{(0)}:
       e_{j,t}=[0,0]\text{ whenever }t\in\{0,g+2\}\bigr\}.
 \end{gathered}
 \label{eq:flag-fault-space}
\end{equation}
The actual locations $\Lambda_{\pi,\mathrm{act}}^{(0)}$ (see
\Cref{fig:ham-ZA-propagation}(b)) on which
$\mathcal V_\pi^{(0)}$ permits nontrivial faults are given by
\begin{equation}
 \Lambda_{\pi,\mathrm{act}}^{(0)}
 \triangleq\mathbb Z^+_{n_*}\times\mathbb Z^+_{g+1},\qquad
 N\triangleq|\Lambda_{\pi,\mathrm{act}}^{(0)}|
 =(g+1)n_*,
 \label{eq:flag-active-locations}
\end{equation}
with input and output data locations, immediately after encoding
and immediately before readout, respectively,
\begin{equation}
 \begin{aligned}
 I_{\mathrm{in}}
 \triangleq\{(f+1+j,1):j\in\mathbb Z^+_n\},\quad
 I_{\mathrm{out}}
 \triangleq\{(f+1+j,g+1):j\in\mathbb Z^+_n\}.
 \end{aligned}
 \label{eq:flag-data-boundaries}
\end{equation}
\end{definition}

In circuit-level fault counting, a nonidentity two-qubit
Pauli associated with one CNOT counts as one fault, even
when both qubits are affected
\cite[Sec.~2.1]{ChamberlandBeverland2018Flags}.  Spacetime weight instead counts the affected qubit locations.
We therefore group the two output locations of each CNOT
and use the resulting operation-wise weight and weight profiles
throughout the remainder of this section. The philosophy is to count the number of failed \emph{operations} in a circuit rather than counting the number of failed \emph{locations}.

\begin{definition}[Faulty operations and operation-wise weight profiles]
\label{def:flag-operation-weight}
For $r\in\mathbb Z^+_g$, let $c_{\pi,r}\in\mathbb Z^+_{n_*}\setminus\{1\}$
label the control qubit of $\mathcal G_{\pi(r)}$, whose target
has label $1$. Define
\begin{equation}
 \begin{gathered}
 B_{\pi,r}\triangleq\{(c_{\pi,r},r+1),(1,r+1)\},\\
 B_{\pi,r}\subseteq\Lambda_{\pi,\mathrm{act}}^{(0)},
 \qquad |B_{\pi,r}|=2,\\
 I_{\pi,2}\triangleq\{B_{\pi,r}:r\in\mathbb Z^+_g\},\\
 I_{\pi,1}\triangleq
 \bigl\{\{\ell\}:\ell\in
 \Lambda_{\pi,\mathrm{act}}^{(0)}
       \setminus\!\bigcup_{B\in I_{\pi,2}}B\bigr\}.
 \end{gathered}
 \label{eq:flag-operation-partition}
\end{equation}
The operation-wise location set
$\mathcal O_\pi\triangleq I_{\pi,1}\sqcup I_{\pi,2}$
partitions the actual spacetime locations in \Cref{def:flag-pauli-model} as
\[
 \Lambda_{\pi,\mathrm{act}}^{(0)}
 =\bigsqcup_{B\in\mathcal O_\pi}B.
\]
Its single-qubit and CNOT blocks form the partition
$\Pi_\pi\triangleq\{I_{\pi,1},I_{\pi,2}\}$, with
\begin{equation}
 \begin{aligned}
 N_1&\triangleq|I_{\pi,1}|=N-2g=n_*+g(n_*-2),\\
 N_2&\triangleq|I_{\pi,2}|=g.
 \end{aligned}
 \label{eq:flp-location-numbers}
\end{equation}
For $e\in\mathcal V_\pi^{(0)}$ and $B\in\mathcal O_\pi$,
write $e_B=(e_{j,t})_{(j,t)\in B}\in\mathbb F_2^{2|B|}$,
with qubit indices ordered increasingly within each block; while
operation $B$ is faulty when $e_B\ne0$.
Its operation-wise weight profile and total operation weight (see
\Cref{fig:ham-ZA-propagation}(b)) are
\begin{equation}
 \begin{gathered}
 a\triangleq|\{B\in I_{\pi,1}:e_B\ne0\}|,\qquad
 b\triangleq|\{B\in I_{\pi,2}:e_B\ne0\}|,
 \\
 \mathbf w_\pi(e)\triangleq(a,b)\in
 \mathcal J\triangleq[N_1]\times[N_2],
 \qquad
 \operatorname{wt}_{\mathrm{loc},\pi}(e)\triangleq a+b.
 \end{gathered}
 \label{eq:flag-fault-count}
\end{equation}
\end{definition}

\begin{remark}
    Hereafter, for the remainder of this section, omit the coordinates fixed to zero in
\Cref{eq:flag-fault-space}, identifying
$\mathcal V_\pi^{(0)}\cong
(\mathbb F_2^2)^{N_1}\oplus(\mathbb F_2^4)^{N_2}
\cong\mathbb F_2^{2N}$.
Write $e=(e_0,\ldots,e_g)$, where
$e_r=[x_r,z_r]\in\mathbb F_2^{2n_*}$ contains the labels
at time $t=r+1$ in the stated qubit order.
All subsequent matrices use these coordinates, and character
duals are taken in this reduced ambient space $\mathbb F_2^{2N}$.
Accordingly, the enumerator and LP formulas below use
$d_{0,\pi}=2N$.
\end{remark}

\begin{proposition}[Common ideal action]
\label{prop:fw-ideal}
Every ordering $\pi\in\mathfrak P$ has the same ideal branches,
\begin{equation}
 \begin{aligned}
 C_{s,\alpha}^{\pi}
 &=\delta_{\alpha,0}|s\rangle_A\otimes|0^f\rangle_F
   \otimes\frac{I_Q+(-1)^sP_Z}{2}C_{\mathrm{enc}}\\
 &=\delta_{s,0}\delta_{\alpha,0}|0\rangle_A\otimes|0^f\rangle_F
   \otimes C_{\mathrm{enc}}.
 \end{aligned}
 \label{eq:flag-ideal-branches}
\end{equation}
\end{proposition}
\begin{proof}
The CNOTs commute and square to the identity, so the flag CNOT pairs can be brought together by commuting them through the other CNOTs which makes them identity, hence leaving a measurement of $P_Z$, with
$P_ZC_{\mathrm{enc}}=C_{\mathrm{enc}}$.
\end{proof}

\begin{proposition}[Fault propagation]
\label{prop:flag-propagation}
Let $\varepsilon_j$ denote the $j$th standard basis row of
$\mathbb F_2^{n_*}$, with $\varepsilon_A=\varepsilon_1$.
For $0\leq r\leq g$, put
$v_{\pi,r}=\sum_{j=r+1}^{g}\varepsilon_{c_{\pi,j}}$.
The CNOTs applied at steps $r+1,\ldots,g$ have combined
unitary $\mathcal G_{\pi(g)}\cdots\mathcal G_{\pi(r+1)}$. Denote the symplectic matrix of their combined unitary by $F_{\pi,r}$ and define the \emph{propagation matrix} $\overrightarrow L_\pi$ by stacking them. Explicitly, the symplectic matrices have a closed form
\begin{equation}
 \begin{gathered}
 F_{\pi,r}=
 \begin{bmatrix}
 I_{n_*}+v_{\pi,r}^{\mathsf T}\varepsilon_A&0\\
 0&I_{n_*}+\varepsilon_A^{\mathsf T}v_{\pi,r}
 \end{bmatrix},\quad \mathrm{where} \quad 
 F_{\pi,g}=I_{2n_*},\qquad
 \overrightarrow L_\pi\triangleq
 \begin{bmatrix}F_{\pi,0}\\ \vdots\\ F_{\pi,g}\end{bmatrix}
 \in\mathbb F_2^{2N\times2n_*}.
 \end{gathered}
 \label{eq:flag-explicit-suffix}
\end{equation}
For every $e\in\mathcal V_\pi^{(0)}$, the propagated Pauli
label immediately before readout is
$e\overrightarrow L_\pi$. Define the coordinate selection matrices
$J_A\in\mathbb F_2^{2n_*\times1}$,
$J_F\in\mathbb F_2^{2n_*\times f}$, and
$J_Q\in\mathbb F_2^{2n_*\times2n}$ by
$pJ_A=x_A$, $pJ_F=z_F$, and $pJ_Q=[x_Q,z_Q]$ for
$p=[x_A,x_F,x_Q,z_A,z_F,z_Q]\in\mathbb F_2^{2n_*}$.
The data label and record flips are then
\begin{equation}
 \begin{aligned}
 q_\pi(e)\triangleq e\overrightarrow L_\pi J_Q,\quad 
 s_\pi(e)\triangleq e\overrightarrow L_\pi J_A,\quad
 \alpha_\pi(e)\triangleq e\overrightarrow L_\pi J_F.
 \end{aligned}
 \label{eq:flag-boundary-data-record}
\end{equation}
The unique nonzero Kraus operator of the noisy realization up to phase, is given by,
\[
 |s_\pi(e)\rangle_A\otimes|\alpha_\pi(e)\rangle_F
 \otimes P(q_\pi(e))C_{\mathrm{enc}}.
\]
\end{proposition}

\begin{proof}
The CNOT updates $x_A\mapsto x_A+x_c$ and
$z_c\mapsto z_c+z_A$, together with
\Cref{Prop: Symplectic Clifford,prop:ft-level0-signature},
give the propagation formulas.
The branch expression follows from
\Cref{eq:flag-ideal-branches}.
\end{proof}
Following \Cref{Cor: Logical dim Pres}, write $V\triangleq V_{\mathrm{enc}}^{\mathrm{out}}
 =\{q\in\mathbb F_2^{2n}:
       P(q)C_{\mathrm{enc}}=\pm C_{\mathrm{enc}}\}$. Choose matrices $\mathsf S\in\mathbb F_2^{(n-k_L)\times2n}$, $\mathsf G\in\mathbb F_2^{2k_L\times2n}$
whose rows are, respectively, a basis of $V$ and
representatives of a basis of $V^{\perp_s}/V$.
Thus $\mathsf S$ contains stabilizer labels,
$\mathsf G$ contains logical Pauli labels, and the rows
of $[\mathsf S;\mathsf G]$ form a basis of $V^{\perp_s}$.

\begin{proposition}[Signature matrices]
\label{prop:fw-signatures}
Valid level-$0$ signature matrices in the sense of
\Cref{def:ft-signature-matrices} are
\begin{equation}
 \begin{gathered}
 H_\pi=
 \begin{bmatrix}
 J_A^{\mathsf T}\\J_F^{\mathsf T}\\
 \mathsf S\Omega_nJ_Q^{\mathsf T}
 \end{bmatrix}\overrightarrow L_\pi^{\mathsf T},\quad 
 L_\pi=\mathsf G\Omega_nJ_Q^{\mathsf T}
                  \overrightarrow L_\pi^{\mathsf T},
 \qquad M_\pi=\begin{bmatrix}H_\pi\\L_\pi\end{bmatrix}.
 \end{gathered}
 \label{eq:fw-HL}
\end{equation}
Their action on a permitted fault $e\in\mathcal V_\pi^{(0)}$ is
\[
 H_\pi e^{\mathsf T}=
 \begin{bmatrix}
 s_\pi(e)\\
 \alpha_\pi(e)^{\mathsf T}\\
 \mathsf S\Omega_nq_\pi(e)^{\mathsf T}
 \end{bmatrix},
 \qquad
 L_\pi e^{\mathsf T}
 =\mathsf G\Omega_nq_\pi(e)^{\mathsf T}.
\]
Consequently,
\[
 \begin{aligned}
 H_\pi e^{\mathsf T}=0
 &\iff s_\pi(e)=0,\ \alpha_\pi(e)=0,\
        q_\pi(e)\in V^{\perp_s},\\
 M_\pi e^{\mathsf T}=0
 &\iff s_\pi(e)=0,\ \alpha_\pi(e)=0,\
        q_\pi(e)\in V.
 \end{aligned}
\]
The ranks are independent of the ordering $\pi$, namely,
\begin{equation}
 \begin{aligned}
h\triangleq\operatorname{rank}H_\pi=f+1+n-k_L,\quad
 \operatorname{rank}L_\pi =2k_L,\quad
m\triangleq\operatorname{rank}M_\pi=f+1+n+k_L.
 \end{aligned}
 \label{eq:flag-common-ranks}
\end{equation}
\end{proposition}

\begin{proof}
Apply \Cref{def:ft-signature-matrices} to the propagated
branch, using symplectic commutation and
$(V^{\perp_s})^{\perp_s}=V$.
Since $F_{\pi,g}=I_{2n_*}$, propagation is surjective,
so the independent boundary checks give the stated ranks.
\end{proof}

Write $p=e\overrightarrow L_\pi
=[x_A,x_F,x_Q,z_A,z_F,z_Q]$.
By the selector matrix definitions in \Cref{prop:flag-propagation}, the first two blocks of
$H_\pi e^{\mathsf T}$ are
\[
 \begin{bmatrix}
 J_A^{\mathsf T}\\ J_F^{\mathsf T}
 \end{bmatrix}
 \overrightarrow L_\pi^{\mathsf T}e^{\mathsf T}
 =
 \begin{bmatrix}
 (pJ_A)^{\mathsf T}\\ (pJ_F)^{\mathsf T}
 \end{bmatrix}
 =
 \begin{bmatrix}
 x_A\\ z_F^{\mathsf T}
 \end{bmatrix}.
\]
These are precisely the record-flip bits because
\[
 \begin{aligned}
 Z_AP(p)&=(-1)^{x_A}P(p)Z_A,\\
 X_{F_a}P(p)&=(-1)^{z_{F_a}}P(p)X_{F_a}.
 \end{aligned}
\]
Since in the noiseless case, measurement outcomes are zero by
\Cref{eq:flag-ideal-branches}, the noisy records are
$s_\pi(e)=x_A$ and $\alpha_\pi(e)=z_F$.
Similarly, the third block
$\mathsf S\Omega_nq_\pi(e)^{\mathsf T}$ gives the
commutation parities with the code stabilizer generators,
which define the ideal output checks. These checks are available
to the ideal recovery in the distance criterion; while the circuit's
retained measurement record is $(s,\alpha)$. Following, \cite[Definition~6]{ChamberlandBeverland2018Flags} we now define the $t$-flag criteria as follows.

\begin{definition}[$t$-flag condition]
\label{def:fw-tflag}
Let $p_Z\in V$ be the binary label of $P_Z$, and define
\[
 \operatorname{wt}_{P_Z}(q)
 \triangleq
 \min\{\operatorname{wt}P(q),\operatorname{wt}P(q+p_Z)\}.
\]
For an integer $t\geq1$, an ordering $\pi$ is
\emph{$t$-flag} if every $e\in\mathcal V_\pi^{(0)}$ satisfies
\begin{equation}
 \left.
 \begin{gathered}
 1\leq\operatorname{wt}_{\mathrm{loc},\pi}(e)\leq t,\\
 \alpha_\pi(e)=0
 \end{gathered}
 \right\}
 \ \Longrightarrow\
 \operatorname{wt}_{P_Z}(q_\pi(e))
 \leq\operatorname{wt}_{\mathrm{loc},\pi}(e).
 \label{eq:flag-t-condition}
\end{equation}
\end{definition}

\subsection{Flag-LP distance bound}
\label{subsec:flag-lp}

Use the spaces $\mathcal K_{0,\pi},\mathcal T_{0,\pi}$ and the
four distributions $U^\pi,T^\pi,\widehat U^\pi,\widehat T^\pi$
of \Cref{eq:we-fault-subspaces,eq:we-nested-distributions},
with operation profiles $\mathbf w_\pi$ and character duals
in the reduced space $\mathbb F_2^{2N}$.

\begin{definition}[Operation-wise distance]
\label{def:flag-operation-distance}
For $\pi\in\mathfrak P$, the operation-wise distance of
$\mathbf C_\pi$ is defined as
\begin{equation}
 d_\pi\triangleq
 \min_{\substack{e\in\mathcal V_\pi^{(0)}\\
                 H_\pi e^{\mathsf T}=0,\ L_\pi e^{\mathsf T}\ne0}}
       \operatorname{wt}_{\mathrm{loc},\pi}(e),
 \qquad \min\varnothing=\infty.
 \label{eq:flp-location-distance}
\end{equation}
\end{definition}
Each faulty operation affects one or two spacetime locations,
so $\operatorname{wt}_{\mathrm{loc},\pi}(e)
\leq\operatorname{wt}_{\mathrm{ST}}(e)
\leq2\operatorname{wt}_{\mathrm{loc},\pi}(e)$ and hence
$d_\pi\leq d_{\mathrm{circ},0}(\pi)\leq2d_\pi$. For each $\pi\in\mathfrak P$, the distributions
$U^\pi,T^\pi,\widehat U^\pi,\widehat T^\pi$ satisfy
the counting and MacWilliams constraints
\Cref{eq:lp-capacity,eq:lp-zero,eq:lp-cardinalities,eq:we-split-macwilliams}
with the substitutions
\begin{equation}
 \begin{gathered}
 \Gamma_{a,b}=3^a15^b\binom{N_1}{a}\binom{N_2}{b},\\
 \mathbb K_{i,j}(a,b)=
 \mathbb K_i^{(4)}(a;N_1)\mathbb K_j^{(16)}(b;N_2),\\
 d_{0,\pi}=2N=2N_1+4N_2,\\
 |\mathcal K_{0,\pi}|=2^{2N-h},\qquad
 |\mathcal T_{0,\pi}|=2^{2N-m}.
 \end{gathered}
 \label{eq:flp-krawtchouk-16}
\end{equation}
Here $\mathbb K_j^{(q)}$ is defined in
\Cref{eq:we-krawtchouk}.
Indeed, the character sum in
\Cref{eq:we-local-character-sum} over a nonzero CNOT label
in $\mathbb F_2^4$ equals $15$ for the trivial character
and $-1$ otherwise; the proof of
\Cref{prop:we-signature-macwilliams} therefore applies with
local alphabet sizes $4$ and $16$. To express the flag condition, retain the complete data label.
For $y=(s,\alpha,q)\in
\mathbb F_2\times\mathbb F_2^f\times\mathbb F_2^{2n}$, define
\begin{equation}
 C_{a,b}^\pi(y)\triangleq
 \left|\left\{e\in\mathcal V_\pi^{(0)}:
 \begin{gathered}
  \mathbf w_\pi(e)=(a,b),\\
  (s_\pi(e),\alpha_\pi(e),q_\pi(e))=y
 \end{gathered}
 \right\}\right|.
 \label{eq:flp-resolved-counts}
\end{equation}
These counts refine the signature counts, namely, the logical signature
alone does not determine $\operatorname{wt}_{P_Z}(q)$.

\begin{proposition}[Resolved counts and vanishing conditions]
\label{prop:flag-forbidden-sectors}
Put $\beta=1+f+2n$. For every $\pi\in\mathfrak P$, $C_{a,b}^\pi(y)\geq0$ and
\begin{equation}
 \begin{gathered}
 C_{0,0}^\pi(y)=\mathbf1_{\{y=0\}},\qquad
 \sum_yC_{a,b}^\pi(y)=\Gamma_{a,b},\\
 \sum_{(a,b)\in\mathcal J}C_{a,b}^\pi(y)=2^{2N-\beta},\\
 U_{a,b}^\pi=\sum_{q\in V^{\perp_s}}C_{a,b}^\pi(0,0,q),\\
 T_{a,b}^\pi=\sum_{q\in V}C_{a,b}^\pi(0,0,q).
 \end{gathered}
 \label{eq:flp-resolved-totals}
\end{equation}
The ordering $\pi$ is $t$-flag if and only if 
\begin{equation}
 \begin{gathered}
 C_{a,b}^\pi(s,0,q)=0\qquad(s\in\mathbb F_2),\\
 \text{whenever }1\leq a+b\leq t,\quad
 \operatorname{wt}_{P_Z}(q)>a+b.
 \end{gathered}
 \label{eq:flp-flag-zeros}
\end{equation}
For an integer $D\geq1$, its operation distance satisfies
$d_\pi\geq D$ if and only if
\begin{equation}
 U_{a,b}^\pi=T_{a,b}^\pi
 \qquad\bigl((a,b)\in\mathcal J,\ 0<a+b<D\bigr).
 \label{eq:flp-distance}
\end{equation}
\end{proposition}

\begin{proof}
All coefficients are nonnegative because they count fault patterns.
The zero fault is the only pattern with profile $(0,0)$.
Grouping all faults with profile $(a,b)$ according to their
values of $y$ gives
\[
 C_{0,0}^\pi(y)=\mathbf1_{\{y=0\}},
 \qquad
 \sum_y C_{a,b}^\pi(y)=\Gamma_{a,b}.
\]

Fix $y=(s,\alpha,q)$ and choose the faults at times
$t=1,\ldots,g$ arbitrarily.
To obtain this prescribed $y$, the final fault at $t=g+1$
must have specific values of its $x_A,z_F,x_Q,z_Q$ bits.
Its remaining $z_A,x_F$ bits are unrestricted.
There are therefore
\[
 2^{2gn_*}\,2^{f+1}=2^{2N-\beta}
\]
fault patterns giving each $y$.
Summing over their profiles proves the third equality in
\Cref{eq:flp-resolved-totals}.
By \Cref{prop:fw-signatures}, the faults counted by
$U_{a,b}^\pi$ have $s=\alpha=0$ and $q\in V^{\perp_s}$,
while those counted by $T_{a,b}^\pi$ additionally have
$q\in V$. This proves the remaining two equalities. The counts in \Cref{eq:flp-flag-zeros} vanish exactly when
no unflagged fault of operation weight between $1$ and $t$
violates \Cref{def:fw-tflag}.
Finally, $U_{a,b}^\pi-T_{a,b}^\pi$ counts the faults with
operation-wise weight profile $(a,b)$ that have zero syndrome and nontrivial
logical effect. Thus $d_\pi\geq D$ exactly when this
difference vanishes for every $0<a+b<D$, proving
\Cref{eq:flp-distance}.
\end{proof}

Let $A_j,B_j$ be the prescribed weight distributions of
$V,V^{\perp_s}$, respectively, as in
\Cref{prop:lp-boundary-enumerators}.
Because the weight profile does not distinguish data-boundary
locations from other operation locations, we use the following
total-weight constraint.

\begin{proposition}[Boundary constraint]
\label{prop:flp-code-boundaries}
Every ordering $\pi\in\mathfrak P$ satisfies
\begin{equation}
 \sum_{\substack{(a,b)\in\mathcal J\\a+b=j}}
       (U_{a,b}^\pi-T_{a,b}^\pi)
 \geq2(B_j-A_j),\qquad 1\leq j\leq n.
 \label{eq:flp-boundary-harmful}
\end{equation}
For $k_L>0$, it follows that
$d_\pi\leq d_{\mathrm{code}}$ and
$2(d_\pi-1)\leq n-k_L$.
\end{proposition}
\begin{proof}
Each weight-$j$ label $q\in V^{\perp_s}\setminus V$ gives
distinct input-only and output-only faults, both with zero
records since $\langle p_Z,q\rangle_s=0$ and with nontrivial
logical effect.
Their data labels remain $q$ as input data faults propagate only
an additional syndrome-ancilla $X$ label
$\langle p_Z,q\rangle_s=0$.
Both faults have operation weight $j$, since each operation block
contains at most one data-qubit location.
Counting gives \Cref{eq:flp-boundary-harmful}; the first
$j$ with $B_j>A_j$ gives the distance bound, and
the code Singleton bound, as in \Cref{cor:boundary-singleton},
gives the last inequality.
\end{proof}

\begin{definition}[Flag-circuit distance LP]
\label{def:flp-final}
For integers $t,D\geq1$, let
$\mathrm{FlagLP}_{\mathfrak P}(t,D)$ be the feasibility LP
in real arrays $U,T,\widehat U,\widehat T,C$, indexed by
$(a,b)\in\mathcal J$ and, for $C$, by $y$.
Require
\Cref{eq:lp-capacity,eq:lp-zero,eq:lp-cardinalities,eq:we-split-macwilliams}
with the substitutions in \Cref{eq:flp-krawtchouk-16},
together with $C\geq0$ and
\Cref{eq:flp-resolved-totals,eq:flp-flag-zeros,eq:flp-distance,eq:flp-boundary-harmful}.
In these constraints, replace
$U^\pi,T^\pi,\widehat U^\pi,\widehat T^\pi,C^\pi$
by $U,T,\widehat U,\widehat T,C$, respectively.
For $k_L>0$, also require $2(D-1)\leq n-k_L$. Additional linear constraints $\mathcal A_{\mathfrak P}$
may be imposed provided they hold for
$(U^\pi,T^\pi,\widehat U^\pi,\widehat T^\pi,C^\pi)$
for every $\pi\in\mathfrak P$.
All data and additional constraints remain fixed as $D$ varies.
\end{definition}
This specializes \Cref{def:lp-spacetime} to one class with
common parameters, operation weights, and the additional
resolved counts. Without $\mathcal A_{\mathfrak P}$, the LP
uses only the shared code and circuit parameters.

\begin{theorem}[LP distance bound]
\label{thm:flp-bound}
Every $t$-flag ordering $\pi\in\mathfrak P$ with $d_\pi\geq D$
supplies a feasible point of $\mathrm{FlagLP}_{\mathfrak P}(t,D)$.
Consequently, infeasibility implies $d_\pi\leq D-1$ for
every $t$-flag ordering in the prescribed family.
If $k_L>0$ and this subfamily is nonempty, then, for every
$t$-flag ordering $\pi\in\mathfrak P$,
\begin{equation}
 \begin{gathered}
 d_{\mathrm{LP}}\triangleq
 \min_{D\in\mathbb Z_{\geq1}}\{D-1:\
       \mathrm{FlagLP}_{\mathfrak P}(t,D)\text{ is infeasible}\},\\
 d_\pi\leq d_{\mathrm{LP}}\leq d_{\mathrm{code}}.
 \end{gathered}
 \label{eq:flp-distance-bound}
\end{equation}
Feasibility alone does not imply realization by an ordering.
\end{theorem}
\begin{proof}
Assign the LP variables by
\[
 (U,T,\widehat U,\widehat T,C)
 =(U^\pi,T^\pi,\widehat U^\pi,\widehat T^\pi,C^\pi).
\]
These actual counts from an actual circuit satisfy all constraints by
\Cref{prop:flag-forbidden-sectors,prop:flp-code-boundaries}
and the general enumerator identities, together with the
assumed validity of $\mathcal A_{\mathfrak P}$.
They therefore give a feasible real point.
The contrapositive proves the obstruction;
\Cref{eq:flp-boundary-harmful} excludes
$D=d_{\mathrm{code}}+1$, ensuring the stated minimum exists.
\end{proof}

\subsection{CSS Hamming-code analytical benchmark}
\label{subsec:flp-css-hamming}

\begin{definition}[CSS Hamming code encoder]
\label{def:css-hamming}
Fix an integer $\mu\geq3$ and put
$n=2^\mu-1$, $k_L=n-2\mu$, and $w=2^{\mu-1}$.
Let $\mathsf H\in\mathbb F_2^{\mu\times n}$ contain every
nonzero column $c\in\mathbb F_2^\mu$ exactly once.
The binary Hamming code with parity-check matrix $\mathsf H$
is $\{x\in\mathbb F_2^n:\mathsf Hx^{\mathsf T}=0\}$. Let $Q_{\mathrm L}\cong\mathbb C^{2^{k_L}}$ and
$Q=\bigotimes_{c\ne0}Q_c$, with $Q_c\cong\mathbb C^2$
ordered by the columns of $\mathsf H$.
An encoder of the CSS Hamming code
\cite{ChaoReichardt2018TwoQubits} is an isometric stabilizer
circuit satisfying, for every $i\in\mathbb Z^+_\mu$,
\begin{equation}
 \begin{gathered}
 C_{\mathrm{enc}}:Q_{\mathrm L}\longrightarrow Q,
 \qquad C_{\mathrm{enc}}^\dagger C_{\mathrm{enc}}=I_{Q_{\mathrm L}},\\
 \left(\prod_{c:c_i=1}X_{Q_c}\right)C_{\mathrm{enc}}
 =C_{\mathrm{enc}},\\
 \left(\prod_{c:c_i=1}Z_{Q_c}\right)C_{\mathrm{enc}}
 =C_{\mathrm{enc}}.
 \end{gathered}
 \label{eq:ham-encoder-checks}
\end{equation}
The code is $\operatorname{im}(C_{\mathrm{enc}})$.
Each generating check has weight $w$, and, in the notation of
\Cref{Cor: Logical dim Pres}, its output stabilizer space is
\begin{equation}
 V=V_{\mathrm{enc}}^{\mathrm{out}}
 =\{[u\mathsf H,v\mathsf H]:u,v\in\mathbb F_2^\mu\}.
 \label{eq:ham-code-data}
\end{equation}
\end{definition}

\begin{proposition}[Analytic boundary obstruction]
\label{lem:ham-boundary-enumerators}
\label{prop:ham-boundary-bound}
The code has parameters $[[n,k_L,3]]$, and, for $j\in[n]$,
\begin{equation}
 \begin{aligned}
 A_j&=\delta_{j,0}+3n\delta_{j,w}+n(n-1)\delta_{j,3w/2},\\
 B_j&=2^{-2\mu}\bigl[
 \mathbb K_j^{(4)}(0;n)+3n\mathbb K_j^{(4)}(w;n)+n(n-1)\mathbb K_j^{(4)}(3w/2;n)\bigr].
 \end{aligned}
 \label{eq:ham-boundary-enumerators}
\end{equation}
In particular, $B_1=B_2=A_3=0$ and $B_3=n(n-1)/2$.
For every flag count and ordering family,
$\mathrm{FlagLP}_{\mathfrak P}(t,D)$ is infeasible for $D\geq4$.
\end{proposition}
\begin{proof}
Rows of $\mathsf H$ have weight $w$ and distinct rows overlap
on $w/2$ columns, both even. Hence
$\mathsf H\mathsf H^{\mathsf T}=0$; its rank $\mu$ gives $k_L$.
For $u,v\in\mathbb F_2^\mu$, if the $2\times\mu$ matrix
$[u;v]$ has rank $r$, then
$[u\mathsf H,v\mathsf H]$ has weight $2^\mu-2^{\mu-r}$.
Ranks $0,1,2$ occur $1,3n,n(n-1)$ times, respectively,
giving $A_j$; \Cref{eq:lp-output-macwilliams} gives $B_j$.
The distinct nonzero columns exclude nonzero normalizer
labels of weight one or two. Each triple $\{c,d,c+d\}$
supports exactly the weight-three labels $XXX,YYY,ZZZ$;
there are $n(n-1)/6$ such triples and $w\geq4$, proving
$B_3=n(n-1)/2$, $A_3=0$, and distance three.
Finally, at any feasible LP point, the boundary constraint
\Cref{eq:flp-boundary-harmful} requires
\begin{equation}
 \sum_{\substack{(a,b)\in\mathcal J\\a+b=3}}
       (U_{a,b}-T_{a,b})
 \geq n(n-1)>0,
 \label{eq:ham-boundary-certificate}
\end{equation}
whereas the distance constraints \Cref{eq:flp-distance}
make the left-hand side zero
for $D\geq4$.
\end{proof}

\begin{definition}[Chao--Reichardt one-flag construction]
\label{def:ham-chao-reichardt}
For the CSS Hamming-code encoder of \Cref{def:css-hamming},
take $f=1$. Fix a binary basis of $\mathbb F_{2^{\mu-1}}$
and write the column labels as
$c\in\mathbb F_2\times\mathbb F_{2^{\mu-1}}$.
Choose $\xi\in\mathbb F_{2^{\mu-1}}$ of multiplicative order
$w-1$ and enumerate the columns of $\mathsf H$ with first coordinate one by
\begin{equation}
 c_1=(1,0),\qquad
 c_{j+2}=(1,\xi^j)\quad(0\leq j\leq w-2).
 \label{eq:ham-finite-field-order}
\end{equation}
Relabel $Q_{c_i}$ as $Q_i$ for $i\in\mathbb Z^+_w$ and set
$P_Z=\prod_{i=1}^{w}Z_{Q_i}$.
The Chao--Reichardt construction
\cite[arXiv version, Sec.~III~C, Claim~1]{ChaoReichardt2018TwoQubits}
is the circuit $\mathbf C_{\pi_*}$ of \Cref{def:fw-family}
with ordering
\begin{equation}
 \bigl(\pi_*(1),\ldots,\pi_*(w+2)\bigr)
 =\bigl(1,w+1,2,\ldots,w-1,w+2,w\bigr).
 \label{eq:ham-bracketed-order}
\end{equation}
The two flag CNOTs occur immediately after the first and
immediately before the last data CNOT, respectively;
see \Cref{fig:flag-family-grouping-hamming} for $\mu=3$.
\end{definition}
For the benchmark LP, let $\mathfrak P_{\mathrm{Ham}}$
be the set of all permutations of $\mathbb Z^+_{w+2}$ and take
$\mathcal A_{\mathfrak P_{\mathrm{Ham}}}=\varnothing$.
The common parameters are
\begin{equation}
 \begin{gathered}
 g=w+2,\qquad N_1=n(w+3)+2,\qquad N_2=w+2,\\
 h=2\mu+2,\qquad m=2n-2\mu+2.
 \end{gathered}
 \label{eq:ham-lp-parameters}
\end{equation}

\begin{theorem}[Attainment of the CSS Hamming LP bound]
\label{thm:ham-tight-benchmark}
Under \Cref{def:flag-pauli-model,def:flag-operation-weight},
$\pi_*$ is $1$-flag, its distributions give a feasible point of
$\mathrm{FlagLP}_{\mathfrak P_{\mathrm{Ham}}}(1,3)$, and
\begin{equation}
 d_{\pi_*}
 =\max_{\substack{\pi\in\mathfrak P_{\mathrm{Ham}}\\
                  \pi\text{ is }1\text{-flag}}}d_\pi
 =d_{\mathrm{LP}}=3.
 \label{eq:ham-tight-distance}
\end{equation}
\end{theorem}
\begin{proof}
For single CNOT faults, Chao and Reichardt's analysis
\cite[arXiv version, Sec.~III~C, Claim~1]{ChaoReichardt2018TwoQubits}
establishes that unflagged data errors have weight at most
one modulo $P_Z$. For flagged faults, the $Z$ components
are determined modulo $P_Z$ by the $X$-check syndrome.
The $X$ component has weight at most one, so the distinct
nonzero columns of $\mathsf H$ determine it from the
$Z$-check syndrome. Thus the full output-code syndrome
determines every flagged data error modulo stabilizers.
The additional singleton locations preserve these properties:
data faults, including both boundaries, remain unflagged
single-qubit data errors; flag faults give no data error,
and $I_Q$ is already among the flagged CNOT possibilities;
and the sole singleton $A$-location, at $t=1$, gives
$I_Q$ or $P_Z$ without a flag.
Thus $\pi_*$ is $1$-flag.

Write any undetected fault of operation weight at most two
as $e=e_1+e_2$, by splitting its occupied operation blocks, with each $e_i$
of operation weight at most one and allowing a zero summand.
Linearity and $H_{\pi_*}e^{\mathsf T}=0$ imply that their
$H_{\pi_*}$ signatures agree, so their flag outcomes
and output-code syndromes coincide.
If the common flag is zero, reducing both data errors modulo
$P_Z$ gives a zero-syndrome error of weight at most two,
which is a stabilizer because $d_{\mathrm{code}}=3$.
If the common flag is one, the cited distinguishability
property gives the same conclusion.
Hence $q_{\pi_*}(e)\in V$ and $d_{\pi_*}\geq3$.
The boundary bound in \Cref{prop:flp-code-boundaries} gives
$d_{\pi_*}\leq3$.
Finally, \Cref{thm:flp-bound,prop:ham-boundary-bound}
give LP feasibility at $D=3$ and infeasibility at $D=4$,
proving \Cref{eq:ham-tight-distance}.
\end{proof}
\bibliography{bibliography}

\end{document}